\documentclass[12pt]{article}

\usepackage{natbib}            % \citet, \citep, etc.
\usepackage{breakcites}        % allow line breaks in long citations
\usepackage[T1]{fontenc} 
\usepackage{lmodern, cmap}

\usepackage{mathtools}         % loads amsmath
\usepackage{amssymb,amsfonts,amsthm}
\allowdisplaybreaks

\usepackage[bb=libus]{mathalpha}

\usepackage{graphicx}
\usepackage{epstopdf}
\DeclareGraphicsExtensions{.pdf,.png,.jpg,.jpeg}

\usepackage{float}
\usepackage{caption}
\usepackage[caption=false]{subfig} % keep if you use \subfloat

\usepackage{tikz}
\usetikzlibrary{arrows,petri,topaths}
\usetikzlibrary{patterns}

\usepackage{pgfplots}
\pgfplotsset{compat=1.18}
\usepgfplotslibrary{fillbetween}

\usepackage{tkz-graph}
\usepackage{tkz-berge}

\usepackage{setspace}
\usepackage{indentfirst}

\usepackage{versionPO}
\excludeversion{notes}   % Include notes?
\includeversion{links}   % Turn hyperlinks on?

\ifnotes{%
  \usepackage[paperwidth=10in,top=1in,bottom=1in,left=1in,right=2.5in]{geometry}%
  \usepackage[textwidth=1.4in,shadow,colorinlistoftodos]{todonotes}%
}{%
  \usepackage[paperwidth=9.75in,top=1in,bottom=1in,left=1.25in,right=1.25in]{geometry}%
}

\usepackage{enumitem}
\usepackage{rotating}
\usepackage{fancyvrb}
\usepackage{pdfpages}
\usepackage{datetime}
\newdateformat{mydate}{\monthname[\THEMONTH] \THEYEAR}

\usepackage{marginnote}
\usepackage{url}
\usepackage{accents}
\usepackage{xparse}
\usepackage{letltxmacro}
\usepackage{nicefrac}
\usepackage{faktor}
\usepackage{wasysym}
\usepackage{etoolbox}   % needed for \patchcmd
\usepackage{chngcntr}
\usepackage{apptools}

\usepackage[hidelinks]{hyperref}
\iflinks{}{\hypersetup{draft=true}}% <-- no blank lines inside \iflinks args

\let\oldmarginpar\marginpar
\ifnotes{%
  \makeatletter\let\chapter\@undefined\makeatother%
  \newcommand{\smalltodo}[2][]{\todo[caption={#2}, size=\scriptsize, fancyline, #1]{\begin{spacing}{.5}#2\end{spacing}}}%
  \newcommand{\rhs}[2][]{\smalltodo[color=green!30,#1]{{\bf RS:} #2}}%
  \newcommand{\rhsnolist}[2][]{\smalltodo[nolist,color=green!30,#1]{{\bf RS:} #2}}%
  \newcommand{\rhsfn}[2][]{\renewcommand{\marginpar}{\marginnote}\smalltodo[color=green!30,#1]{{\bf RS:} #2}\renewcommand{\marginpar}{\oldmarginpar}}%
}{%
  \newcommand{\smalltodo}[2][]{}%
  \newcommand{\rhs}[2][]{}%
  \newcommand{\rhsnolist}[2][]{}%
  \newcommand{\rhsfn}[2][]{}%
}

\newcounter{footequation}
\renewcommand{\thefootequation}{F.\arabic{footequation}}

\makeatletter
\LetLtxMacro\latex@@footnote\footnote
\RenewDocumentCommand{\footnote}{om}{%
  \begingroup
  \let\c@equation\c@footequation
  \let\theequation\thefootequation
  \IfValueTF{#1}{\latex@@footnote[#1]{#2}}{\latex@@footnote{#2}}%
  \endgroup
}
\makeatother

\newtheorem{theorem}{Theorem}[section]

\newtheorem{corollary}[theorem]{Corollary}

\newtheorem{definition}{Definition}
\newtheorem{example}[theorem]{Example}

\newtheorem{lemma}{Lemma}[section]

\newtheorem{proposition}[theorem]{Proposition}
\newtheorem{remark}[theorem]{Remark}

\newtheorem{assumption}{Assumption}

\AtAppendix{\counterwithin{lemma}{section}}
\AtAppendix{\counterwithin{assumption}{section}}

\makeatletter
\patchcmd{\env@cases}{1.2}{0.72}{}{}  % adjust cases spacing
\makeatother

\newcommand{\R}{\mathbb{R}}

\newcommand{\E}{\mathbb{E}}

\DeclareMathOperator*{\argmax}{argmax}

\DeclareMathOperator*{\linspan}{span}

\DeclareMathOperator{\ran}{ran}
\DeclareMathOperator{\Var}{Var}
\DeclareMathOperator{\Cov}{Cov}

\newcommand{\abs}[1]{\left|#1\right|}
\newcommand{\norm}[1]{\left\lVert#1\right\rVert}

\mathchardef\mhyphen="2D

\makeatletter
\def\overUnderArrow{\@ifnextchar[\overUnderArrow@i{\overUnderArrow@i[]}}
\def\overUnderArrow@i[#1]#2#3{%
\ifx\relax#1\relax\array[b]{c}\overset{\text{#2}}{\uparrow}\\#3\endarray
\else\ifx\relax#2\relax
\array[t]{c}#3\\\underset{\text{#1}}{\downarrow}\endarray
\else
\array{c}\overset{\text{#2}}{\uparrow}\\#3\\\underset{\text{#1}}{\downarrow}\endarray
\fi\fi}
\makeatother

\setlist[description]{topsep=5pt,itemsep=5pt,font={\slshape\normalfont}}

\begin{document}

\setlist{noitemsep} % Reduce space between list items (itemize, enumerate, etc.)
\onehalfspacing % Use 1.5 spacing
% Use endnotes instead of footnotes - redefine \footnote command

%\doublespacing % Double space

\title{An Infinite-Dimensional Insider Trading Game} 

%\iffalse

\author{Christian Keller \and Michael C. Tseng\thanks{Christian Keller is from the Department of Mathematics, University of Central Florida.
Michael Tseng is from the Department of Economics, University of Central Florida.
The research of Christian Keller was supported in part by NSF grant DMS-2106077.
Corresponding author: Michael Tseng, E-mail: michael.tseng@ucf.edu. Address: 
Department of Economics,
College of Business,
University of Central Florida,
12744 Pegasus Dr., Orlando, Florida 32816-1400.
}
}

% Create title page with no page number

%\fi

\date{
%{\footnotesize {\sl Latest Version Available \Googlesite}} 
\mydate\today
}

\renewcommand{\thefootnote}{\fnsymbol{footnote}}

\singlespacing

\maketitle

\vspace{0.5cm}

\rule{\linewidth}{1pt} %\noindent \rule{6.5in}{1pt}

%\vspace{-.2in}
{\flushleft \textbf{Abstract}}

We generalize the seminal framework of \cite{kyle1985continuous} to a many-asset setting, bridging the gap between informed-trading theory and modern trading practices.
Specifically, we formulate an infinite-dimensional Bayesian trading game in which the informed trader's private information may concern arbitrary aspects of the cross-sectional payoff structure across a continuum of traded assets.
In this general setting, we obtain a parsimonious equilibrium characterized by a single scalar fixed point, which yields closed-form characterizations of the equilibrium trading strategy, price impact within and across markets, and the information efficiency of equilibrium prices.

\vspace{0.5cm}

\noindent 

\medskip
\vspace{0.5cm}
%\noindent \textit{JEL Classification}: C61; G13; G14.


\thispagestyle{empty}

\clearpage

\onehalfspacing
\setcounter{footnote}{0}
\renewcommand{\thefootnote}{\arabic{footnote}}
\setcounter{page}{1}

\section{Introduction}

\cite{kyle1985continuous} introduced the canonical model of strategic informed trading and remains a benchmark in market microstructure theory.
A large subsequent literature has refined this paradigm along many mathematically and economically important dimensions---including risk aversion and inventory considerations (e.g., \cite{cho2003continuous,ccetin2016markovian,ekren2025stochastic,back2020optimaltransport}), 
alternative specifications of liquidity demand and noise trading (\cite{biagini2012insider}), departures from Gaussian fundamentals and linear pricing (e.g., \cite{cho2000insider}), dependence between liquidity demand and fundamentals (\cite{aase2012partially}), as well as
features such as limit orders and stochastic liquidity demand (e.g., \cite{rochet1994insider,collin2016insider}).
Despite this breadth, a common restriction persists across much of the theory: private information is modeled as low-dimensional, and the informed trader trades a single security (or, at most, a small finite collection under strong parametric assumptions), 
even in multi-insider and multi-asset environments (e.g., \cite{caballe1994imperfect,foster1996strategic,back2000imperfect}).
By contrast, modern trading routinely spans a large universe of correlated assets and derivatives, where payoff heterogeneity is intrinsically high-dimensional---indeed, effectively infinite-dimensional in the presence of option surfaces---and cross-asset inference is central.

In this paper, we address this gap between theory and practice by developing a general many-asset Kyle game with essentially \emph{arbitrary} payoff heterogeneity across assets and signals (subject only to technical conditions).
This generality makes our model applicable across a wide range of trading environments and asset classes.
Despite the model's infinite-dimensional primitives, we show that it admits a parsimonious equilibrium pinned down by a single scalar fixed point.
Within this equilibrium, we characterize the informed trader's strategy, price impact both within and across markets, and the information efficiency of equilibrium prices.

For concreteness, we interpret the traded assets primarily as state-contingent claims on future states---Arrow--Debreu securities in the sense of \cite{arrow1954existence}. 
This interpretation is purely expositional: our analysis applies equally under any alternative reading of the model’s assets. 
Viewed through the Arrow--Debreu lens, our contingent claims subsume, in the usual economic sense, the familiar contingent-claim payoffs of standard derivative and asset-pricing frameworks (e.g., \cite{black1973pricing}).
In particular, the classical \cite{breeden1978prices} formula yields an option-market specialization, in which Arrow--Debreu securities are recast as a complete menu of European options.

In mathematical finance terminology, the Arrow--Debreu state-price density is the (discounted) risk-neutral density (see, e.g., \cite{duffie2001dynamic}).
The \cite{breeden1978prices} formula recovers this density from option prices.
Actual option strike grids may therefore be regarded as practical discretizations of this stylized continuum.
Indeed, widely traded synthetic derivatives such as the Cboe VIX, Cboe SKEW, and OTC variance swaps are explicitly constructed via such discretizations 
(see \url{https://www.cboe.com/tradable_products/vix/}). 
Accordingly, the Arrow--Debreu reading aligns the model with the options---more generally, derivatives---viewpoint.

Under this option-market translation, our results consolidate and sharpen several themes from the empirical options literature and yield new data-facing implications.
Moreover, the equilibrium informed trading strategy recovers option-trading strategies that are widely used in practice but lie outside the scope of existing informed-trading models.
We discuss these finance implications---ranging from cross-strike price discovery to mechanisms generating the volatility smile---in the companion paper \cite{kellertseng}, which foregrounds the economic results in a discrete formulation designed to speak directly to the large empirical options literature.
In this paper, we undertake the task of placing this general trading environment---formulated in the mathematically natural continuum setting---and its equilibrium outcome on a rigorous footing.

At the level of generality we consider, the trading game is inherently infinite-dimensional.
The informed trader faces an infinite-dimensional portfolio choice problem: he chooses an unrestricted portfolio of traded assets---mathematically, an arbitrary (H\"older) function on the state space.
The market maker, in turn, faces an infinite-dimensional Bayesian inference problem: given his conjecture about the informed trader's strategy, he forms a posterior over the trader's private signal from the realized order-flow sample path across assets.
The resulting posterior mean induces his (zero-profit, under the conjecture) pricing kernel, which maps each realized order-flow sample path to the associated schedule of prices across assets---a generally nonlinear mapping between the relevant function spaces.

Equilibrium is therefore \emph{a priori} an infinite-dimensional fixed point that couples these two problems through a mutual consistency requirement.
Such a fixed-point problem is not amenable to the standard methods used in finite-dimensional Kyle-type settings (and, more generally, in finite-dimensional Bayesian games). A tractable equilibrium analysis therefore requires new techniques.

Our analysis addresses several attendant issues that arise in this infinite-dimensional Bayesian game. For orientation, we touch upon two here.
The first is unavoidable at the outset and must be resolved before the remaining analysis can proceed cleanly.
The market maker's Bayesian inference requires---under his conjecture about the insider's strategy---an exponential likelihood involving the stochastic integral of the conjectured drift against the realized order-flow sample path.
The classical It\^{o} integral, however, is not defined pathwise, and therefore cannot be used to define the likelihood (and hence the posterior).
We address this by using a pathwise stochastic integral from rough paths theory (see, e.g., \cite{lyons2007differential}, \cite{Hairer2014}, and \cite{friz2020course}), which yields a well-defined likelihood for each realized 
order-flow path.
This pathwise formulation, in turn, allows us to identify a (function-valued) sufficient statistic for the market maker's inference.

The second issue is conceptually central. Although the primitives of the game are infinite-dimensional, we show that its equilibrium problem admits a sharp reduction. Using the (noise-adjusted) asset-payoff covariance kernel and the 
induced reproducing kernel Hilbert space, we ``whiten'' the sufficient statistic above and pass to an isomorphic \emph{canonical game} that is invariant to the payoff specification across assets and signals, as well as to cross-asset heterogeneity in noise-trading intensities. 
In this canonical game, payoffs reduce to evaluation at the realized signal. This canonical reduction yields a particularly tractable equilibrium, pinned down by a scalar fixed point.

The rest of the paper is organized as follows.
Section~\ref{sec: model outline} presents the model.
Section~\ref{sec: Market Maker's Pricing Kernel} develops the market maker's inference problem and the induced pricing kernel.
Section~\ref{sec: Informed Trader's Problem} analyzes the informed trader's portfolio choice.
Section~\ref{sec: finite dim reduction} carries out the reduction to an isomorphic canonical game.
Section~\ref{sec: symmetric equilibrium} constructs equilibrium.
Section~\ref{sec: discussion} discusses price discovery implications for contingent claims, including informed demand, price impact, and the information efficiency of prices.
Section~\ref{sec: conclusion} concludes.

\section{Model}
\label{sec: model outline}

For clarity of exposition while maintaining full precision, we present the model in two steps.
In Section~\ref{sec: trading game}, we first describe the economic structure of the trading game: the players, their information, actions, and timing.
Section~\ref{sec: assumptions} then states the precise mathematical assumptions on the economic primitives---the underlying probabilistic setup and the regularity assumptions on the payoffs, order flow, and strategies.
Section~\ref{SS:AnOptionsFormulation} describes the options specialization of the model.

\subsection{The Trading Game}
\label{sec: trading game}

There are two risk-neutral agents---the insider and the market maker.
There is a risk-free asset in perfectly elastic supply at a risk-free rate of zero, i.e., agents can borrow and lend without restrictions at this normalized rate---a standard frictionless benchmark; see \cite[Ch.~6]{bodie2021investments}.

At $t=0$, the insider observes a signal $s$ that informs him of the probability distribution over the possible $t=1$ states of the world.
The set of $t=1$ states is indexed by a compact interval $[\underaccent{\bar}{x}, \bar{x}]$ (compactness is imposed only for expositional convenience; our results extend to unbounded index sets under appropriate integrability conditions). 
The market maker has a Bayesian prior over the possible signals. At $t=0$, there is a complete market of Arrow--Debreu (AD) securities for $t=1$ states (extensions to incomplete markets are immediate, for instance by restricting trade to claims that are measurable 
with respect to a given partition of the state space; we focus on the complete-market case here).

After observing his private signal at $t=0$, the insider submits his demand for AD securities to maximize his expected utility at $t=1$.
The market maker receives the combined order flow of the insider and noise traders across AD markets and executes the orders at his zero-profit prices.

The possible signals lie in a Borel probability space $(S, \pi_0)$, where the probability measure $\pi_0(ds)$ is the market maker's prior.       
Conditional on a signal $s \in S$, the probability distribution over the $t=1$ states is specified by a density $\eta(\,\cdot\, ,s ) \colon [\underaccent{\bar}{x}, \bar{x}] \rightarrow \mathbb{R}$.
After observing $s$, the insider 
chooses a portfolio $W(\,\cdot\, , s) \colon [\underaccent{\bar}{x}, \bar{x}] \rightarrow \mathbb{R}$, where $W(x , s)$ is the insider's order
for the state $x$ security.

Noise trader orders are assumed to be normally distributed with mean zero in each market and uncorrelated across markets. 
Over the continuum, this means that the noise trades across an (infinitesimal) increment of states $[x, x \! + \! dx]$ follow $\sigma(x)dB_x$, where $(B_x)$ is a standard Brownian motion over state $x$, and $\sigma(x)$ is the noise trading intensity at $x$. 
This is analogous to assuming the noise trades follow $dB_t$ across the time interval $ [t, t \! + \! dt]$ in a dynamic setting; see \cite{kyle1985continuous} and \cite{back1993asymmetric}.

The cumulative combined order flow received by the market maker is then a sample path $\omega$ of the stochastic process $(Y_x)$ over state $x$ specified by
\begin{equation}
\label{eqn: dY_x}
\underbrace{dY_x}_{\substack{\text{combined order}}} = \;\; \underbrace{W(x,s)dx}_{\substack{\text{insider order}}} \;\; + \; \underbrace{\sigma(x) dB_x}_{\substack{\text{noise order}}}\!\!.
\end{equation}
For each $x$, $Y_x =\int_{\underaccent{\bar}{x}}^x dY_{x'}$ is the cumulative combined order flow over the AD markets $[\underaccent{\bar}{x}, x]$. 
Equivalently, the market maker receives the combined orders for each market $x \in [\underaccent{\bar}{x}, \bar{x}]$.

The market maker has a belief $\widetilde{W}(\,\cdot\, , \,\cdot\, ) \colon [\underaccent{\bar}{x}, \bar{x}] \times S \rightarrow \mathbb{R}$ regarding the insider's trading strategy.
After receiving order flow $\omega$, the market maker updates his prior $\pi_0(ds)$ based on his belief $\widetilde{W}$ to the posterior $\pi_1(ds, \omega \kern 0.045em ; \widetilde{W})$ regarding 
the insider's signal. 
His zero-profit prices for the AD securities are their expected payoffs conditional on $\omega$ per his belief,
\[
\underbrace{P(x, \omega \kern 0.045em ; \widetilde{W})}_\text{security $x$ price} = \int_S \! \eta(x, s) \pi_1(ds, \omega \kern 0.045em ; \widetilde{W}), \,\, x \in [\underaccent{\bar}{x}, \bar{x}].
\]
In other words, these are the market maker's competitive break-even prices under his belief.

Conditional on observing $s$ and given market maker belief $\widetilde{W}(\,\cdot\, , \,\cdot\,)$, the insider's AD portfolio choice problem is
\begin{equation}
\label{eqn: informed trader's problem informal}
\max_{ W(\,\cdot\,)} \mathbb{E}^{\mathbb{P}_{\scriptscriptstyle W}} [ \int_{\underaccent{\bar}{x}}^{\bar{x}} \!\! ( \eta(x,s) -  P(x, \omega \kern 0.045em ; \widetilde{W}) ) \cdot W(x) dx ] \equiv  \max\limits_{ {\scriptscriptstyle W(\cdot) } } J(W; \widetilde{W}, s) 
\end{equation}
where the expectation $\mathbb{E}^{\mathbb{P}_{\scriptscriptstyle W}}[\,\cdot\,]$ is taken over order flow $\omega$  
under its probability law $\mathbb{P}_{\!\scriptscriptstyle W}$ of \eqref{eqn: dY_x} induced by insider portfolio choice $W(\,\cdot\,)$.
The functional $J(\, \cdot \, ; \widetilde{W}, s)$ defined in \eqref{eqn: informed trader's problem informal} is the insider's expected utility functional conditional on $s$ and given market maker belief $\widetilde{W}$.

\begin{remark}
\label{rmk: AD interpretation not nec}
As noted in the introduction, the model is, at its core, a general many-asset framework for strategic informed trading, allowing essentially arbitrary specification $\eta(\,\cdot\,, \,\cdot\,)$ of payoff heterogeneity across assets and signals, together 
with heterogeneous noise-trading intensities $\sigma(\,\cdot\,)$, subject only to the technical assumptions in Section~\ref{sec: assumptions}. 
Although we focus on the empirically salient Arrow--Debreu interpretation, our results do not depend on this interpretation. 
They apply under any alternative view of the assets in the model.
\end{remark}

\subsection{Mathematical Assumptions}
\label{sec: assumptions}

Here we state the standing mathematical assumptions for the model.
The possible signals lie in a Borel probability space $(S,\pi_0)$, where $\pi_0(ds)$ is the market maker's prior.
The possible order-flow realizations are elements of the measurable space $(\Omega,\mathcal{F})$, where
\[
\Omega \coloneqq C([\underaccent{\bar}{x},\bar{x}],\R),
\qquad
\mathcal{F} \text{ is the Borel $\sigma$-field induced by } \|\cdot\|_\infty.
\]
We write $\omega_x$ for the canonical coordinate process on $\Omega$.
Let $\mathbb{P}_{\!\scriptscriptstyle 0}$ denote the probability measure on $(\Omega,\mathcal{F})$ under which the canonical process
$x\mapsto \omega_x$ has the same law as the Gaussian process $(\sigma(x)B_x)_{x\in[\underaccent{\bar}{x},\bar{x}]}$
(cf.\ the order-flow specification~\eqref{eqn: dY_x}).

For $\alpha\in(0,1]$, we recall that the H\"older space $C^{\alpha}([\underaccent{\bar}{x},\bar{x}],\R)$ consists of those
$f\in C([\underaccent{\bar}{x},\bar{x}],\R)$ with finite H\"older seminorm
\[
[f]_\alpha \coloneqq \sup_{x\neq y}\frac{|f(x)-f(y)|}{|x-y|^\alpha},
\]
equipped with the usual norm $\|f\|_\alpha \coloneqq \|f\|_\infty + [f]_\alpha$.

Fix $\gamma\in(\tfrac13,\tfrac12)$ and $\delta\in(0,1]$ with $\delta+\gamma>1$.
Define
\[
\Omega_\gamma
\coloneqq
\bigl\{\omega\in\Omega:\ \omega(0)=0 \ \text{and}\ \omega\in C^\gamma([\underaccent{\bar}{x},\bar{x}],\R)\bigr\},
\qquad
\mathcal{F}_\gamma \coloneqq \{A\cap \Omega_\gamma:\ A\in\mathcal{F}\},
\]
and let $\mathbb{P}_{\!\scriptscriptstyle 0,\gamma}$ be the restriction of $\mathbb{P}_{\!\scriptscriptstyle 0}$ to
$(\Omega_\gamma,\mathcal{F}_\gamma)$.

\begin{assumption}
\label{assumption: MM's posterior}
$\,$

\nopagebreak
(i)
For every $s\in S$, the insider's AD portfolio $W(\,\cdot\,,s)\colon[\underaccent{\bar}{x},\bar{x}]\to\R$
lies in $C^{\delta}([\underaccent{\bar}{x},\bar{x}],\R)$.

(ii)
We restrict attention to order-flow realizations $\omega\in\Omega_\gamma$. Equivalently, we work on the
measurable space $(\Omega_\gamma,\mathcal{F}_\gamma)$.

(iii)
The map $s\mapsto W(\,\cdot\,,s)$ is continuous as a map $S\to C^{\delta}([\underaccent{\bar}{x},\bar{x}],\R)$, where the topology on $S$ is the one generating its Borel $\sigma$-field on $S$.

(iv)
The squared noise-trading intensity $\sigma^2(\,\cdot\,)$ across AD markets lies in
$C^{\delta}([\underaccent{\bar}{x},\bar{x}],\R)$ and satisfies $\sigma^2(x)>0$ for all
$x\in[\underaccent{\bar}{x},\bar{x}]$.
\end{assumption}

Assumption~\ref{assumption: MM's posterior}(i), together with $\delta+\gamma>1$, ensures that the integrand $W(\,\cdot\,,s)$ and the order-flow path $\omega$ have sufficient pathwise H\"older regularity for the 
integral $\int \cdot\, d\omega_x$ appearing in the likelihood and posterior (and hence in the pricing kernel) to be well-defined (see Theorem~\ref{thm: Market Maker's Posterior}).
Assumption~\ref{assumption: MM's posterior}(ii) restricts attention to $\gamma$-H\"older sample paths. This is without loss of generality under the reference measure $\mathbb{P}_{\!\scriptscriptstyle 0}$: since $\Omega_\gamma\in\mathcal{F}$ and Brownian sample paths 
are almost surely H\"older continuous of every exponent $<\tfrac12$, we have $\mathbb{P}_{\!\scriptscriptstyle 0}(\Omega_\gamma)=1$.
Assumption~\ref{assumption: MM's posterior}(iii) is a mild regularity condition on how the insider’s strategy depends on the signal, ensuring in particular that the relevant objects in the market maker’s updating vary measurably (indeed, continuously) in $s$.
Finally, Assumption~\ref{assumption: MM's posterior}(iv) imposes regularity and nondegeneracy of the noise variance across Arrow--Debreu markets---if noise were degenerate in some markets, then order flow there would be perfectly revealing and trading would collapse to a zero-sum game, 
leading to a degenerate equilibrium.

\paragraph{Equilibrium Definition}

In equilibrium, the optimal trading strategy of the insider, given the market maker's pricing kernel $P(\,\cdot\, , \,\cdot\, ; W^*)$ based on the latter's 
belief $W^*(\,\cdot\, , \,\cdot\,)$, coincides with $W^*(\,\cdot\, , \,\cdot\,)$.
In other words, conditional on observing each $s$, the insider's optimal portfolio is $W^*(\,\cdot\,,s)$, thereby confirming the market maker's belief, 
in the sense of the standard \emph{perfect Bayesian equilibrium} for incomplete-information games; see \cite[Ch.~28]{watson2013strategy}.

\begin{definition}
\label{def: equilibrium}
A \textbf{(perfect Bayesian) equilibrium} in our model is an admissible trading strategy
$W^* \colon [\underaccent{\bar}{x}, \bar{x}] \times S \to \mathbb{R}$
such that, for every signal $s \in S$,
\[
W^*(\cdot,s)\in \operatorname*{arg\,max}_{W(\cdot)}\, J\bigl(W;W^*,s\bigr),
\]
where $J(\,\cdot\,; W^*, s)$ is the insider's expected utility functional defined in \eqref{eqn: informed trader's problem informal}.

\end{definition}

%%%%%%%%%%%%%%%%%%%

\paragraph{Single-Asset Case (\cite{kyle1985continuous})}
If the traded-asset index set $[\underaccent{\bar}{x},\bar{x}]$ collapses to a singleton $\{x_0\}$, then the payoff specification across assets $\eta(\,\cdot\,,s)$ reduces to the scalar payoff
$v \coloneqq \eta(x_0,s)$
of the one traded asset. The market maker’s prior on $s$ therefore induces a prior distribution for $v$, and an admissible insider strategy is simply a measurable map $W\colon \mathbb{R}\to\mathbb{R}$, where $W(v)$ is the insider’s order when he observes that the asset payoff is $v$.
The market maker receives total order flow
$
\omega = W(v) + \varepsilon,
$
where $\varepsilon$ denotes noise order, and quotes the zero-profit price
\[
P(\omega; \widetilde{W})=\mathbb E\!\left[v \mid \omega \,;\, \widetilde{W}\right],
\]
given his belief $\widetilde{W}(\,\cdot\,)$ about the insider’s strategy.

In this case, Definition~\ref{def: equilibrium} reduces to the following: for each $v\in\mathbb R$,
\[
W^*(v)\in \arg\max_{w\in\mathbb R}\; \mathbb E\!\left[(v-P(w+\varepsilon; W^*))\,w \,\big|\, v \right],
\]
which is the familiar single-asset equilibrium condition of \cite[Def.~1]{kyle1985continuous}. %(in particular, see his Eq.~(2.2)).

Under the Gaussian assumption $v\stackrel{d}{\sim}\mathcal N(v_0,\sigma_v^2)$ and
$\varepsilon\stackrel{d}{\sim}\mathcal N(0,\sigma_\varepsilon^2)$, we have Kyle's linear equilibrium 
\begin{equation}
\label{eqn: Kyle equil}
W^*(v) = \beta\,(v-v_0),
\qquad
P(\omega) = v_0+\lambda\,\omega,
\end{equation}
where $\beta = \sigma_{\varepsilon}/\sigma_v$ and $\lambda = \sigma_v/(2\,\sigma_{\varepsilon})$ is the (linear) price-impact coefficient (often called \emph{Kyle's lambda} in the finance literature).

\subsection{An Options Specialization}
\label{SS:AnOptionsFormulation}

Our model admits an options specialization when the state is the terminal price $X$ of an underlying asset, taking values in $[\underaccent{\bar}{x},\bar x]$.
Assume that, at $t=0$, there are markets in the risk-free asset, the underlying, and a menu of European puts and calls maturing at $t=1$, with strikes $K\in[\underaccent{\bar}{x},\bar x]$ and payoffs $(X-K)_-$ and $(X-K)_+$, respectively.
Let $K_0$ denote the market maker’s prior mean of $X$.

Under the usual regularity (ensuring existence of a distributional second derivative), we have the classical~\cite{breeden1978prices} formula
\begin{align}
\label{eqn: Breeden-Litzenberger}
W(x, s)
&= W(K_0, s) + W'(K_0, s)(x - K_0)
 + \int_{\underaccent{\bar}{x}}^{K_0}  W''(K, s)\,(x-K)_-\,dK
 + \int_{K_0}^{\bar{x}} W''(K, s)\,(x-K)_+\,dK .
\end{align}
Thus, the Arrow--Debreu position $W(\,\cdot\,,s)$ can be replicated in the standard way by a portfolio consisting of a bond position $W(K_0,s)-K_0W'(K_0,s)$, $W'(K_0,s)$ units of the underlying, and a signed option-holding density $K\mapsto W''(K,s)$ (puts for $K<K_0$, calls for $K>K_0$).

In this specialization, the insider's private signal $s$ can encode essentially arbitrary information about the distribution of the terminal underlying price $X$---for instance, information about volatility, skewness, or tail risk.
He then submits an order $W'(K_0,s)$ in the underlying and an option-strip order $K \mapsto W''(K,s)$ across strike $K$.

The market maker receives aggregate (insider plus noise) order flow in the underlying and in the option strip. The underlying order flow is
\[
\omega_{\scriptscriptstyle a}=W'(K_0,s)+\varepsilon_{\scriptscriptstyle a},
\qquad 
\varepsilon_{\scriptscriptstyle a}\stackrel{d}{\sim}\mathcal N(0,\sigma_{\scriptscriptstyle a}^2).
\]
For the option strip, noise order flow follows a Gaussian process with independent increments across strikes. Accordingly, the option order flow admits the infinitesimal representation
\[
dY_{K}=W''(K,s)\,dK+\sigma(K)\,dB_{K},
\]
where $(B_{K})_{K\in[\underaccent{\bar}{x},\bar x]}$ is a Brownian motion indexed by the strike parameter $K$, and $\sigma(K)$ denotes the noise trading intensity at strike $K$.

The remainder of the trading game proceeds exactly as in the Arrow--Debreu formulation. 
Conversely, the results we obtain in the Arrow--Debreu formulation specialize directly to this setting, in which information about the underlying is traded through options. 
We comment on a few illustrative empirical and practical implications of this specialization in Section~\ref{sec: discussion}.

%%%%%%%%%%%%%%

\section{The Market Maker's Inference}
\label{sec: Market Maker's Pricing Kernel}

\subsection{Bayes' Rule}
\label{sec: Bayes' Rule}

We begin with a formal derivation of the market maker's Bayesian update; in the next subsection we make the argument rigorous by replacing the stochastic integral with a pathwise construction.

The realized order flow $\omega$ is an element of the path space $\Omega := C([\underaccent{\bar}{x}, \bar{x}], \R)$. 
Given the market maker's conjecture $\widetilde{W}(\,\cdot\,,\cdot)$ for the insider's trading strategy, conditional on a signal $s$ the received order flow is a sample path $\omega\in\Omega$ of the process
\begin{equation}
\label{eqn: dY s x}
dY_x = \widetilde{W}(x,s)\,dx + \sigma(x)\,dB_x .
\end{equation}

The market maker then applies Bayes' rule. The latent parameter is the signal $s$ (equivalently, the drift profile $\widetilde{W}(\,\cdot\,,s)$), with prior $\pi_0(ds)$, and the observation is the realized order-flow path $\omega$.
Let $\mathbb{P}_{\scriptscriptstyle \widetilde{W}(\,\cdot\,,s)}$ and $\mathbb{P}_{\! \scriptscriptstyle 0}$ denote the probability measures on $(\Omega,\mathcal{F})$ corresponding to the dynamics \eqref{eqn: dY s x} and to the zero-drift case $dY_x=\sigma(x)\,dB_x$, respectively.
Thus, $\mathbb{P}_{\scriptscriptstyle \widetilde{W}(\,\cdot\,,s)}$ is the law of order flow under signal $s$, while $\mathbb{P}_{\! \scriptscriptstyle 0}$ serves as a reference (noise-only) law.

Formally, the conditional likelihood of $\omega$ under $s$ (relative to $\mathbb{P}_{\! \scriptscriptstyle 0}$) is the likelihood ratio
$\mathrm{d}\mathbb{P}_{\scriptscriptstyle \widetilde{W}(\,\cdot\,,s)}/\mathrm{d}\mathbb{P}_{\! \scriptscriptstyle 0}$.
The Girsanov theorem suggests the representation
\begin{equation}
\label{eqn: Radon-Nikodym derivative}
\frac{ \mathrm{d} \mathbb{P}_{\scriptscriptstyle \widetilde{W}(\,\cdot\,,s)} }{ \mathrm{d} \mathbb{P}_{\! \scriptscriptstyle 0}}
=
\exp\!\left(
\int_{\underaccent{\bar}{x}}^{\bar{x}} \frac{\widetilde{W}(x,s)}{\sigma(x)}\,dB_x
-\frac12\int_{\underaccent{\bar}{x}}^{\bar{x}} \frac{\widetilde{W}(x,s)^2}{\sigma(x)^2}\,dx
\right)
\end{equation}
(see \citet[Section~3.5]{karatzas2012brownian}). 

Consequently, Bayes' rule suggests the posterior on $S$ given $\omega$ to be
\begin{equation}
\label{eqn: candidate expression for posterior}
\pi_1(ds \mid \omega\,;\widetilde{W})
=
\frac{1}{C(\omega\,;\widetilde{W})}\,
\frac{ \mathrm{d} \mathbb{P}_{\scriptscriptstyle \widetilde{W}(\,\cdot\,,s)} }{ \mathrm{d} \mathbb{P}_{\! \scriptscriptstyle 0}}(\omega)\,\pi_0(ds),
\qquad
C(\omega\,;\widetilde{W})
:=
\int_S 
\frac{ \mathrm{d} \mathbb{P}_{\scriptscriptstyle \widetilde{W}(\,\cdot\,,u)} }{ \mathrm{d} \mathbb{P}_{\! \scriptscriptstyle 0}}(\omega)\,\pi_0(du),
\end{equation}
so that $\pi_1(\cdot \mid \omega\,;\widetilde{W})$ integrates to one over $S$.

\subsection{The Posterior and Pricing Kernel}
\label{sec: Posterior and Pricing Kernel}

The market maker's posterior is meant to condition on a realized path $\omega$, whereas the right-hand side of the candidate expression~\eqref{eqn: candidate expression for posterior} is not defined pointwise in $\omega$. 
In particular, the likelihood ratio suggested by~\eqref{eqn: Radon-Nikodym derivative} is specified only as an $\mathbb{P}_{\! \scriptscriptstyle 0}$-a.s.\ equivalence class, since the It\^{o} integral $\int_{\underaccent{\bar}{x}}^{\bar{x}} \frac{\widetilde{W}(x,s)}{\sigma(x)}\,dB_x$ 
is not defined pathwise.

We now make Bayes' rule in our setting rigorous by replacing the stochastic integral with a pathwise $\omega$-by-$\omega$ integral. 
Concretely, on the (full $\mathbb{P}_{\! \scriptscriptstyle 0}$-measure) H\"older path space we consider a deterministic integral of the form 
\[
\int_{\underaccent{\bar}{x}}^{\bar{x}} \frac{\widetilde{W}(x,s)}{\sigma^2(x)}\,d\omega_x
\]
as in \cite{friz2020course}.
This provides an $\omega$-by-$\omega$ version of the stochastic integral tailored to our order-flow setting. It agrees with the corresponding It\^{o} integral $\mathbb{P}_{\! \scriptscriptstyle 0}$-a.s.\ and thereby produces a genuine likelihood functional
$\omega\mapsto \frac{d\mathbb{P}_{\scriptscriptstyle \widetilde{W}(\,\cdot\,,s)}}{d\mathbb{P}_{\! \scriptscriptstyle 0}}(\omega)$, hence a well-defined posterior conditioned on the observed order-flow path.

\begin{lemma}
\label{lemma: lemma 1 for example}
$\;$
\nopagebreak
Under Assumption~\ref{assumption: MM's posterior}, the following holds for the market maker's Bayesian inference problem.

(i) ($\int \cdot\, d \omega_x$ Integral) For all $\omega\in\Omega_\gamma$, $ W \in C^{\delta}([\underaccent{\bar}{x}, \bar{x}],\R)$, 
and $x\in [\underaccent{\bar}{x}, \bar{x}]$, the limit of Riemann sums
\begin{align*}
\int_{\underaccent{\bar}{x}}^x \frac{ W_y}{\sigma^2(y)} \,d\omega_y \, \equiv \,
\lim_{\substack{\max\limits_k\abs{x_{k+1}-x_k}\to 0\\ \underaccent{\bar}{x}=x_0<\cdots<x_n=x}} \,
\sum_{k=0}^{n-1} \frac{ W_{x_k} }{\sigma^2(x_k)} \cdot[\omega_{x_{k+1}\wedge x}-\omega_{x_k\wedge x}] %\quad \;\;\; \mbox{(\textbf{pathwise integral})}
\end{align*}
exists and therefore defines an $\omega$-by-$\omega$
Young integral.
%(\cite{friz2020course}, Chapter~4).

(ii) (Joint Measurability of Data and Parameter) The map
\begin{align*}
\underbrace{ (\omega,W) }_\text{(data, parameter)} \mapsto \int_{\underaccent{\bar}{x}}^{\bar{x}} \frac{W_x}{\sigma^2(x)}\,d\omega_x,\quad
(\Omega_\gamma,[\,\cdot\,]_\gamma) \times C^{\delta}([\underaccent{\bar}{x}, \bar{x}],\R) \to\R, 
\end{align*}
is continuous---in particular, measurable.
%(\cite{friz2020course}, Theorem 4.17).

(iii) (Conditional Likelihood of Data) For all $x\in [\underaccent{\bar}{x}, \bar{x}]$ and $W \in C^{\delta}([\underaccent{\bar}{x}, \bar{x}],\R)$, there exists a $\mathbb{P}_{\! \scriptscriptstyle 0, \gamma}$-null set
$N$, which may depend on $x$ and $W$, such that, for all $\omega\in\Omega_\gamma \! \setminus \! N$,
\begin{align*}
\int_{\underaccent{\bar}{x}}^x \frac{W_y}{\sigma^2(y)} \,d\omega_y=\left[\int_{\underaccent{\bar}{x}}^x \frac{W_y}{\sigma^2(y)} \,dB_y\right](\omega),
\end{align*}
where the integral on the right-hand side is a version of the It\^{o} integral.
%(\cite{friz2020course}, Proposition~5.1).

\end{lemma}

For the reader’s convenience, we include a proof, following the approach in \cite[especially Theorem~4.17 and Proposition~5.1]{friz2020course}, which treats a more general setting.

\begin{proof}

(i)
We note first that H\"older continuity is preserved by taking quotients when the denominator is bounded away from zero.
Let $W$, $\sigma \in C^\delta([\underaccent{\bar}{x},\bar{x}], \R)$ with $\min\limits_{ x\in [\underaccent{\bar}{x},\bar{x}]} \sigma(x) > \alpha > 0$ for some $\alpha > 0$. Then
\begin{align*}
\abs{\frac{W(y)}{\sigma(y)}-\frac{W(x)}{\sigma(x)}}& \le \frac{\abs{\sigma(x)}\,\abs{W(y)-W(x)}+\abs{W(x)}\,\abs{\sigma(y)-\sigma(x)}}{\abs{\sigma(y)}\,\abs{\sigma(x)}}\\
& \le \frac{\norm{W}_\infty\,[W]_\delta+\norm{\sigma}_\infty\,[\sigma]_\delta}{\alpha^2}\cdot\abs{y-x}^\delta,
\end{align*}
which implies $\frac{W}{\sigma} \in C^\delta([\underaccent{\bar}{x},\bar{x}], \R)$. Therefore, it suffices to prove the claim with $W$ in place of $\frac{W}{\sigma^2}$.

Given $\alpha$, $\beta>0$, we shall consider the space
$C^{\alpha,\beta}_2([\underaccent{\bar}{x},\bar{x}],\R)$ of all functions $\Xi$ from $\{(y,x):\, \underaccent{\bar}{x} \le y\le x\le \bar{x}\}$ to $\R$ with seminorm
\begin{align*}
[\Xi]_{\alpha,\beta} \, \equiv \, \underbrace{[\Xi]_\alpha}_{\equiv \, \sup\limits_{y<x} \frac{\abs{\Xi_{y,x}}}{(x-y)^\alpha}}
+\sup_{y<r<x}\frac{\abs{\Xi_{y,x}-\Xi_{y,r}-\Xi_{r,x}}}{\abs{x-y}^\beta}<\infty.
\end{align*}

For $\omega\in\Omega_\gamma$ and $W\in C^\delta([\underaccent{\bar}{x},\bar{x}], \R)$,
define $\Xi^{\omega,W}$ by 
\begin{align*}
\Xi^{\omega,W}_{y,x} \, \equiv \, W_y\cdot(\omega_x-\omega_y).
\end{align*}
Since $\abs{\Xi_{y,x}^{\omega,W}}\le [\omega]_\gamma\,\abs{x-y}^\gamma$ and, for $\underaccent{\bar}{x} \le y\le r\le x\le \bar{x}$,
\begin{align*}
\abs{\Xi^{\omega,W}_{y,x}-\Xi^{\omega,W}_{y,r}-\Xi^{\omega,W}_{r,x}}
=\abs{(W_y-W_r)\cdot (\omega_x-\omega_r)}\le [W]_\delta\, [\omega]_\gamma\,\abs{x-y}^{\gamma+\delta},
\end{align*}
we have $\Xi^{\omega,W}\in C^{\gamma,\gamma+\delta}_2([\underaccent{\bar}{x},\bar{x}],\R)$. 
Thus, because $\gamma\le 1<\gamma+\delta$, it follows from the Sewing Lemma % of
(\cite{friz2020course}, Lemma 4.2) that the limit of Riemann sums
\begin{align*}
\lim_{\substack{\max\limits_k\abs{x_{k+1}-x_k}\to 0\\ \underaccent{\bar}{x} =x_0<\cdots<x_n=x}} \,
\sum_{k=0}^{n-1} W_{x_k} \cdot[\omega_{x_{k+1}\wedge x}-\omega_{x_k\wedge x}]
\end{align*}
exists. This proves the claim.

(ii)
We will prove the continuity of the map (this is a stronger property than that stated in Lemma~\ref{lemma: lemma 1 for example}(ii))
\begin{align*}
(\omega,W)\mapsto \int_0^{\bullet} W_t\,d\omega_t,\quad (\Omega_\gamma,[\,\cdot\,]_\gamma)\times C^\delta([\underaccent{\bar}{x},\bar{x}], \R) \to C^\gamma([\underaccent{\bar}{x}, \bar{x}],\R).
\end{align*}
By the Sewing Lemma quoted in (i), it suffices to check the continuity of
\begin{align*}
(\omega,W)\mapsto \Xi^{\omega,W},\quad (\Omega_\gamma,[\,\cdot\,]_\gamma)\times C^\delta([\underaccent{\bar}{x} ,\bar{x}], \R) \to C^{\gamma,\gamma+\delta}_2([\underaccent{\bar}{x} , \bar{x}],\R).
\end{align*}
Given $\omega$, $\tilde{\omega}\in\Omega_\gamma$, $W$, $\tilde{W}\in C^\delta([\underaccent{\bar}{x} ,\bar{x}], \R)$,
\begin{align*}
&\abs{(\Xi^{\omega,W}_{y,x}-\Xi^{\omega,W}_{y,r}-\Xi^{\omega,W}_{r,x})-
(\Xi^{\tilde{\omega},\tilde{W}}_{y,x}-\Xi^{\tilde{\omega},\tilde{W}}_{y,r}-
\Xi^{\tilde{\omega},\tilde{W}}_{r,x})}\\
&\qquad =\abs{(W_y-W_r)\cdot(\omega_x-\omega_r)-(\tilde{W}_y-\tilde{W}_r)\cdot(\tilde{\omega}_x-
\tilde{\omega}_r)}\\
&\qquad \le [W-\tilde{W}]_\delta\,[\omega]_\gamma\,\abs{x-y}^{\gamma+\delta}+
[\tilde{W}]_\delta\,[\omega-\tilde{\omega}]_\gamma\,\abs{x-y}^{\gamma+\delta}.
\end{align*}
Therefore (below we use the H\"older norm $\norm{\, \cdot \,}_W$ instead of the seminorm $[\, \cdot \,]_{W}$)
\begin{align*}
\abs{\Xi^{\omega,W}_{y,x}-\Xi^{\tilde{\omega},\tilde{W}}_{y,x}} &=\abs{W_y \cdot(\omega_x-\omega_y)-\tilde{W}_y \cdot(\tilde{\omega}_x-\tilde{\omega}_y)}\\
&\le \norm{W}_\infty\,[\omega]_\gamma\,\abs{x-y}^\gamma+
\norm{W-\tilde{W}}_\delta\cdot(\bar{x}^\delta \vee 1)\,[\tilde{\omega}]_{\gamma}\,\abs{x-y}^\gamma.
\end{align*}
This proves the claim.

(iii)
Let $\int_{\underaccent{\bar}{x}}^x W_y\,dB_y$ be a given version of the It\^{o} integral.
By the continuity of $W$, the chosen version is a probability limit (see \cite{revuz2013continuous})
\begin{align*}
\int_{\underaccent{\bar}{x}}^x W_y\,dB_y = 
\lim_{\substack{\max\limits_k\abs{x_{k+1}-x_k}\to 0\\ \underaccent{\bar}{x} = x_0 <\cdots<x_n=x}} \,
\sum_{k=0}^{n-1} W_{x_k}\cdot[B_{x_{k+1}\wedge x}-B_{x_k\wedge x}]
\quad\text{in $\mathbb{P}_{\! \scriptscriptstyle 0, \gamma}$-probability.}
\end{align*}
Then one can pass to a subsequence of the (implicitly given) sequence of partitions such that,
for $\mathbb{P}_{\! \scriptscriptstyle 0, \gamma}$-a.e.~$\omega\in\Omega_\gamma$ (where we use the same notation for the subsequence),
\begin{align*}
\left[\int_{\underaccent{\bar}{x}}^x W_y\,dB_y\right](\omega) \equiv
\lim_{\substack{\max\limits_k\abs{x_{k+1}-x_k}\to 0\\ \underaccent{\bar}{x}=x_0<\cdots<x_n=x}} \, 
\sum_{k=0}^{n-1} W_{x_k}\cdot[\omega_{x_{k+1}\wedge x}-\omega_{x_k\wedge x}]
=\int_{\underaccent{\bar}{x}}^x W_y\,d\omega_y.
\end{align*}
This proves the claim.
\end{proof}

\begin{lemma}
\label{lemma: lemma 2 for example}

Under Assumption~\ref{assumption: MM's posterior}, let $\omega$ denote the canonical process on $(\Omega_\gamma,\mathcal{F}_\gamma,\mathbb{P}_{\!\scriptscriptstyle 0,\gamma})$.
Fix $W\in C^\delta([\underaccent{\bar}{x},\bar{x}],\R)$ and define a probability measure $\mathbb{P}_{\!\scriptscriptstyle W}$ on $(\Omega_\gamma,\mathcal{F}_\gamma)$ by
\begin{align}
\label{eqn: Radon-Nikodym density, omega by omega}
\frac{\mathrm{d}\mathbb{P}_{\!\scriptscriptstyle W}}{\mathrm{d}\mathbb{P}_{\!\scriptscriptstyle 0,\gamma}}(\omega)
=
\exp\!\left\{
\int_{\underaccent{\bar}{x}}^{\bar{x}} \frac{W_x}{\sigma^2(x)}\,\mathrm{d}\omega_x
-\frac12\int_{\underaccent{\bar}{x}}^{\bar{x}} \frac{W_x^2}{\sigma^2(x)}\,\mathrm{d}x
\right\},
\end{align}
where $\int_{\underaccent{\bar}{x}}^{\bar{x}} \frac{W_x}{\sigma^2(x)}\,\mathrm{d}\omega_x$ is the pathwise integral from Lemma~\ref{lemma: lemma 1 for example}(i).
Then, under $\mathbb{P}_{\!\scriptscriptstyle W}$, there exists a standard Brownian motion $(B_x)_{x\in[\underaccent{\bar}{x},\bar{x}]}$ such that the canonical process satisfies
\[
d \omega_x = W_x\,dx + \sigma(x)\,dB_x .
\]

\end{lemma}

\begin{proof}
This follows immediately from Lemma~\ref{lemma: lemma 1 for example}(iii).
\end{proof}

\begin{theorem} 
\label{thm: Market Maker's Posterior}
\nopagebreak

Fix a market maker belief $\widetilde{W}(\,\cdot\, , \,\cdot\,)$ and let $\omega\in\Omega_\gamma$ denote the realized cumulative order-flow path (equivalently, the collection of infinitesimal order flows $(d\omega_x)_{x\in[\underaccent{\bar}{x},\bar{x}]}$).
Then the market maker's posterior on $S$ is the probability measure $\pi_1(\,\cdot\,, \omega \kern 0.045em ; \widetilde{W})$ given by
\begin{equation}
\label{eqn: rigorous expression for posterior}
\pi_1(ds, \omega \kern 0.045em ; \widetilde{W})
=
\frac{1}{C(\omega \kern 0.045em ; \widetilde{W})}\,
\exp\!\left\{
\int_{\underaccent{\bar}{x}}^{\bar{x}} \! \frac{\widetilde{W}(x,s)}{\sigma^2(x)}\, d\omega_x
-\frac12 \int_{\underaccent{\bar}{x}}^{\bar{x}} \! \frac{\widetilde{W}(x,s)^2}{\sigma^2(x)}\, dx
\right\}\,\pi_0(ds),
\end{equation}
where $\int \cdot\, d\omega_x$ is the pathwise integral from Lemma~\ref{lemma: lemma 1 for example}(i), and
$C(\omega \kern 0.045em ; \widetilde{W})$ is the normalizing constant chosen so that $\int_S \pi_1(ds, \omega \kern 0.045em ; \widetilde{W})=1$.
In particular, $\pi_1(\,\cdot\,, \omega \kern 0.045em ; \widetilde{W})$ depends on $\omega$ only through the projection profile
$\bigl(\int_{\underaccent{\bar}{x}}^{\bar{x}} \! \frac{\widetilde{W}(x,s)}{\sigma^2(x)}\, d\omega_x \bigr)_{s\in S}$.
The market maker's zero-profit Arrow--Debreu prices are then
\[
P(x, \omega \kern 0.045em ; \widetilde{W})
=
\int_{S} \eta(x, s)\, \pi_1(ds, \omega \kern 0.045em ; \widetilde{W}),
\qquad x \in [\underaccent{\bar}{x}, \bar{x}] .
\]
\end{theorem}

\begin{proof}
It suffices to establish the posterior formula \eqref{eqn: rigorous expression for posterior}, after which the pricing identity follows from the market maker's zero-profit condition.

Per Assumption~\ref{assumption: MM's posterior}, we work on the H\"older path space $(\Omega_\gamma,\mathcal{F}_\gamma,\mathbb{P}_{\!\scriptscriptstyle 0,\gamma})$.
For each $s\in S$, define the likelihood functional
\[
L(\omega,s)
\coloneqq
\exp\!\left\{
\int_{\underaccent{\bar}{x}}^{\bar{x}} \frac{\widetilde{W}(x,s)}{\sigma^2(x)}\, d\omega_x
-\frac12 \int_{\underaccent{\bar}{x}}^{\bar{x}} \frac{\widetilde{W}(x,s)^2}{\sigma^2(x)}\, dx
\right\},
\qquad \omega\in\Omega_\gamma,
\]
where $\int\cdot\, d\omega_x$ is the pathwise integral from Lemma~\ref{lemma: lemma 1 for example}(i).
By Lemma~\ref{lemma: lemma 1 for example}(i)--(ii), $L(\,\cdot\,,s)$ is well-defined on $\Omega_\gamma$ and $(\omega,s)\mapsto L(\omega,s)$ is jointly measurable.

For each $s\in S$, define a probability measure $\mathbb{P}_{\scriptscriptstyle \widetilde{W}(\,\cdot\,,s)}$ on $(\Omega_\gamma,\mathcal{F}_\gamma)$ by
\[
\frac{ \mathrm{d} \mathbb{P}_{\scriptscriptstyle \widetilde{W}(\,\cdot\,,s)} }{ \mathrm{d} \mathbb{P}_{\! \scriptscriptstyle 0,\gamma} }(\omega)
=
L(\omega,s).
\]
Lemma~\ref{lemma: lemma 2 for example} implies that, under $\mathbb{P}_{\scriptscriptstyle \widetilde{W}(\,\cdot\,,s)}$, the canonical process has the same law as the drifted order-flow model
$dY_x=\widetilde{W}(x,s)\,dx+\sigma(x)\,dB_x$.
Consequently, the joint law of $(\omega,s)$ induced by the belief $\widetilde{W}$ and prior $\pi_0$ is
\[
\mathbb{Q}(d\omega,ds)
\coloneqq
\mathbb{P}_{\scriptscriptstyle \widetilde{W}(\,\cdot\,,s)}(d\omega)\,\pi_0(ds)
=
L(\omega,s)\,\mathbb{P}_{\! \scriptscriptstyle 0,\gamma}(d\omega)\,\pi_0(ds).
\]
Disintegrating $\mathbb{Q}$ with respect to its $\omega$-marginal and applying Fubini--Tonelli yields the conditional distribution of $s$ given $\omega$:
\[
\pi_1(ds, \omega \kern 0.045em ; \widetilde{W})
=
C(\omega \kern 0.045em ; \widetilde{W})^{-1}\,L(\omega,s)\,\pi_0(ds),
\qquad
C(\omega \kern 0.045em ; \widetilde{W})
\coloneqq
\int_S L(\omega,u)\,\pi_0(du),
\]
which is exactly \eqref{eqn: rigorous expression for posterior}.
\end{proof}

\subsection*{Cross-Market Inference} 

Having established Theorem~\ref{thm: Market Maker's Posterior} on a pathwise (pointwise-in-$\omega$) footing, we can give a transparent interpretation of the market maker's updating rule~\eqref{eqn: rigorous expression for posterior}. 

Fixing the market maker's conjecture $\widetilde{W}(\,\cdot\,,\cdot\,)$, each signal $s \in S$ corresponds to a hypothesized drift profile $x\mapsto \widetilde{W}(x,s)$ for the order-flow process.
The posterior weight assigned to each $s$ is obtained by an exponential tilt of the prior $\pi_0$ whose data-dependent term is the noise-adjusted pairing
\[
\int_{\underaccent{\bar}{x}}^{\bar{x}} \! \frac{\widetilde{W}(x,s)}{\sigma^2(x)}\, d\omega_x .
\]
Accordingly, the realized order-flow path $\omega$ influences the posterior only through the profile
$\bigl(\int_{\underaccent{\bar}{x}}^{\bar{x}} \! \frac{\widetilde{W}(x,s)}{\sigma^2(x)}\, d\omega_x \bigr)_{s\in S}$. 
In other words, signals whose conjectured informed demand profile $\widetilde{W}(\,\cdot\,,s)$ is more aligned with the observed order flow receive larger posterior mass, while the quadratic term
$\frac12\int_{\underaccent{\bar}{x}}^{\bar{x}} \! \frac{\widetilde{W}(x,s)^2}{\sigma^2(x)}\,dx$
is the corresponding likelihood penalty for larger conjectured drifts.

The factor $\sigma^{-2}(x)$ makes explicit how heterogeneous noise affects learning: order flow in markets with larger noise variance is down-weighted in the statistic and contributes less to inference. 
In particular, as $\sigma(\cdot)$ becomes large, the projection statistic becomes negligible and the posterior remains close to the prior. In the options interpretation, for instance, if $\widetilde{W}(\,\cdot\,,s)$ predicts \emph{straddle}-like demand under a high-volatility 
signal $s$ (a standard volatility trade; cf.\ \cite{hull2003options}), then order flow exhibiting a pronounced straddle component increases the above pairing and therefore shifts posterior mass toward that signal.

When $[\underaccent{\bar}{x},\bar{x}]$ is collapsed to a singleton as in \cite{kyle1985continuous}, the order flow $\omega$ becomes scalar and Theorem~\ref{thm: Market Maker's Posterior} recovers the familiar single-asset inference logic in Kyle. 
There, if the market maker conjectures that an insider who observes value $s$ submits a buy order, then observing net buying increases the posterior weight on $s$ and moves the price toward $s$, and conversely for net selling (cf.\ the Kyle equilibrium pricing rule~\eqref{eqn: Kyle equil}).

\subsection{Sufficient Statistics}
\label{sec: sufficient statistics}

\begin{definition}
\label{def: overlap measures}
$\;$
\nopagebreak

(i) For the market maker, the noise-adjusted projection coefficient of a realized cumulative order-flow path $\omega$ onto his belief $\widetilde{W}(\,\cdot\,, s)$ is
\begin{equation}
\label{eqn: MM overlap measure}
\Pi_{mm}(\omega, s ; \widetilde{W})
:= \int_{\underaccent{\bar}{x}}^{\bar{x}}  \frac{ \widetilde{W}(x, s) }{\sigma^2(x)} \, d \omega_x .
\end{equation}

(ii) For the insider, the noise-adjusted projection coefficient of a portfolio payoff profile $W(\,\cdot\,)$ onto $\widetilde{W}(\,\cdot\,, s)$ is
\begin{equation}
\label{eqn: insider overlap measure}
\Pi_{insider}(W, s ; \widetilde{W})
:= \int_{\underaccent{\bar}{x}}^{\bar{x}}  \frac{ \widetilde{W}(x, s) }{\sigma^2(x)} \, W(x)\, dx .
\end{equation}
\end{definition}

In view of Theorem~\ref{thm: Market Maker's Posterior}, for fixed $\widetilde{W}$ the market maker’s inference is summarized by the sufficient statistic
$s \mapsto \Pi_{mm}(\omega, s ; \widetilde{W})$.
The insider-facing quantity $\Pi_{insider}(W, s ; \widetilde{W})$ plays an analogous role for portfolio choice, as we show in Section~\ref{sec: Informed Trader's Problem}.
The distinction between the two is informational---the market maker computes $\Pi_{mm}$ from aggregate order flow $\omega$, whereas the insider can condition on his own chosen payoff profile $W$.

Together, $\Pi_{mm}$ and $\Pi_{insider}$ make explicit the \emph{cross-market} linkage between quantities and prices for their respective agents. 
Through $\Pi_{mm}(\omega, \,\cdot\, ; \widetilde{W})$, the market maker compares the realized order flow $\omega$ to the conjectured signal-contingent informed-demand profiles $\widetilde{W}(\,\cdot\,,s)$ and sets his prices accordingly.
Through $\Pi_{insider}(W, \,\cdot\, ; \widetilde{W})$, the insider evaluates how a contemplated portfolio $W$ interacts with these same profiles and, through the induced posterior, affects prices.
In equilibrium, the conditional expectation of $\Pi_{mm}(\omega, s ; \widetilde{W})$ under the distribution of $\omega$ given $s$ equals $\Pi_{insider}(W, s ; \widetilde{W})$.

\section{The Insider's Portfolio Choice}
\label{sec: Informed Trader's Problem}

In this section, we derive two necessary conditions for the insider’s portfolio-choice problem to be well-posed: a first-order (variational) condition and a viability (no-free-lunch with zero price impact) condition.
Since, in equilibrium, the insider must solve his portfolio-choice problem against the market maker’s pricing rule, these conditions are, in particular, necessary for equilibrium.

Fix a signal realization $s\in S$ and a market-maker belief $\widetilde{W}(\,\cdot\,,\,\cdot\,)$.
For any admissible insider portfolio $W\in C^{\delta}([\underaccent{\bar}{x},\bar{x}],\R)$, let $\mathbb{P}_{\!\scriptscriptstyle W}$ denote the induced law on $(\Omega_\gamma,\mathcal{F}_\gamma)$ under which the canonical order-flow path satisfies
$d\omega_x = W(x)\,dx + \sigma(x)\,dB_x$ (cf.\ \eqref{eqn: dY_x} and Lemma~\ref{lemma: lemma 2 for example}).
Recalling the Arrow--Debreu pricing kernel $P(x,\omega\kern0.045em;\widetilde{W})$ from \eqref{eqn: informed trader's problem informal}, an application of Fubini--Tonelli yields the equivalent representation
\begin{align}
\max\limits_{ {\scriptscriptstyle W(\cdot)  \in  C^{\delta}\!([\underaccent{\bar}{x}, \bar{x}], \R)} } \! J(W; \widetilde{W}, s)
&= \max\limits_{  {\scriptscriptstyle W(\,\cdot\,)  \in  C^{\delta}\!([\underaccent{\bar}{x}, \bar{x}], \R)} }
\left\{
\underbrace{ \int_{\underaccent{\bar}{x}}^{\bar{x}} \!\!\! W(x)\,\eta(x,s)\, dx }_\text{expected payoff}
-
\underbrace{ \int_{\underaccent{\bar}{x}}^{\bar{x}} \!\!\! W(x)\, \overline{P}(x, W; \widetilde{W}) \, dx }_\text{expected cost}
\right\},
\label{eqn: informed trader's problem}
\end{align}
where
\[
\overline{P}(\,\cdot\,, W; \widetilde{W})
\;:=\;
\mathbb{E}^{\mathbb{P}_{\!\scriptscriptstyle W}}\!\bigl[P(\,\cdot\,, \omega \kern 0.045em ; \widetilde{W})\bigr]
\]
is the $\mathbb{P}_{\!\scriptscriptstyle W}$-expected pricing kernel faced by the insider when he submits $W$.

We write $\partial J(W;\widetilde{W},s)$ for the G\^ateaux derivative of $J(\,\cdot\,;\widetilde{W},s)$ at $W$.
That is, for each direction $v\in C^{\delta}([\underaccent{\bar}{x},\bar{x}],\R)$,
\[
\partial J(W;\widetilde{W},s)(v)
:= \lim_{\varepsilon\to 0}\frac{J(W+\varepsilon v;\widetilde{W},s)-J(W;\widetilde{W},s)}{\varepsilon},
\]
whenever the limit exists.

\begin{remark}
In standard finance terminology, $\partial J(W;\widetilde{W},s)(v)$ is the insider’s ``marginal'' expected profit in the direction $v$---that is, the first-order change in the objective induced by an infinitesimal perturbation of the portfolio in the direction $v$ (see, e.g., \cite{bodie2021investments} for textbook uses of marginal benefits and costs in basic portfolio choice).
We will use this marginal terminology throughout, without further comment, when the finance context warrants.
\end{remark}

\subsection{First-Order Condition}
\label{subsec: FOC and Arrow-Debreu-Kyle Decomposition}

A perturbation of $W$ in direction $v$ affects $J(W;\widetilde{W},s)$ in two ways. 
First, it changes the payoff term directly, through the linear pairing $\int_{\underaccent{\bar}{x}}^{\bar{x}} v(x)\,\eta(x,s)\,dx$. 
Second, it changes the induced order-flow law $\mathbb{P}_{\!\scriptscriptstyle W}$ and hence the market maker’s posterior, which enters the expected pricing kernel $\overline{P}(\,\cdot\,,W;\widetilde{W})$. 
Accordingly, a local adjustment of the position in a given security typically induces a change in prices across the entire cross-section of Arrow--Debreu securities. 
Theorem~\ref{thm: informed trader 1} makes this precise by computing the G\^{a}teaux derivative $\partial J(W;\widetilde{W},s)(v)$ and separating the ``no-impact'' Arrow--Debreu component from the additional price-impact term generated by the induced change in the posterior. 
Corollary~\ref{cor: def of price impact} then records the associated cross-market price-impact kernel.

In finance terms, the insider’s trades to exploit private information entail an informational trading cost, because his orders reveal information to the market maker. 
For each security, the insider balances its marginal payoff against a marginal cost with two components: first, the prevailing execution price, and second, the additional marginal cost generated by the cross-market repricing induced by information leakage.

\begin{theorem}[Insider First-Order Condition]
\nopagebreak
\label{thm: informed trader 1}
Fix $s\in S$ and a market-maker belief $\widetilde{W}(\,\cdot\,,\cdot\,)$.  Under the standing assumptions, the functional
$J(\,\cdot\,;\widetilde{W},s)$ is G\^{a}teaux differentiable on the admissible class.  For any admissible portfolio $W$ and any admissible direction
$v$, its G\^{a}teaux derivative admits the decomposition
\[
\partial \! \kern 0.08em J(W; \widetilde{W}, s)(v)
=
\partial \! \kern 0.08em J_p(W)(v)
-
\partial \! \kern 0.08em J_{\scriptscriptstyle \! AD}(W)(v)
-
\partial \! \kern 0.08em J_{\scriptscriptstyle \! K}(W)(v),
\]
where, with $(\widetilde{W},s)$ held fixed,
\begin{align}
\partial \! \kern 0.08em J_p(W)(v)
&:= \int_{\underaccent{\bar}{x}}^{\bar{x}} v(x)\,\eta(x,s)\,dx,
\label{eqn: marginal payoff}\\[0.25em]
\partial \! \kern 0.08em J_{\scriptscriptstyle \! AD}(W)(v)
&:= \int_{\underaccent{\bar}{x}}^{\bar{x}} \overline{P}(x, W; \widetilde{W})\,v(x)\,dx,
\label{eqn: AD term}\\[0.25em]
\partial \! \kern 0.08em J_{\scriptscriptstyle \! K}(W)(v)
&:= \int_{\underaccent{\bar}{x}}^{\bar{x}} W(x)\,
\mathbb{E}^{\mathbb{P}_{\scriptscriptstyle W}}\!\Bigl[
\Cov\!\Bigl(\eta(x,\cdot),\Pi_{insider}(v,\cdot;\widetilde{W})\mid \omega \,;\widetilde{W}\Bigr)
\Bigr]\,dx.
\label{eqn: Kyle term in FOC}
\end{align}
In particular, if $W$ maximizes $J(\,\cdot\,;\widetilde{W},s)$ over the admissible class, then
$\partial \! \kern 0.08em J(W; \widetilde{W}, s)(v)=0$ for all admissible directions $v$, equivalently
\begin{equation}
\label{eqn: insider FOC}
 \partial \! \kern 0.08em J_p(W)(\,\cdot\,)  =  \partial \! \kern 0.08em J_{\scriptscriptstyle \! AD}(W)(\,\cdot\,) + \partial \! \kern 0.08em J_{\scriptscriptstyle \! K}(W)(\,\cdot\,) .
\end{equation}
\end{theorem}

\begin{proof}
Fix $s\in S$ and a market-maker belief $\widetilde{W}(\,\cdot\,,\cdot\,)$. Recall from
\eqref{eqn: informed trader's problem informal}--\eqref{eqn: informed trader's problem} that
\[
J(W;\widetilde{W},s)
=
\int_{\underaccent{\bar}{x}}^{\bar{x}} W(x)\,\eta(x,s)\,dx
-
\int_{\underaccent{\bar}{x}}^{\bar{x}} W(x)\,\overline{P}(x,W;\widetilde{W})\,dx,
\]
where $\overline{P}(x,W;\widetilde{W})=\mathbb{E}^{\mathbb{P}_{\scriptscriptstyle W}}\!\bigl[P(x,\omega \kern 0.045em;\widetilde{W})\bigr]$ and
\[
P(x,\omega \kern 0.045em;\widetilde{W})
=
\int_S \eta(x,s')\,\pi_1(ds',\omega \kern 0.045em;\widetilde{W})
\]
by Theorem~\ref{thm: Market Maker's Posterior}; see \eqref{eqn: rigorous expression for posterior}.

\medskip\noindent
\textit{Step 1: Push-forward representation $\overline{P}$ (with explicit $W$-dependence).}
Using \eqref{eqn: rigorous expression for posterior}, define
\begin{equation}
\label{eqn: caligraphic I}
\mathcal{I}(\omega, s ; \widetilde{W})
=
\int_{\underaccent{\bar}{x}}^{\bar{x}} \frac{\widetilde{W}(x,s)}{\sigma^2(x)}\,d\omega_x
-\frac12\int_{\underaccent{\bar}{x}}^{\bar{x}} \frac{\widetilde{W}(x,s)^2}{\sigma^2(x)}\,dx .
\end{equation}
Then
\begin{align}
\nopagebreak
\overline{P}(x,W;\widetilde{W})
&=\mathbb{E}^{\mathbb{P}_{\scriptscriptstyle W}}\!\Bigl[\int_S \eta(x,u)\,\pi_1(du,\omega \kern 0.045em;\widetilde{W})\Bigr] \nonumber\\
&=\mathbb{E}^{\mathbb{P}_{\scriptscriptstyle W}}\!\Bigl[\int_S \eta(x,u)\,
\frac{e^{\mathcal{I}(\omega,u;\widetilde{W})}\,\pi_0(du)}{C(\omega \kern 0.045em;\widetilde{W})}\Bigr] \nonumber\\
&=\mathbb{E}^{\mathbb{P}_{\scriptscriptstyle 0,\gamma}}\!\Bigl[\int_S \eta(x,u)\,
\frac{e^{\mathcal{I}(\omega,u;\widetilde{W})}\,
e^{\int_{\underaccent{\bar}{x}}^{\bar{x}} \frac{\widetilde{W}(x',u)\,W(x')}{\sigma^2(x')}\,dx'}\,
\pi_0(du)}{C'(\omega \kern 0.045em;\widetilde{W})}\Bigr],
\label{eqn: pushed forward AD bar}
\end{align}
where
\begin{equation}
\label{eqn: C prime}
C'(\omega \kern 0.045em;\widetilde{W})
=
\int_S e^{\mathcal{I}(\omega,u;\widetilde{W})}\,
e^{\int_{\underaccent{\bar}{x}}^{\bar{x}} \frac{\widetilde{W}(x',u)\,W(x')}{\sigma^2(x')}\,dx'}\,\pi_0(du).
\end{equation}
The identity \eqref{eqn: pushed forward AD bar} is the usual Cameron--Martin/Girsanov shift for the order-flow dynamics \eqref{eqn: dY_x}
(cf.\ \cite{karatzas2012brownian}): under $\mathbb{P}_{\scriptscriptstyle W}$ the canonical process has the same law as
$x\mapsto \omega_x+\int_{\underaccent{\bar}{x}}^{x} W(x')\,dx'$ under $\mathbb{P}_{\scriptscriptstyle 0,\gamma}$.

\medskip\noindent
\textit{Step 2: The payoff term.}
The map $W\mapsto \int_{\underaccent{\bar}{x}}^{\bar{x}} W(x)\,\eta(x,s)\,dx$ is linear, hence its G\^{a}teaux derivative in direction $v$ equals
$\int_{\underaccent{\bar}{x}}^{\bar{x}} v(x)\,\eta(x,s)\,dx$, which is \eqref{eqn: marginal payoff}.

\medskip\noindent
\textit{Step 3: The cost term.}
Let
\[
J_c(W) := \int_{\underaccent{\bar}{x}}^{\bar{x}} W(x)\,\overline{P}(x,W;\widetilde{W})\,dx,
\qquad
f(\varepsilon) := J_c(W+\varepsilon v).
\]
Under the standing regularity conditions (so that differentiation may be interchanged with integration/expectation),
\begin{align}
\label{eqn: two integrals}
\partial \! \kern 0.08em J_c(W)(v)
&= f'(0) \nonumber\\
&=\int_{\underaccent{\bar}{x}}^{\bar{x}} v(x)\,\overline{P}(x,W;\widetilde{W})\,dx
+\int_{\underaccent{\bar}{x}}^{\bar{x}} W(x)\,g(x)\,dx,
\end{align}
where $g(x):=\left.\frac{d}{d\varepsilon}\right|_{\varepsilon=0}\overline{P}\bigl(x,W+\varepsilon v;\widetilde{W}\bigr)$.
The first integral in \eqref{eqn: two integrals} is exactly \eqref{eqn: AD term}.

\medskip\noindent
\textit{Step 4: Identifying $g(x)$ as a conditional covariance.}
For fixed $W$ and $(u,\omega)$, define
\begin{equation}
\label{eqn: dots}
l(u,\omega; \widetilde{W})
:=
\exp\!\Bigl(\mathcal{I}(\omega,u;\widetilde{W})\Bigr)\,
\exp\!\Bigl(\int_{\underaccent{\bar}{x}}^{\bar{x}}
\frac{\widetilde{W}(x',u)\,W(x')}{\sigma^2(x')}\,dx'\Bigr).
\end{equation}
Then \eqref{eqn: C prime} reads $C'(\omega \kern 0.045em;\widetilde{W})=\int_S l(u,\omega; \widetilde{W})\,\pi_0(du)$, and the integrand in \eqref{eqn: pushed forward AD bar} is
\[
\int_S \eta(x,u)\,\frac{l(u,\omega; \widetilde{W})}{C'(\omega \kern 0.045em;\widetilde{W})}\,\pi_0(du).
\]
For a direction $v$, set
\[
\Pi_{insider}(v,u;\widetilde{W})
:=
\int_{\underaccent{\bar}{x}}^{\bar{x}} \frac{\widetilde{W}(x',u)}{\sigma^2(x')}\,v(x')\,dx'.
\]
Differentiating the normalized weight $l/C'$ in direction $v$ yields the usual centered term
\[
\Pi_{insider}(v,\cdot;\widetilde{W})-\mathbb{E}\!\left[\Pi_{insider}(v,\cdot;\widetilde{W})\mid \omega \kern 0.045em;\widetilde{W}\right],
\]
and hence the standard covariance identity: for each fixed $\omega$,
\[
g(x)
=
\mathbb{E}^{\mathbb{P}_{\scriptscriptstyle W}}\!\Bigl[
\Cov\!\bigl(\eta(x,\cdot),\Pi_{insider}(v,\cdot;\widetilde{W}) \mid \omega \kern 0.045em;\widetilde{W}\bigr)
\Bigr].
\]
Substituting this expression for $g(x)$ into \eqref{eqn: two integrals} identifies the second integral in \eqref{eqn: two integrals} with
\eqref{eqn: Kyle term in FOC}. Combining \eqref{eqn: marginal payoff} with \eqref{eqn: two integrals} gives the stated decomposition of
$\partial \! \kern 0.08em J(W;\widetilde{W},s)(v)$, and hence the first-order condition \eqref{eqn: insider FOC}.
This proves the theorem.
\end{proof}

The first-order condition~\eqref{eqn: insider FOC} is an identity of (bounded) linear functionals on the admissible tangent space. Via~\eqref{eqn: marginal payoff},
the payoff derivative $\partial \! \kern 0.08em J_p(W)$ is represented under the integral pairing by the signal--conditional payoff profile $x\mapsto \eta(x,s)$.
Likewise, via~\eqref{eqn: AD term}, the Arrow--Debreu cost component $\partial \! \kern 0.08em J_{\scriptscriptstyle \! AD}(W)$ is represented by the pricing kernel
$x\mapsto \overline{P}(x,W;\widetilde{W})$; this is the prevailing marginal execution price. The remaining term
$\partial \! \kern 0.08em J_{\scriptscriptstyle \! K}(W)$ is the price-impact functional---it captures the additional marginal cost from cross-market repricing induced by information leakage.

What distinguishes the price-impact functional from the payoff and Arrow--Debreu terms is that it is entirely mediated by inference.
We next consider its underlying mechanism and the resulting cross-market repricing effect.

\subsection{Cross-Market Price Impact}
\label{sec: Kyle Term and Cross-Market Price Impact}

\paragraph{The Price-Impact Functional $\partial \! \kern 0.08em J_{\scriptscriptstyle \! K}(W)$}
The functional $\partial \! \kern 0.08em J_{\scriptscriptstyle \! K}(W)$ arises from the $W$-dependence of the pricing
kernel $x\mapsto \overline{P}(x,W;\widetilde{W})$, which enters through posterior updating.
By the covariance representation~\eqref{eqn: Kyle term in FOC}, for a given aggregate order-flow realization
$\omega$ the G\^ateaux derivative of the (random) price $P(x,\omega\kern 0.045em;\widetilde{W})$ in the direction
$v$ is
\begin{equation}
\label{eqn: price impact of v}
\Cov \Bigl(\eta(x,\cdot),\Pi_{insider}(v,\cdot;\widetilde{W})\mid \omega \kern 0.045em;\widetilde{W}\Bigr),
\end{equation}
that is, the conditional covariance (over the signal variable) computed under the market maker's posterior on $S$
given $\omega$ (with $\widetilde{W}$ fixed).
Here $\eta(x,\,\cdot\,)$ is the signal-wise payoff profile of security $x$, and
$\Pi_{insider}(v,\,\cdot\,;\widetilde{W})$ is the insider’s noise-adjusted projection coefficient of $v$ against the
signal-contingent profile $\widetilde{W}(\,\cdot\,,s)$ in the $\sigma^{-2}$-weighted pairing, as in
Definition~\ref{def: overlap measures}(ii).
Taking expectation over $\omega$ under $\mathbb{P}_{\scriptscriptstyle W}$ then yields $\partial \! \kern 0.08em J_{\scriptscriptstyle \! K}(W)(v)$.

The conditional covariance in \eqref{eqn: price impact of v} has a direct finance meaning: it is the \emph{information footprint} of $v$ for security $x$---how much information a marginal $v$-trade reveals about $x$’s payoff.
Accordingly, the price impact \eqref{eqn: price impact of v} is high precisely when, under the market maker’s belief, trades in direction $v$ are viewed as highly informative about $x$’s payoff.

\paragraph{Pairwise Cross-Impact}
When $v$ represents a unit position in security $y$, \eqref{eqn: price impact of v} specializes to the
cross price impact of $y$-order flow on the price of $x$.
To make this rigorous, fix $y\in(\underaccent{\bar}{x},\bar{x})$ and choose an approximate identity
$(v_n)_{n\ge 1}\subset C^\delta([\underaccent{\bar}{x},\bar{x}])$ centered at $y$
(e.g., $v_n\ge 0$, $\int v_n(z)\,dz=1$, and $v_n \Rightarrow \delta_y$ weakly as measures).
Then, by Definition~\ref{def: overlap measures}(ii),
\[
\Pi_{\mathrm{insider}}(v_n,\,\cdot\,;\widetilde W)
=\int_{\underaccent{\bar}{x}}^{\bar{x}} \frac{\widetilde W(z,\cdot)}{\sigma^2(z)}\,v_n(z)\,dz
\;\longrightarrow\;
\frac{\widetilde W(y,\cdot)}{\sigma^2(y)}
\qquad\text{in }L^2(S,\mu),
\]
where the convergence is the standard approximation-of-identity limit, under our standing regularity assumptions.
We may therefore extend $\Pi_{\mathrm{insider}}(\,\cdot\,,\,\cdot\,;\widetilde W)$ to the unit-mass trade $\delta_y$ by setting
\begin{equation}
\label{eqn:Pi_insider_delta_y}
\Pi_{\mathrm{insider}}(\delta_y,\,\cdot\,;\widetilde W)=\frac{\widetilde W(y,\cdot)}{\sigma^2(y)}.
\end{equation}
In the options interpretation, $\delta_y$ corresponds to the familiar vanishing-width ``spike'' localized at strike $y$; see, e.g., \cite{hull2003options}.

\begin{corollary}[Cross Price Impact]
\label{cor: def of price impact}
Let $x\in[\underaccent{\bar}{x},\bar{x}]$ and $y\in(\underaccent{\bar}{x},\bar{x})$, and write $\frac{\partial}{\partial W(y)}$ for the directional derivative corresponding to a unit trade at $y$
(realized via any approximate identity converging to $\delta_y$ as above).
Under market maker belief $\widetilde W$, the cross price impact of a unit trade in $y$ on the price of $x$,
conditional on aggregate order flow $\omega$, is
\begin{equation}
\label{eqn: conditional price impact y on x}
\frac{\partial}{\partial W(y)}\,P(x,\omega\kern0.045em;\widetilde W)
=
\frac{1}{\sigma^2(y)}\,
\Cov\!\Bigl(\eta(x,\cdot),\,\widetilde W(y,\cdot)\mid \omega\kern0.045em;\widetilde W\Bigr),
\end{equation}
where the conditional covariance is taken over the signal variable under the market maker's posterior on $S$
given $\omega$.
Consequently, for the expected pricing kernel
$\overline P(x,W;\widetilde W):=\mathbb{E}^{\mathbb{P}_{\scriptscriptstyle W}}\!\bigl[P(x,\omega\kern0.045em;\widetilde W)\bigr]$,
the expected cross price impact is
\begin{equation}
\label{eqn: price impact y on x}
\frac{\partial}{\partial W(y)}\,\overline P(x,W;\widetilde W)
=
\mathbb{E}^{\mathbb{P}_{\scriptscriptstyle W}}\!\left[
\frac{\partial}{\partial W(y)}\,P(x,\omega\kern0.045em;\widetilde W)
\right].
\end{equation}
\end{corollary}

\begin{proof}
By the covariance representation~\eqref{eqn: price impact of v}
\[
\frac{\partial}{\partial \varepsilon}P\!\left(x,\omega\kern0.045em;\widetilde W+\varepsilon v_n\right)\Big|_{\varepsilon=0}
=
\Cov\!\left(\eta(x,\cdot),\,\Pi_{\mathrm{insider}}(v_n,\cdot;\widetilde W)\mid \omega\kern0.045em;\widetilde W\right).
\]
for any approximate identity $(v_n)_{n\ge1}$ with $v_n\Rightarrow \delta_y$.
Letting $n\to\infty$ and passing to $L^2$-limit~\eqref{eqn:Pi_insider_delta_y} yields \eqref{eqn: conditional price impact y on x}. Taking expectations under
$\mathbb{P}_{\scriptscriptstyle W}$ gives \eqref{eqn: price impact y on x}.
\end{proof}

Corollary~\ref{cor: def of price impact} generalizes the standard ``informativeness $\times$ liquidity'' principle from market microstructure.
The (posterior) covariance between
$\widetilde{W}(y,\,\cdot\,)$ and $\eta(x,\,\cdot\,)$ captures what the market maker infers about $x$'s payoff
from order flow in $y$: if it is positive, then high demand for $y$ is viewed as evidence that $x$ has a high payoff,
so buys of $y$ raise the price of $x$, and vice versa. The factor $\sigma^{-2}(y)$ is the
usual liquidity effect: more noise trading in $y$ makes $y$-order flow less informative and therefore
attenuates this cross-impact.
The next example gives a simple, concrete illustration.

\begin{example}
\label{example: remark on price impact from insider perspective}
Fix two distinct signals $s\neq s'$. Suppose security $y$ pays only under $s$ (i.e., $\eta(y,t)=0$ for $t\neq s$) and the market maker believes (plausibly) that the insider buys $y$ only under $s$ (i.e., $\widetilde W(y,t)=0$ for $t\neq s$ with $\widetilde W(y,s)>0$). 
So the market maker views buy order flow in $y$ as evidence in favor of signal $s$.

\begin{enumerate}[label=(\roman*)]
\item If security $x$ also pays only under $s$ (i.e., $\eta(x,t)=0$ for $t\neq s$), then
\[
\Cov \bigl(\eta(x,\cdot),\widetilde W(y,\cdot)\mid \omega;\widetilde W\bigr)\ge 0
\] and the cross-impact of $y$ on $x$ is positive: buying $y$ makes $s$ more likely---under the market maker's belief---and therefore raises the price of $x$.

\item If instead $x$ pays only under $s'$ (i.e., $\eta(x,t)=0$ for $t\neq s'$), then the cross-impact is negative: buying $y$ shifts probability mass away from $s'$, lowering the price of $x$.
\end{enumerate}
\end{example}

\subsection{Viability}
\label{sec: no arb}

Viability is a general necessary condition for equilibrium---at equilibrium prices, agents' optimization problems must be well-posed.
In reduced-form models, where prices are taken as \emph{given}, no-arbitrage is typically imposed as a pricing restriction (e.g., existence of an equivalent martingale measure).
Here, as in standard finance models, prices are equilibrium \emph{outcomes}, jointly determined with agents' optimal actions.
Accordingly, the relevant equilibrium restriction is the stronger viability condition: the insider's problem must be well-posed, i.e., there is no portfolio that can be scaled to yield unbounded expected profit (equivalently, utility) at the prevailing prices.
For a textbook discussion of the arbitrage-viability-equilibrium link, we refer to \cite[Ch.~1]{duffie2001dynamic}.

We now characterize this viability condition.
Again, the projection profile $\Pi_{\mathrm{insider}}(W,\,\cdot\,;\widetilde W)$ from Definition~\ref{def: overlap measures}(ii) naturally occurs.

\begin{definition}
\label{def: zero price impact subspace}
The \textbf{zero price impact subspace} under market maker belief $\widetilde W$ is the linear subspace
of the insider's admissible portfolios given by
\[
\mathcal{V}_0(\widetilde W)
:=\Bigl\{\,W\in C^\delta([\underaccent{\bar}{x},\bar{x}],\R) \;\Big|\;\Pi_{\mathrm{insider}}(W,\cdot;\widetilde W)=0\,\Bigr\}.
\]
\end{definition}

If some portfolio $W\in\mathcal V_0(\widetilde W)$ yielded strictly positive (expected) utility to the insider, he could scale it up with no price impact to obtain unbounded utility, so his optimization problem would be ill-posed. 
This is Theorem~\ref{thm: informed trader 2}.

\begin{theorem}[No Free Lunch with Zero Price Impact]
\label{thm: informed trader 2}
Fix a market maker belief $\widetilde W$ and a signal $s$. A necessary condition for the insider's
problem~\eqref{eqn: informed trader's problem} to admit an optimal portfolio is that
\[
J(W;\widetilde W,s)=0 \qquad \text{for all } W\in \mathcal{V}_0(\widetilde W).
\]
Otherwise, the insider can obtain unbounded utility by scaling up a zero price impact portfolio.
\end{theorem}

\begin{proof}
Let $W$ be an admissible portfolio. By linearity of
$W\mapsto \Pi_{\mathrm{insider}}(W,\,\cdot\,;\widetilde W)$ (Definition~\ref{def: overlap measures}(ii)),
\[
\Pi_{\mathrm{insider}}(W+\alpha V,\,\cdot\,;\widetilde W)
=\Pi_{\mathrm{insider}}(W,\,\cdot\,;\widetilde W)
\qquad \text{for all } \alpha\in\R \text{ and } V\in\mathcal V_0(\widetilde W).
\]
By its pushed-forward representation~\eqref{eqn: pushed forward AD bar}, the expected price functional
$W\mapsto \overline P(\,\cdot\,,W;\widetilde W)$ depends on $W$ only through $\Pi_{\mathrm{insider}}(W,\,\cdot\,;\widetilde W)$.
Hence, for any admissible $W$ and $V\in\mathcal V_0(\widetilde W)$,
\[
\overline P(\,\cdot\,,W+\alpha V;\widetilde W)
=\overline P(\,\cdot\,,W;\widetilde W)
\qquad \text{for all }\alpha\in\R.
\]
Fix $s\in S$ and set $\overline P_W:=\overline P(\,\cdot\,,W;\widetilde W)$. Using the definition of the expected utility functional $J$,
\[
J(W+\alpha V;\widetilde W,s)
=\int_{\underaccent{\bar}{x}}^{\bar{x}}\bigl(\eta(x,s)-\overline P_W(x)\bigr)\bigl(W(x)+\alpha V(x)\bigr)\,dx
=J(W;\widetilde W,s)+\alpha\,J(V;\widetilde W,s).
\]
If there exists $V\in\mathcal V_0(\widetilde W)$ with $J(V;\widetilde W,s)\neq 0$, then choosing $\alpha$ of the appropriate sign
and letting $|\alpha|\to\infty$ makes $J(W+\alpha V;\widetilde W,s)$ arbitrarily large. Therefore, the value of
problem~\eqref{eqn: informed trader's problem} is $+\infty$ and no maximizer can exist. It follows that a necessary
condition for an optimal portfolio is
\[
J(V;\widetilde W,s)=0 \qquad \text{for all } V\in\mathcal V_0(\widetilde W),
\]
as claimed.
\end{proof}

\section{Canonical Game}
\label{sec: finite dim reduction}

In this section we carry out the canonical game reduction.
We begin by noting that, in equilibrium, the insider's zero-payoff portfolios must coincide with the zero price impact portfolios.
Indeed, under the correct equilibrium belief (Definition~\ref{def: equilibrium}), a zero-payoff portfolio is payoff-irrelevant and therefore has zero price impact.
Conversely, by Theorem~\ref{thm: informed trader 2}, any zero price impact portfolio must yield zero payoff; otherwise it could be scaled up without moving prices, yielding unbounded utility.

This lets us characterize the portfolios that are payoff-relevant for the insider in equilibrium.
Throughout, define
\[
\langle f,g\rangle_{\sigma}
:=\int_{\underaccent{\bar}{x}}^{\bar{x}} \frac{f(x)g(x)}{\sigma^2(x)}\,dx,
\qquad 
\|f\|_\sigma^2:=\langle f,f\rangle_\sigma,
\]
and let $\mathcal H_\sigma:=L^2([\underaccent{\bar}{x},\bar{x}],\sigma^{-2}(x)\,dx)$ be the associated Hilbert space.
By Assumption~\ref{assumption: MM's posterior}(iv), $\|\cdot\|_\sigma$ is equivalent to the usual $L^2(dx)$ norm; consequently,
orthogonality and closed linear spans agree in $\mathcal H_\sigma$ and in $L^2(dx)$.

\begin{proposition}
\label{cor: finite S 1}
Without loss of generality, the insider's equilibrium portfolios $W^*(\,\cdot\,,s)$, $s\in S$, may be restricted to
\[
\mathcal M:=\overline{\operatorname{span}}\{\eta(\,\cdot\,,s):s\in S\},
\]
where the closure is taken in $L^2(dx)$ (equivalently, in $\mathcal H_\sigma$).
\end{proposition}

\begin{proof}
Fix \(s\in S\). Decompose any admissible portfolio as
\[
W(\,\cdot\,,s)=W_{\mathcal M}(\,\cdot\,,s)+W_{\mathcal M^\perp}(\,\cdot\,,s),
\qquad
W_{\mathcal M}(\,\cdot\,,s)\in\mathcal M,\ \ W_{\mathcal M^\perp}(\,\cdot\,,s)\in\mathcal M^\perp,
\]
where \(\mathcal M^\perp\) denotes the orthogonal complement in \(L^2(dx)\) (equivalently in \(\mathcal H_\sigma\)).
By definition of \(\mathcal M\), for every \(s'\in S\),
\[
\int_{\underaccent{\bar}{x}}^{\bar x} W_{\mathcal M^\perp}(x,s)\,\eta(x,s')\,dx=0,
\]
so \(W_{\mathcal M^\perp}(\,\cdot\,,s)\) has identically zero payoff across signal realizations. In equilibrium,
zero-payoff portfolios coincide with zero price impact portfolios (as noted above), hence adding
\(W_{\mathcal M^\perp}(\,\cdot\,,s)\) does not change equilibrium prices and does not affect expected utility.
Therefore, replacing \(W(\,\cdot\,,s)\) by its projection \(W_{\mathcal M}(\,\cdot\,,s)\) is without loss of generality.
\end{proof}

\medskip

\paragraph{Coefficient Representation}
We now make the reduction in Proposition~\ref{cor: finite S 1} explicit in coefficient space.
Define the coefficient synthesis operator $T\colon L^2(S,\mu)\to\mathcal H_\sigma$ by
\begin{equation}
\label{eqn: T_operator_definition}
(T\theta)(x):=\int_S \theta(s)\,\eta(x,s)\,\mu(ds), \qquad \theta\in L^2(S,\mu).
\end{equation}
with the integral understood in the Bochner sense.
Then $\mathrm{ran}(T)\subseteq \mathcal M$, and hence $\mathcal M=\overline{\mathrm{ran}(T)}$.
So every payoff-relevant portfolio $W\in\mathcal M$ admits a representation $W=T\theta$ for some coefficient
$\theta\in L^2(S,\mu)$.
Similarly, we represent a given market maker belief $\bigl(\widetilde W(\,\cdot\,,s)\bigr)_{s\in S}\subset\mathcal M$ by choosing a measurable map
$s\mapsto \tilde\theta(s)\in L^2(S,\mu)$ such that
\[
\widetilde W(\cdot,s)=T\tilde\theta(s), \qquad s\in S.
\]

\begin{definition}
\label{def: information intensity matrix} 
Fix a choice of positive square root ${\bf L}$ of $T^*T$, and call it the \textbf{information intensity operator}.
Equivalently, ${\bf L}^2$ is the integral operator
\[
({\bf L}^2 f)(s)=\int_S k(s,t)\,f(t)\,\mu(dt),
\qquad f\in L^2(S,\mu),
\]
with kernel
\begin{equation}
\label{eqn: info intensity kernel}
k(s,t):=\langle \eta(\,\cdot\,,s),\eta(\,\cdot\,,t)\rangle_{\sigma}
=\int_{\underaccent{\bar}{x}}^{\bar{x}}
\frac{\eta(x,s)\eta(x,t)}{\sigma^2(x)}\,dx.
\end{equation}
\end{definition}

${\bf L}$ is the multi-asset operator analogue of Kyle’s information-intensity parameter
$\sigma_v/\sigma_\varepsilon$ \cite{kyle1985continuous} (recalled in~\eqref{eqn: Kyle equil}).
We compare information intensity using the (partial) Loewner order on positive operators.
In this sense, larger ${\bf L}$ leads to more informative order flow in equilibrium across assets.

\paragraph{Orthonormalization/Whitening}
Let $\mathcal H_k$ be the reproducing kernel Hilbert space on $S$ with kernel $k(\,\cdot\, , \,\cdot\,)$ from \eqref{eqn: info intensity kernel},
\[
\mathcal H_k:=\overline{\mathrm{ran}({\bf L})},
\qquad
\langle {\bf L}\theta_1,{\bf L}\theta_2\rangle_{\mathcal H_k}
:=\langle \theta_1,\theta_2\rangle_{L^2(S,\mu)}
\ \ \text{on }\mathrm{ran}({\bf L}),
\]
extended by completion. By the reproducing property, evaluation at each signal $s$ is represented by the
kernel section
\begin{equation}
\label{eqn: reproducing}
f(s)=\langle f,k(\,\cdot\,,s)\rangle_{\mathcal H_k},
\qquad f\in\mathcal H_k,\ s\in S.
\end{equation}
(When $S$ is finite, $\mathcal H_k \simeq \mathbb{R}^{|S|}$, and the sections $k(\,\cdot\,,s)$ may be identified with the columns of the Gram matrix $k(\,\cdot\,,\cdot\,)$ implementing coordinate evaluation.)
Define the \emph{whitening transformation}
\begin{equation}
\label{eqn: change of basis}
\theta \mapsto \hat{\theta} := {\bf L}\theta \in \mathcal H_k .
\end{equation}
For market maker belief coefficients $s\mapsto \tilde\theta(s)$, set
\begin{equation}
\label{eqn: fin dim notation for W for all signals}
\hat\theta(s):={\bf L}\tilde\theta(s)\in\mathcal H_k,
\qquad s\in S,
\end{equation}
and the associated whitened belief operator
\begin{equation}
\label{eqn: Dhat operator}
(\widehat\Theta f)
:=\int_S f(s)\,\hat\theta(s)\,\mu(ds),
\qquad f\in \mathcal H_k,
\end{equation}
with the integral understood in the Bochner sense.

\begin{remark}
\label{rmk: normalization of information intensity}
In finance terms, the whitening transformation~\eqref{eqn: change of basis} normalizes the trading game by its information intensity---one unit in a whitened signal coordinate corresponds to one unit of information exposure to that signal.
Mathematically, this unit exposure to a signal $s$ is represented by the evaluation functional $\mathcal H_k\ni f \mapsto f(s)$.
\end{remark}

\begin{theorem}
\label{thm: canonical game}
Under the whitening transformation \eqref{eqn: change of basis}, the Bayesian trading game between the insider and the market maker is isomorphic to the following \textbf{canonical game}:

\begin{itemize}
\item For each signal $u\in S$ there is a traded \textbf{pseudo-security}. Conditional on the realized signal $s\in S$, pseudo-security $u$ pays $\mathbf 1_{\{u=s\}}$. Equivalently, an order $\hat\theta\in\mathcal H_k$ has realized payoff $\hat\theta(s)$.

\item The insider observes $s$. The market maker's prior on $S$ is $\mu$.

\item The insider submits an order $\hat\theta\in\mathcal H_k$ for the pseudo-securities. The market maker receives order flow
\[
\hat\omega = \hat\theta + \widehat X,
\]
where the noise trades (across the pseudo-securities) $\widehat X$ is an isonormal Gaussian element on $\mathcal H_k$, i.e., $\{\widehat X(f):f\in\mathcal H_k\}$ is centered Gaussian with covariance
$\E[\widehat X(f)\widehat X(g)]=\langle f,g\rangle_{\mathcal H_k}$. 

\item The market maker's belief is the measurable map $s\mapsto \hat\theta(s)$ (equivalently, the operator $\widehat\Theta$ as in \eqref{eqn: Dhat operator}).
Given $\hat\omega$, his posterior on $(S,\mathcal S)$ is
\begin{equation}
\label{eqn: pi_1 finite S new}
\hat{\pi}_1(ds\mid \hat\omega \kern 0.045em ; \widehat{\Theta})
=
\frac{
\exp\!\left( \hat\omega(\hat\theta(s)) - \frac12 \|\hat\theta(s)\|_{\mathcal H_k}^2 \right)\mu(ds)
}{
\int_S \exp\!\left( \hat\omega(\hat\theta(u)) - \frac12 \|\hat\theta(u)\|_{\mathcal H_k}^2 \right)\mu(du)
},
\end{equation}
where $\hat\omega(h):=\langle h,\hat\omega\rangle_{\mathcal H_k}$ whenever $\hat\omega\in\mathcal H_k$
(by Riesz representation).

\item Conditional on observing $s\in S$ and given the market maker's belief $\widehat\Theta$, the insider maximizes expected profit from trading the pseudo-securities:
\begin{equation}
\label{eqn: further transformed insider's problem for finite S case}
\max_{\hat{\theta} \, \in \, \mathcal H_k} \,
\hat\theta(s)
-
\int_S \hat\theta(u)\,
\hat{\pi}_1(du \mid \hat\theta + \widehat X \kern 0.045em ; \widehat{\Theta})
\ \equiv\ 
\max_{\hat{\theta} \, \in \, \mathcal H_k} J(\hat{\theta} \kern 0.045em ; \widehat{\Theta}, s),
\end{equation}
where $\hat{\pi}_1(\cdot\mid \hat\theta+\widehat X;\widehat\Theta)$ denotes \eqref{eqn: pi_1 finite S new} evaluated at the random observation $\hat\omega=\hat\theta+\widehat X$.
\end{itemize}
\end{theorem}

\begin{proof}
By Proposition~\ref{cor: finite S 1}, we may restrict the insider to $\mathcal M=\overline{\mathrm{ran}(T)}$ and represent his choice as $W=T\theta$.

\smallskip
\emph{Step 1: Sufficient statistic in coefficient space.}
For any given belief $s \mapsto \widetilde W(\,\cdot\,,s)=T\tilde\theta(s)$, the market maker's sufficient statistic is the projection profile $\Pi_{mm}(\omega, \,\cdot\,;\widetilde W)$ (Definition~\ref{def: overlap measures} and Theorem~\ref{thm: Market Maker's Posterior}).  
As a process indexed by $s\in S$, $\Pi_{mm}(\omega, \,\cdot\,;\widetilde W)$ is jointly Gaussian with covariance kernel
\begin{equation}
\label{eqn: cov_pi_mm}
\mathrm{Cov}\bigl(\Pi_{mm}(\omega,s;\widetilde W),\Pi_{mm}(\omega,t;\widetilde W)\bigr)
=
\langle \widetilde W(\cdot,s),\widetilde W(\cdot,t)\rangle_\sigma
=
\langle \tilde\theta(s), {\bf L}^2 \tilde\theta(t)\rangle_{L^2(S,\mu)}.
\end{equation}
Similarly, the conditional mean shift induced by insider choice $W=T\theta$ equals
$\Pi_{insider}(W,s;\widetilde W)=\langle \theta,{\bf L}^2 \tilde\theta(s)\rangle_{L^2(S,\mu)}$.

\smallskip
\emph{Step 2: Whitening.}
Under the whitening transformation \eqref{eqn: change of basis}, the covariance~\eqref{eqn: cov_pi_mm} becomes the $\mathcal H_k$ inner product
$\langle \hat\theta(s),\hat\theta(t)\rangle_{\mathcal H_k}$, and the mean shift becomes $\langle \hat\theta,\hat\theta(s)\rangle_{\mathcal H_k}$.
Thus, the sufficient statistic $\Pi_{mm}(\omega, \,\cdot\,;\widetilde W)$ is equivalent in law to an isonormal Gaussian observation
$\hat\omega=\hat\theta+\widehat X$ on $\mathcal H_k$.

\smallskip
\emph{Step 3: Posterior and payoff.}
The market maker posterior in the original game is exponential in the sufficient statistic (Theorem~\ref{thm: Market Maker's Posterior}). After whitening, this yields \eqref{eqn: pi_1 finite S new}.
Finally, by the reproducing property~\eqref{eqn: reproducing}, the payoff functional associated with signal $s$ is evaluation of $\hat\theta$ at $s$.
Substituting the posterior pricing rule gives the objective \eqref{eqn: further transformed insider's problem for finite S case}.
\end{proof}

\paragraph{Equilibrium in the Canonical Game}
An equilibrium of the canonical game is given by a (measurable) trading strategy for the pseudo-securities
\[
s \mapsto \theta^*(s)\in\mathcal H_k
\quad\text{(equivalently, by the operator $\Theta^*$ defined as in \eqref{eqn: Dhat operator})},
\]
such that, for $\mu$-a.e.\ $s\in S$, the order $\theta^*(s)$ is a best response to the market maker's
belief $\Theta^*$ in the insider's problem \eqref{eqn: further transformed insider's problem for finite S case}, i.e.,
\begin{equation}
\label{eqn: transformed equil def for finite S case}
\theta^*(s)\in \argmax_{\theta \,\in\, \mathcal H_k} J(\theta \kern 0.045em ; \Theta^*, s).
\end{equation}

\begin{corollary}
\label{cor: equil isomorphism}
Under the isomorphism between the two games, an equilibrium $\Theta^*$ of the canonical game induces an equilibrium of the original trading game (Definition~\ref{def: equilibrium}) with insider trading strategy
\begin{equation}
\label{eqn: canonical form informed demand}
W^*(\,\cdot\,, s)
=
\int_S \beta(s)(u)\,\eta(\,\cdot\,,u)\,\mu(du),
\qquad
\text{where }\ \beta(s) := {\bf L}^{-1}\theta^*(s)\ \text{ for }\mu\text{-a.e.\ }s\in S.
\end{equation}
Here ${\bf L}^{-1}$ denotes the Moore--Penrose inverse of ${\bf L}$ (which is the usual inverse when ${\bf L}$ is injective).
We call the family $\bigl(\beta(s)\bigr)_{s\in S}$ (equivalently ${\bf L}^{-1}\Theta^*$) the \textbf{canonical form} of the original-game equilibrium \eqref{eqn: canonical form informed demand}.
\end{corollary}

\begin{proof}
Let $\theta^*(s)\in\mathcal H_k$ denote the canonical-game equilibrium order at signal $s$.
For $\mu$-a.e.\ $s\in S$, set
\[
\beta(s):={\bf L}^{-1}\theta^*(s)\in L^2(S,\mu),
\qquad
W^*(\,\cdot\,,s):=T\beta(s)
\ \ \text{(see \eqref{eqn: T_operator_definition}).}
\]
To invoke Definition~\ref{def: equilibrium}, we work with an admissible representative of
$W^*(\,\cdot\,,s)$ in $C^\delta([\underaccent{\bar}{x},\bar{x}],\R)$, chosen so that
$s\mapsto W^*(\,\cdot\,,s)$ is continuous %as in Assumption~\ref{assumption: MM's posterior}(iii), 
by a
standard continuous selection theorem (e.g., \cite{Michael1956}).
Then the integral representation of $W^*(\,\cdot\,, s)$ in \eqref{eqn: canonical form informed demand} follows immediately.
\end{proof}

\paragraph{Insider FOC in the Canonical Game}
Let $p\in L^{1}(S,\mu)$ denote the (random) Radon--Nikodym density of the posterior in
\eqref{eqn: pi_1 finite S new} with respect to $\mu$, i.e.
\[
\hat\pi_{1}(du\mid \hat\omega;\widehat \Theta)=p(u)\,\mu(du).
\]
Define the (random) posterior mean element $\kappa_{p}\in\mathcal H_{k}$ by
\[
\kappa_{p}:=\int_{S} k(\,\cdot\,,u)\,p(u)\,\mu(du),
\]
and define the (random) posterior covariance operator $\mathcal C_{p}:\mathcal H_{k}\to\mathcal H_{k}$
via the bilinear form
\begin{equation}
\label{eqn: Cp_def}
\langle f,\mathcal C_{p} g\rangle_{\mathcal H_{k}}
:=
\int_{S} f(u)\,g(u)\,p(u)\,\mu(du)
-
\Bigl(\int_{S} f(u)\,p(u)\,\mu(du)\Bigr)
\Bigl(\int_{S} g(u)\,p(u)\,\mu(du)\Bigr),
\qquad f,g\in\mathcal H_{k}.
\end{equation}

Proposition~\ref{prop: canonical_game_FOC} below is the isomorphic counterpart to Theorem~\ref{thm: informed trader 1}.

\begin{proposition}[Insider FOC in the Canonical Game]
\label{prop: canonical_game_FOC}
In the canonical game, suppose $\hat\theta\in\mathcal H_k$ solves the
insider's problem \eqref{eqn: further transformed insider's problem for finite S case},
conditional on $s \in S$ and under market maker belief $\widehat\Theta$.
Then $\hat\theta$ satisfies the first-order condition
\begin{equation}
\label{eqn: transformed Bayesian game FOC}
k(\,\cdot\,,s)
-
\underbrace{
\mathbb{E}^{(\hat{\theta}; \widehat{\Theta})}[\kappa_p]\,
\vphantom{\widehat{\Theta}\Bigl(\mathbb{E}^{(\hat{\theta}; \widehat{\Theta})}[\mathcal C_p]\Bigr)\hat{\theta}}
}_{\text{AD term}}
-
\underbrace{
\widehat{\Theta}
\Bigl(
\mathbb{E}^{(\hat{\theta}; \widehat{\Theta})}[\mathcal C_p]
\Bigr)
\hat{\theta}
}_{\text{price impact term}}
=0.
\end{equation}
where $\mathbb{E}^{(\hat{\theta}; \widehat{\Theta})}[\cdot]$ denotes expectation with respect to the law of
$\hat\omega=\hat\theta+\widehat X$ induced by the choice $\hat\theta$ under the belief $\widehat\Theta$.
\end{proposition}

\begin{proof}
Fix $s\in S$ and a market maker belief $\widehat\Theta$. For a candidate order $\hat\theta\in\mathcal H_k$,
let $p_{\hat\theta}\in L^1(S,\mu)$ denote the (random) posterior density induced by the observation
$\hat\omega=\hat\theta+\widehat X$ under $\widehat\Theta$,
and write
\[
J(\hat\theta;\widehat\Theta,s)
=\hat\theta(s)-\E^{(\hat\theta;\widehat\Theta)}
\Bigl[\int_S \hat\theta(u)\,p_{\hat\theta}(u)\,\mu(du)\Bigr].
\]
Let $h\in\mathcal H_k$ and set $\hat\theta_\varepsilon:=\hat\theta+\varepsilon h$.

\medskip
\noindent\emph{Step 1: The marginal payoff term.}
By the reproducing property,
\begin{equation}
\label{eqn: theta_hat_eps_derivative}
\frac{d}{d\varepsilon}\Big|_{\varepsilon=0}\hat\theta_\varepsilon(s)
= h(s)=\langle h,k(\cdot,s)\rangle_{\mathcal H_k}.
\end{equation}

\medskip
\noindent\emph{Step 2: The marginal cost term.}
Using again the reproducing property,
\begin{equation}
\label{eqn: product rule theta p}
\frac{d}{d\varepsilon}\Big|_{\varepsilon=0}\int_S \hat\theta_\varepsilon(u)\,p_{\hat\theta_\varepsilon}(u)\,\mu(du)
=
\int_S h(u)\,p_{\hat\theta}(u)\,\mu(du)
+
\int_S \hat\theta(u)\,\dot p_{\hat\theta}[h](u)\,\mu(du).
\end{equation}
where $\dot p_{\hat\theta}[h]$ denotes the directional derivative of $p_{\hat\theta}$ at $\hat\theta$ in direction $h$.
(The interchange of differentiation with the $\mu$-integral and the outer expectation is justified by the standing
integrability/regularity assumptions and standard dominated-convergence arguments for exponential families.)

The first term in \eqref{eqn: product rule theta p} satisfies
\[
\int_S h(u)\,p_{\hat\theta}(u)\,\mu(du)=\langle h,\kappa_p\rangle_{\mathcal H_k},
\qquad
\kappa_p:=\int_S k(\cdot,u)\,p_{\hat\theta}(u)\,\mu(du),
\]
hence
\begin{equation}
\label{eqn: rkhs representation expected integral}
\E^{(\hat\theta;\widehat\Theta)}\Bigl[\int_S h(u)\,p_{\hat\theta}(u)\,\mu(du)\Bigr]
=\bigl\langle h,\E^{(\hat\theta;\widehat\Theta)}[\kappa_p]\bigr\rangle_{\mathcal H_k}.
\end{equation}

For the second term in \eqref{eqn: product rule theta p}, fix a realization of $\widehat X$ and recall that $p_{\hat\theta}$ is the normalized
Gibbs density
\[
p_{\hat\theta}(u)\ \propto\
\exp\!\Bigl(\hat\omega\bigl(\hat\theta^{\mathrm{MM}}(u)\bigr)-\tfrac12\|\hat\theta^{\mathrm{MM}}(u)\|^2_{\mathcal H_k}\Bigr),
\]
where $u\mapsto \hat\theta^{\mathrm{MM}}(u)$ is the strategy map encoded by the belief $\widehat\Theta$.
Since $\hat\omega$ shifts by $\varepsilon h$ when $\hat\theta$ shifts by $\varepsilon h$, the log-likelihood shifts by
$\varepsilon\langle h,\hat\theta^{\mathrm{MM}}(u)\rangle_{\mathcal H_k}$, and therefore
\[
\dot p_{\hat\theta}[h](u)
=
p_{\hat\theta}(u)\Bigl(\langle h,\hat\theta^{\mathrm{MM}}(u)\rangle_{\mathcal H_k}
-
\int_S \langle h,\hat\theta^{\mathrm{MM}}(v)\rangle_{\mathcal H_k}\,p_{\hat\theta}(v)\,\mu(dv)\Bigr).
\]
Substituting and rearranging gives (for each realization of $p_{\hat\theta}$)
\begin{align*}
\int_S \hat\theta(u)\,\dot p_{\hat\theta}[h](u)\,\mu(du)
&=
\int_S \hat\theta(u)\,p_{\hat\theta}(u)\,\langle h,\hat\theta^{\mathrm{MM}}(u)\rangle_{\mathcal H_k}\,\mu(du) \\
&\quad-
\Bigl(\int_S \hat\theta(u)\,p_{\hat\theta}(u)\,\mu(du)\Bigr)
\Bigl(\int_S \langle h,\hat\theta^{\mathrm{MM}}(v)\rangle_{\mathcal H_k}\,p_{\hat\theta}(v)\,\mu(dv)\Bigr).
\end{align*}
By the definition of the covariance operator $\mathcal C_p$ in \eqref{eqn: Cp_def} and of the belief operator
$\widehat\Theta$, the right-hand side equals
\[
\bigl\langle h,\,\widehat\Theta(\mathcal C_p\hat\theta)\bigr\rangle_{\mathcal H_k},
\]
and hence, by linearity and boundedness of $\widehat\Theta$,
\begin{equation}
\label{eqn: expected_integral_inner_product_identity}
\E^{(\hat\theta;\widehat\Theta)}\Bigl[\int_S \hat\theta(u)\,\dot p_{\hat\theta}[h](u)\,\mu(du)\Bigr]
=
\bigl\langle h,\,\widehat\Theta\bigl(\E^{(\hat\theta;\widehat\Theta)}[\mathcal C_p]\bigr)\hat\theta\bigr\rangle_{\mathcal H_k}.
\end{equation}

\medskip
\noindent\emph{Step 3.}
Combining \eqref{eqn: theta_hat_eps_derivative}, \eqref{eqn: rkhs representation expected integral}, and \eqref{eqn: expected_integral_inner_product_identity} above yields
\[
dJ(\hat\theta)[h]
=
\Bigl\langle h,\,
k(\cdot,s)-\E^{(\hat\theta;\widehat\Theta)}[\kappa_p]
-\widehat\Theta\bigl(\E^{(\hat\theta;\widehat\Theta)}[\mathcal C_p]\bigr)\hat\theta
\Bigr\rangle_{\mathcal H_k}.
\]
Since $h\in\mathcal H_k$ is arbitrary, the claimed identity \eqref{eqn: transformed Bayesian game FOC} follows from the Riesz representation theorem.
\end{proof}

\section{Equilibrium}
\label{sec: symmetric equilibrium}

We first consider equilibrium in the canonical game of Theorem~\ref{thm: canonical game}, formulated on the whitened coefficient space $(\mathcal H_k,\langle\cdot,\cdot\rangle_{\mathcal H_k})$ over the Borel signal space $(S,\mathcal S,\mu)$.
Let
\begin{equation}
\label{eqn: kernel mean element}
\bar{k} \;:=\; \int_S k(\,\cdot\,,s)\,\mu(ds) \in \mathcal H_k
\end{equation}
denote the (nonzero) kernel mean element, and let $P:\mathcal H_k\to\mathcal H_k$ be the rank-one orthogonal
projection onto $\linspan\{\bar{k}\}$. Define the \emph{centering operator} $Q$ as the orthogonal projection onto $\bar{k}^\perp$,
\begin{equation}
\label{eqn: centering operator Q}
Q \;:=\; \mathrm{Id}_{\mathcal H_k}-P.
\end{equation}

\paragraph{Equilibrium Ansatz}
In the canonical game, payoffs reduce to evaluation at the realized $s\in S$ and the noise is isotropic in
$\mathcal H_k$. Consequently, the environment is invariant to (measure-preserving) relabelings of the signal space,
which naturally suggests restricting attention to signal-equivariant equilibria. Accordingly, we consider candidate equilibria in which the strategy lies along
the \emph{centered} representer of evaluation, i.e., the component of $k(\,\cdot\,,s)$ orthogonal to the
prior-mean element $\bar k$.

Equivalently, the insider adopts a \emph{market-neutral long--short position}: he goes long the pseudo-security
indexed by the realized signal and finances it by taking an offsetting short position across the remaining
pseudo-securities, so that his net exposure to the prior mean $\bar k$ is zero.
This leads to the following one-parameter equilibrium ansatz:
\begin{equation}
\label{eqn: equilibrium ansatz for transformed Bayesian game}
\theta^*(s) \;=\; \alpha^*\, Q\,k(\,\cdot\,,s),
\qquad s\in S,
\qquad \text{for some }\alpha^*>0.
\end{equation}
Under this ansatz, $\alpha^*$ is the overall trading scale of the insider's long--short position.
We next show that, within this ansatz, the equilibrium condition reduces to a single scalar equation in $\alpha$.

\begin{theorem}
[Scalar Equilibrium Equation]
\label{thm: scalar_equilibrium_equation}
There exists a function $\Phi:[0,\infty)\to\R$ such that any canonical game equilibrium within the ansatz
class~\eqref{eqn: equilibrium ansatz for transformed Bayesian game} is characterized by a scalar
$\alpha^*>0$ satisfying
$
\Phi(\alpha^*)=0.
$
\end{theorem}

\begin{proof}
Fix $\alpha\ge 0$ and let $\widehat\Theta_\alpha$ denote the belief operator induced, via
\eqref{eqn: Dhat operator}, by the ansatz strategy map
\begin{equation}
\label{eqn: Theta_alpha_strategy_map}
s\mapsto \alpha\,Q\,k(\cdot,s).
\end{equation}

Substituting the ansatz \eqref{eqn: equilibrium ansatz for transformed Bayesian game} into the insider
first-order condition \eqref{eqn: transformed Bayesian game FOC} and applying $Q$ to both sides yields,
for $\mu$-a.e.\ $s\in S$,
\begin{equation}\label{eqn: general-S collated FOC}
Qk(\cdot,s)
\;-\;
Q\,\E_{\widehat\Theta_\alpha}\!\big[\kappa_p\big]
\;-\;
Q\,\widehat\Theta_\alpha\!\Big(\E_{\widehat\Theta_\alpha}\!\big[\mathcal C_p\big]\Big)\,\big(\alpha Qk(\cdot,s)\big)
\;=\;0,
\end{equation}
where $\E_{\widehat\Theta_\alpha}[\cdot]$ denotes expectation under the law of the order-flow observation
$\hat\omega=\alpha Qk(\cdot,s)+\widehat X$ induced by \eqref{eqn: Theta_alpha_strategy_map} under
$\widehat\Theta_\alpha$.

For each $\alpha$, define the residual
\[
F_\alpha(s)
:=
Qk(\cdot,s)
-
Q\,\E_{\widehat\Theta_\alpha}\!\big[\kappa_p\big]
-
Q\,\widehat\Theta_\alpha\!\Big(\E_{\widehat\Theta_\alpha}\!\big[\mathcal C_p\big]\Big)\,\big(\alpha Qk(\cdot,s)\big)
\in \ran(Q)=\bar k^\perp .
\]
By the equivariance of the canonical environment, the map $s\mapsto F_\alpha(s)$ is
signal-equivariant and takes values in $\bar k^\perp$; consequently, for $\mu$-a.e.\ $s$ it is collinear with
$Qk(\cdot,s)$. Hence, there exists a scalar $\Phi(\alpha)\in\R$ such that
\[
F_\alpha(s)=\Phi(\alpha)\,Qk(\cdot,s),
\qquad \mu\text{-a.e.\ } s\in S,
\]
and the proportionality factor depends only on $\alpha$. Specifically, up to a $\mu$-full set, this scalar is given by the normalized projection
\begin{equation}
\label{eqn:Phi-definition}
\Phi(\alpha)
=
\frac{\big\langle Qk(\cdot,s),\,F_\alpha(s)\big\rangle_{\mathcal H_k}}{\|Qk(\cdot,s)\|_{\mathcal H_k}^2}
=
\frac{\Big\langle Qk(\cdot,s),\,
Qk(\cdot,s)-Q\E_{\widehat\Theta_\alpha}[\kappa_p]
-\alpha\,Q\widehat\Theta_\alpha\!\big(\E_{\widehat\Theta_\alpha}[\mathcal C_p]\big)\,Qk(\cdot,s)
\Big\rangle_{\mathcal H_k}}
{\|Qk(\cdot,s)\|_{\mathcal H_k}^2},
\end{equation}
and equivariance implies that this expression is independent of the choice of $s$.

Therefore, \eqref{eqn: general-S collated FOC} is equivalent to
\begin{equation}\label{eqn: matrix equil eqn}
\Phi(\alpha)\,Qk(\cdot,s)=0,
\qquad \mu\text{-a.e.\ } s\in S,
\end{equation}
which holds if and only if $\Phi(\alpha)=0$. Any equilibrium within the ansatz class
\eqref{eqn: equilibrium ansatz for transformed Bayesian game} is thus characterized by a scalar
$\alpha^*>0$ satisfying $\Phi(\alpha^*)=0$.
\end{proof}

\begin{lemma}
\label{lemma: convex price impact general S}
In the canonical game, fix $\alpha>0$ and let $\widehat{\Theta}_\alpha$ denote the market maker's belief operator induced via
\eqref{eqn: Dhat operator} by the candidate strategy map $s\mapsto \alpha\,Q\,k(\,\cdot\,,s)$. 
Then, conditional on each signal $s\in S$, the insider objective
$\hat\theta\mapsto J(\hat\theta;\widehat{\Theta}_\alpha,s)$ is concave on $\mathcal H_k$. 
Consequently, the first-order condition of Proposition~\ref{prop: canonical_game_FOC}
is sufficient for optimality: any $\hat\theta\in\mathcal H_k$ satisfying \eqref{eqn: transformed Bayesian game FOC}
is a best response to $\widehat{\Theta}_\alpha$.
\end{lemma}

\begin{proof}
The second Fr\'echet derivative of $J(\,\cdot\,;\widehat\Theta_\alpha,s)$
in any direction $h\in\mathcal H_k$ can be written in terms of the posterior covariance operator $\mathcal C_p$ as
\[
D^2 J(\hat\theta;\widehat\Theta_\alpha,s)[h,h]
 \;=\; -\,\mathbb E^{(\hat\theta;\widehat\Theta_\alpha)}
\bigl[\;\langle \widehat\Theta_\alpha h,\ \mathcal C_p\,\widehat\Theta_\alpha h\rangle_{\mathcal H_k}\;\bigr].
\]
Since $\mathcal C_p$ is a covariance operator, it is positive semidefinite by construction. Hence, $J(\,\cdot\,;\widehat\Theta_\alpha,s)$ is concave on $\mathcal H_k$.
Therefore, any $\hat\theta$ satisfying the first-order condition \eqref{eqn: transformed Bayesian game FOC}
is a global maximizer, and the condition is sufficient for optimality.
\end{proof}

\begin{lemma}
\label{lemma: claim for MM posterior general S}
In the canonical game, fix $\alpha>0$ and let $\widehat{\Theta}_\alpha$ denote the market maker's belief operator induced via \eqref{eqn: Dhat operator} by the candidate
strategy map $s\mapsto \alpha\, Q\,k(\,\cdot\,,s)$. 
Conditional on the insider observing signal $s\in S$, the market maker's posterior density $p(\,\cdot\,)$ \eqref{eqn: pi_1 finite S new} admits the representation
\begin{equation}
\label{eqn: logistic-normal posterior general S}
p(u)
=
\frac{\exp\!\left(
Z(u)+\alpha^2\left\langle Qk(\cdot,u),Qk(\cdot,s)\right\rangle_{\mathcal H_k}-\frac{\alpha^2}{2}\left\|Qk(\cdot,u)\right\|_{\mathcal H_k}^2
\right)}
{\int_S \exp\!\left(
Z(v)+\alpha^2\left\langle Qk(\cdot,v),Qk(\cdot,s)\right\rangle_{\mathcal H_k}-\frac{\alpha^2}{2}\left\|Qk(\cdot,v)\right\|_{\mathcal H_k}^2
\right)\mu(dv)},
\end{equation}
where $Z(\,\cdot\,)$ is a centered Gaussian process with
$\Cov \bigl(Z(u),Z(v)\bigr)=\alpha^2\left\langle Qk(\,\cdot\,,u),Qk(\cdot,v)\right\rangle_{\mathcal H_k}$.
\end{lemma}

\begin{proof}
Conditional on $s$ and under the belief $\widehat{\Theta}_\alpha$, the observed order flow is
$\hat\omega=\alpha\,Qk(\,\cdot\,,s)+\widehat X$,
where $\widehat X$ is an isonormal Gaussian element on $\mathcal H_k$.
Substituting this expression into \eqref{eqn: pi_1 finite S new} yields
\eqref{eqn: logistic-normal posterior general S} after expanding the Cameron--Martin terms and normalizing.
Since $\widehat X$ is isonormal on $\mathcal H_k$, the process $Z(\,\cdot\,)$ is centered Gaussian with the specified covariance kernel.
\end{proof}

With the above observations about the ansatz~\eqref{eqn: equilibrium ansatz for transformed Bayesian game} in hand, we can now show that it admits an equilibrium.

\begin{theorem}
\label{thm: equilibrium of original trading game}
There exists $\alpha^*>0$ such that the strategy
\begin{equation}
\label{eqn: canonical_equilibrium_theta}
s \mapsto \theta^*(s)=\alpha^*\,Q\,k(\,\cdot\,,s)
\end{equation}
is an equilibrium of the canonical game. Under the isomorphism in Corollary~\ref{cor: equil isomorphism},
the corresponding strategy 
\begin{equation}
\label{eqn: beta_star_definition}
s \mapsto \beta^*(s):={\bf L}^{-1}\theta^*(s).
\end{equation}
is an equilibrium of the original trading game in canonical form.
\end{theorem}

\begin{proof}
Let $\Phi$ be the scalar function defined in \eqref{eqn:Phi-definition} (Theorem~\ref{thm: scalar_equilibrium_equation}). By
Lemma~\ref{lemma: claim for MM posterior general S}, the posterior $p$ induced by the postulated
strategy is a logistic-normal random measure; in particular, $\Phi$ is continuous in $\alpha$ by the
dominated convergence theorem.

Evaluating \eqref{eqn: general-S collated FOC} at $\alpha=0$ gives
$\mathbb E_{\widehat\Theta_0}[\kappa_p]=\bar k$ and the price-impact term vanishes, hence
$\Phi(0)=1>0$.
On the other hand, as $\alpha\to\infty$, the logistic-normal representation in
Lemma~\ref{lemma: claim for MM posterior general S} implies posterior concentration on the true
signal, so $\mathbb E_{\widehat\Theta_\alpha}[\kappa_p]\to k(\cdot,s)$ in $\mathcal H_k$ and
$\mathbb E_{\widehat\Theta_\alpha}[\mathcal C_p]\to 0$. Consequently,
\begin{equation}
\label{eqn:infoadv-vanishes}
\frac{\big|\langle Qk(\cdot,s),\,Q(k(\cdot,s)-\mathbb E_{\widehat\Theta_\alpha}[\kappa_p])\rangle_{\mathcal H_k}\big|}
{\|Qk(\cdot,s)\|_{\mathcal H_k}^2}
\le
\frac{\|Q(k(\cdot,s)-\mathbb E_{\widehat\Theta_\alpha}[\kappa_p])\|_{\mathcal H_k}}{\|Qk(\cdot,s)\|_{\mathcal H_k}}
\longrightarrow 0
\qquad(\alpha\to\infty).
\end{equation}

Moreover, for each realization of $p$, the operator $\mathcal C_p$ is a covariance operator on $\mathcal H_k$
and hence positive semidefinite; therefore $\mathbb E_{\widehat\Theta_\alpha}[\mathcal C_p]$ is
positive semidefinite as well, and in particular
\begin{equation}\label{eqn:priceimpact-nonneg}
\Big\langle Qk(\cdot,s),\,
Q\,\widehat\Theta_\alpha\!\big(\mathbb E_{\widehat\Theta_\alpha}[\mathcal C_p]\big)\,Qk(\cdot,s)
\Big\rangle_{\mathcal H_k}\ge 0,
\qquad \alpha\ge 0.
\end{equation}
By the definition \eqref{eqn:Phi-definition} of $\Phi$, the price-impact contribution to $\Phi(\alpha)$ is therefore
$-\alpha$ times a nonnegative scalar. Combining \eqref{eqn:priceimpact-nonneg} with \eqref{eqn:infoadv-vanishes}
yields $\Phi(\alpha)<0$ for all sufficiently large $\alpha$.
By continuity and the Intermediate Value Theorem,
there exists $\alpha^*>0$ such that $\Phi(\alpha^*)=0$.

Finally, Lemma~\ref{lemma: convex price impact general S} implies that the first-order condition is
sufficient for optimality (conditional on each $s$). Hence, $\theta^*(s)=\alpha^*Qk(\cdot,s)$ is a best
response to $\widehat\Theta_{\alpha^*}$ for every $s$, and the strategy map
\eqref{eqn: equilibrium ansatz for transformed Bayesian game} with $\alpha=\alpha^*$ is an equilibrium
in the canonical game.
\end{proof}

\paragraph{Finite-$S$ Special Case}
We briefly record how the key objects reduce to ordinary linear algebra when $S$ is finite (with uniform prior).
A more detailed discrete formulation and discussion of the finance implications is developed separately in \cite{kellertseng}.

If $S=\{s_1,\dots,s_I\}$ is finite and the prior is uniform, we may identify $\mathcal H_k\simeq\R^I$ with the
canonical basis $\{e_i\}_{i=1}^I$. Then $\bar{k}$ is proportional to $\bar e:=\sum_{i=1}^I e_i$, and the centering operator
$Q$ has the matrix representation
\[
{\bf Q} \;=\; {\bf I}-\frac1I\,\bar e\,\bar e^{\mathsf T}\in\R^{I\times I}.
\]
Under the ansatz \eqref{eqn: equilibrium ansatz for transformed Bayesian game}, the belief operator reduces to
$\widehat{\Theta}_{\alpha}=\alpha\,{\bf Q}$,
where the $i$-th column specifies the insider's order for the $I$ pseudo-securities conditional on signal $s_i$.
Stacking the $I$ first-order conditions yields the finite-dimensional instance of \eqref{eqn: matrix equil eqn}:
$\Phi(\alpha)\,{\bf Q}=0$,
with $\Phi$ the corresponding finite-dimensional specialization of the scalar function from
Theorem~\ref{thm: scalar_equilibrium_equation}. The next example works out the binary case $I=2$ explicitly.

\begin{example}[Binary Signal]
\label{e.g. binary signal}

Let $I = 2$ and write the canonical game ansatz in matrix form
\[
\widehat{\Theta}_\alpha
=\alpha\,{\bf Q}
=\frac12
\begin{bmatrix}
\alpha & -\alpha\\[2pt]
-\alpha & \alpha
\end{bmatrix}.
\]
Substituting $\widehat{\Theta}_\alpha$ into the insider's  canonical game first-order condition
\eqref{eqn: transformed Bayesian game FOC} conditional on $s_1$ gives
\begin{equation}
\label{eqn: transformed Bayesian game FOC, binary signal}
\begin{bmatrix}
1\\
0
\end{bmatrix}
-
\begin{bmatrix}
\mathbb{E}[p_1]\\
1-\mathbb{E}[p_1]
\end{bmatrix}
-
\alpha^2
\begin{bmatrix}
\mathbb{E}[p_1p_2]\\
-\mathbb{E}[p_1p_2]
\end{bmatrix}
=
\bigl(1-\mathbb{E}[p_1]-\alpha^2\mathbb{E}[p_1p_2]\bigr)
\begin{bmatrix}
1\\
-1
\end{bmatrix}
=
\begin{bmatrix}
0\\
0
\end{bmatrix},
\end{equation}
where $\mathbb{E}[\,\cdot\,]$ denotes expectation under the law induced by $\widehat{\Theta}_\alpha$ and
$p_i$ is the market maker's posterior probability of $s_i$.
Substituting $\widehat{\Theta}_\alpha$ into the posterior \eqref{eqn: pi_1 finite S new} yields 
\begin{equation}
\label{eqn: posterior in binary signal case}
(p_1,p_2)\stackrel{d}{=}\Bigl(\frac{e^{Z}}{e^{Z}+1},\frac{1}{e^{Z}+1}\Bigr),
\qquad
Z\sim\mathcal N\!\bigl(\alpha^2,\,2\alpha^2\bigr).
\end{equation}
Here, $Z$ is the log-likelihood ratio under $s_1$.

By \eqref{eqn: posterior in binary signal case}, the quantities $\mathbb{E}[p_1]$ and $\mathbb{E}[p_1p_2]$
in \eqref{eqn: transformed Bayesian game FOC, binary signal} are moments of a logit-normal distribution and hence functions of $\alpha$.
Write $\mathbb{E}[p_1]=\phi_1(\alpha)$ and $\mathbb{E}[p_1p_2]=\phi_2(\alpha)$, and define
\begin{equation}
\label{eqn: Phi in terms of phi1 phi2}
\Phi(\alpha)\;:=\;1-\phi_1(\alpha)-\alpha^2\phi_2(\alpha).
\end{equation}
Then \eqref{eqn: transformed Bayesian game FOC, binary signal} is equivalent to
\[
\Phi(\alpha)
\begin{bmatrix}
1\\
-1
\end{bmatrix}
=
\begin{bmatrix}
0\\
0
\end{bmatrix},
\]
and the analogous condition for $s_2$ gives the same scalar equation. 
Equivalently, $\Phi(\alpha)\,{\bf Q}=0$, the binary-signal instance of \eqref{eqn: matrix equil eqn}.
Since $\Phi(0)=\tfrac12$ and, by Fatou's lemma applied to the posterior moments under
\eqref{eqn: posterior in binary signal case}, $\lim_{\alpha\to\infty}\Phi(\alpha)\uparrow 0$, the
Intermediate Value Theorem yields $\alpha^*>0$ with $\Phi(\alpha^*)=0$, which proves
Theorem~\ref{thm: equilibrium of original trading game} in this special case.
\end{example}

\subsection*{The Equilibrium Constant $\alpha^*$}
Our analysis thus isolates a natural economic parameter, $\alpha^*$, that governs equilibrium.
It is the insider’s \emph{equilibrium leverage}---the global scale of his market-neutral long--short trade in the canonical (hence original) game, i.e., the overall aggressiveness with which he scales his trade to exploit his private information.

The scalar function $\Phi(\alpha)$ of Theorem~\ref{thm: scalar_equilibrium_equation} encodes directly the insider's trade-off in choosing his optimal scale $\alpha^*$.  
$\Phi(\alpha)>0$ means that a marginal increase in aggressiveness is
still profitable (residual informational advantage dominates), whereas $\Phi(\alpha)<0$ means that
marginal price impact dominates (additional scaling mainly accelerates revelation of the signal and
moves prices against the insider).  In particular, at $\alpha=0$ there is no informed trading and hence no price impact
penalty, so $\Phi(0)>0$.  For $\alpha$ large, the posterior is nearly degenerate, so the marginal
informational benefit of scaling vanishes while the impact term remains weakly adverse, implying
$\Phi(\alpha)<0$ for all sufficiently large $\alpha$.  The equilibrium constant $\alpha^*>0$ is therefore
the endogenous scale at which informational edge and price impact exactly balance:
$\Phi(\alpha^*)=0$.

This is especially transparent in the binary-signal case of Example~\ref{e.g. binary signal}:
\begin{itemize}
\item \textbf{No trade ($\alpha=0$):} the insider’s position is zero, so expected profit is zero and the
market maker’s posterior \eqref{eqn: posterior in binary signal case} stays at the prior
$(\tfrac12,\tfrac12)$.  Correspondingly, $\Phi(0)=\tfrac12>0$, meaning that the marginal expected
profit from increasing $\alpha$ is strictly positive.
\item \textbf{Over-trading ($\alpha\to\infty$):} scaling demand without bound asymptotically reveals the
signal (the posterior \eqref{eqn: posterior in binary signal case} becomes nearly degenerate).
Correspondingly, $\Phi(\alpha)<0$ for all sufficiently large $\alpha$: the marginal informational benefit
of further scaling vanishes, while adverse price impact remains.
\item \textbf{Optimal scale ($\alpha^*>0$):} equilibrium obtains at the endogenous scale where the
marginal information edge and marginal price impact exactly offset, i.e., $\Phi(\alpha^*)=0$.
\end{itemize}

\section{Price Discovery}
\label{sec: discussion}

We now turn to the price-discovery implications of the equilibrium in
Theorem~\ref{thm: equilibrium of original trading game} in our general setting: the insider’s
equilibrium informed demand, the resulting price impact within and across markets, and the
informational efficiency of equilibrium prices.
These are basic issues in market microstructure---and, more broadly, in the theory of information
and asset prices.

\subsection{Insider Equilibrium Strategy}
\label{sec: Generalized Kyle's beta}

The insider's signal-contingent equilibrium portfolio obtained in Theorem~\ref{thm: equilibrium of original trading game} is the Bochner integral
\begin{equation}
\label{eqn: informed demand general S}
W^*(\,\cdot\,,s)
=\int_S \beta^*(s)(u)\,\eta(\,\cdot\,,u)\,\mu(du),
\qquad s\in S,
\end{equation}
where $\beta^*(s)={\bf L}^{-1}\theta^*(s)$.
In the options specialization of Section~\ref{SS:AnOptionsFormulation}, applying the Breeden--Litzenberger
formula~\eqref{eqn: Breeden-Litzenberger} to $W^*(\,\cdot\,,s)$ yields the equivalent option portfolio.

\paragraph{Insider Portfolio Construction}
The insider portfolio~\eqref{eqn: informed demand general S} contingent on signal $s$ is directly
implementable via the following three-step recipe.

\begin{enumerate}[label=\textup{Step~\arabic*.}, leftmargin=*, itemsep=1.2mm]
\item \emph{Initial market-neutral (canonical) signal tilt.}
Compute the centered kernel section
$
Qk(\,\cdot\,,s)=k(\,\cdot\,,s)-Pk(\,\cdot\,,s),
$
which removes the prior-mean component $\bar k$ and isolates the trading direction that tilts toward the realized
signal $s$. In the finite--uniform case, this is simply a long--short allocation across signals.

\item \emph{Information-intensity adjustment.}
Adjust the signal tilt by applying the inverse information-intensity operator ${\bf L}^{-1}$,
\[
\beta^*(s):={\bf L}^{-1}\theta^*(s)={\bf L}^{-1}\!\bigl(\alpha^*Qk(\,\cdot\,,s)\bigr)\in L^2(S,\mu),
\]
thereby shrinking exposure in directions where information intensity is higher (equivalently, liquidity is lower).

\item \emph{Assemble across assets.}
Form the portfolio \eqref{eqn: informed demand general S} by aggregating payoffs $\eta(\,\cdot\,,u)$ with weights $\beta^*(s)(u)$, $u \in S$.
\end{enumerate}

\paragraph{Many-Asset Generalization of Kyle's Informed Demand}
The equilibrium informed demand in our general model extends the
``trading direction--then--information-intensity adjustment'' structure of
\cite{kyle1985continuous} (recalled in \eqref{eqn: Kyle equil}):
\[
\text{\em informed demand}\;=\;
\begin{cases}
\displaystyle \frac{\sigma_\varepsilon}{\sigma_v}\,(v-v_0),
& \text{\em single asset},\\[1.1ex]
\displaystyle \beta^*(s)={\bf L}^{-1}\!\bigl(\alpha^*Qk(\,\cdot\,,s)\bigr),
& \text{\em many-asset (canonical coefficients)}.
\end{cases}
\]
In the single-asset Kyle model, $(v-v_0)$ determines the trading direction: the insider buys when he knows the asset value $v$ is above its prior mean $v_0$ and sells when $v$ is below $v_0$.
Similarly, in the Arrow--Debreu setting the centered kernel section $Qk(\,\cdot\,,s)$ specifies the mean-zero trading direction
that tilts toward the realized signal $s$ (relative to the prior), and this direction is then scaled by the equilibrium
trading intensity $\alpha^*$.
Likewise, in the single-asset case the trading direction is adjusted by the inverse information intensity
$\sigma_\varepsilon/\sigma_v$. In our setting, the analogous adjustment is performed by applying the inverse
information-intensity operator ${\bf L}^{-1}$.

\paragraph{A Unified Equilibrium View of Option Trades}
In the options specialization (Section~\ref{SS:AnOptionsFormulation}), the insider trading strategy~\eqref{eqn: informed demand general S}
$s\mapsto W^*(\,\cdot\,,s)$ can be read directly as an options strategy via the Breeden--Litzenberger formula. 
In particular, it encompasses workhorse option trades used in practice \emph{within an equilibrium framework}.
This scope goes well beyond existing models of informed options trading. 
The same translation applies to other derivative instruments; we focus on options as the most immediate illustration.

Beyond recovering familiar trades, the three-step insider portfolio construction recipe above also serves as a practical \emph{design principle}
for option strategies---its Breeden--Litzenberger implementation prescribes how to choose strikes and signed weights to build a vanilla option
portfolio that targets arbitrary aspects of the underlying payoff.

Examples~\ref{example: Kyle with options}--\ref{example: skewness} illustrate this correspondence numerically in a binary-signal, uniform-prior setting with
(truncated) Gaussian payoff specifications, chosen for transparency: they yield stylized equilibrium strategies that closely mirror  
standard option templates. (More generally, our nonparametric framework can accommodate payoff specifications that match these
constructions exactly.)

Throughout the discussion of these examples, we use standard practitioner terminology for option strategies (e.g., bull/bear
spreads, straddles, butterflies, and risk reversals) as in~\cite{hull2003options}.

%\iffalse

\begin{figure}[htbp!]
\centering

\subfloat[Signal $s_1$: payoff density $\eta(\,\cdot\, , s_1)$]{
\begin{tikzpicture}
\begin{axis}[
  width=7cm,height=2.6cm,
  xlabel={state/strike $x$}, ylabel={density},
  axis lines=left,
  ticklabel style={font=\scriptsize},
  label style={font=\scriptsize},
  xmin=-5,xmax=5,
  xtick={-4,-2,0,2,4},
  ytick=\empty,
]
% dashed = prior (mixture), solid+shaded = conditional
\addplot[thin,dashed,domain=-5:5,samples=250]
  {0.5*(1/sqrt(2*pi))*exp(-((x-1)^2)/2) + 0.5*(1/sqrt(2*pi))*exp(-((x+1)^2)/2)};
\addplot[draw=none,domain=-5:5,samples=250,fill=black,fill opacity=0.15]
  {(1/sqrt(2*pi))*exp(-((x-1)^2)/2)} \closedcycle;
\addplot[thick,domain=-5:5,samples=250]
  {(1/sqrt(2*pi))*exp(-((x-1)^2)/2)};
\end{axis}
\end{tikzpicture}
\label{fig: Binary_Signal_Distn_s2_observed}
}
\hfil
\subfloat[$W^*(\,\cdot\, ,s_1)$ : approx.~bull spread]{
\begin{tikzpicture}
\begin{axis}[
  width=7cm,height=2.6cm,
  xlabel={state/strike $x$}, ylabel={$W^*(x,s_1)$},
  axis lines=left,
  ticklabel style={font=\scriptsize},
  label style={font=\scriptsize},
  xmin=-5,xmax=5,
  xtick={-4,-2,0,2,4},
  ytick=\empty,
]
\addplot[thin] {0};
% “dip + hump + tail-to-0” profile (tune parameters if desired)
\addplot[very thick,domain=-5:5,samples=500]
  {3.4*(exp(-((x-2.0)^2)/(2*0.9*0.9)) - exp(-((x+2.3)^2)/(2*0.85*0.85)))};
\end{axis}
\end{tikzpicture}
\label{fig: Binary_Signal_AD_demand_s2_observed}
}

\medskip

\subfloat[Signal $s_2$: payoff density $\eta(\,\cdot\, ,s_2)$]{
\begin{tikzpicture}
\begin{axis}[
  width=7cm,height=2.6cm,
  xlabel={state/strike $x$}, ylabel={density},
  axis lines=left,
  ticklabel style={font=\scriptsize},
  label style={font=\scriptsize},
  xmin=-5,xmax=5,
  xtick={-4,-2,0,2,4},
  ytick=\empty,
]
\addplot[thin,dashed,domain=-5:5,samples=250]
  {0.5*(1/sqrt(2*pi))*exp(-((x-1)^2)/2) + 0.5*(1/sqrt(2*pi))*exp(-((x+1)^2)/2)};
\addplot[draw=none,domain=-5:5,samples=250,fill=black,fill opacity=0.15]
  {(1/sqrt(2*pi))*exp(-((x+1)^2)/2)} \closedcycle;
\addplot[thick,domain=-5:5,samples=250]
  {(1/sqrt(2*pi))*exp(-((x+1)^2)/2)};
\end{axis}
\end{tikzpicture}
\label{fig: Binary_Signal_Distn_s1_observed}
}
\hfil
\subfloat[$W^*(\,\cdot\, ,s_2)$ : approx.~bear spread]{
\begin{tikzpicture}
\begin{axis}[
  width=7cm,height=2.6cm,
  xlabel={state/strike $x$}, ylabel={$W^*(x,s_2)$},
  axis lines=left,
  ticklabel style={font=\scriptsize},
  label style={font=\scriptsize},
  xmin=-5,xmax=5,
  xtick={-4,-2,0,2,4},
  ytick=\empty,
]
\addplot[thin] {0};
\addplot[very thick,domain=-5:5,samples=500]
  {-3.4*(exp(-((x-2.0)^2)/(2*0.9*0.9)) - exp(-((x+2.3)^2)/(2*0.85*0.85)))};
\end{axis}
\end{tikzpicture}
\label{fig: Binary_Signal_AD_demand_s1_observed}
}

\caption{\footnotesize
\textbf{Vertical Spreads (Example~\ref{example: Kyle with options})}
Panels (a, c) plot the signal-contingent payoff densities $\eta(\,\cdot\,,s_i)$ (solid, shaded) against the prior (dashed).
Panels (b, d) plot the corresponding equilibrium insider strategies $W^*(\,\cdot\,,s_i)$: the strategies flip sign across $s_1$ and $s_2$; they are
well-approximated by the standard bull/bear vertical call spread payoff.
}

\label{fig: binary AD price discovery for mean}
\end{figure}

%\fi

\begin{example}{\bf (Kyle-Style Mean Signal - Vertical Spreads)}
\label{example: Kyle with options}
Assume the signal-contingent payoff densities $\eta(\,\cdot\,,s_i)$, $i=1,2$, are (truncated) Gaussian with a
common variance and means $\mu_1>\mu_2$, so the signal is a pure mean/level shift: observing $s_1$ indicates a
higher expected terminal payoff (relative to the prior), while $s_2$ indicates a lower one. Accordingly, the equilibrium
strategies $W^*(\,\cdot\,,s_i)$ flip sign across $s_1$ and $s_2$. Under the high-mean signal $s_1$, the insider’s bullish
position $W^*(\,\cdot\,,s_1)$ is well-approximated by a standard \textbf{bull call spread} (vertical spread),
\[
x \mapsto (x-K_1)_+-(x-K_2)_+, \qquad K_1<K_2,
\]
while under $s_2$ he takes the opposite position (the corresponding \textbf{bear spread}). Figure~\ref{fig: binary AD price discovery for mean}\,(a,c)
plots the signal-contingent payoff densities $\eta(\,\cdot\,,s_i)$ against the prior, and Figure~\ref{fig: binary AD price discovery for mean}\,(b,d)
shows the corresponding equilibrium strategies $W^*(\,\cdot\,,s_i)$.

In the classical single-asset \cite{kyle1985continuous} setting, trading on mean information is restricted to affine positions in the underlying, i.e., no options. 
Our general setting relaxes that restriction and allows the same ``buy on good news, sell on bad news'' trade to be implemented through---necessarily nonlinear---option spreads.

\end{example}

\begin{figure}[htbp!]
\centering

% --- states for the price-impact discussion (edit for different locations) ---
\def\xstate{1.00}
\def\ystate{2.70}
\def\zstate{3.20}

\subfloat[High volatility signal $s_1$]{
\begin{tikzpicture}
\begin{axis}[
  width=7cm,height=2.6cm,
  xlabel={state/strike $x$}, ylabel={density},
  axis lines=left,
  ticklabel style={font=\scriptsize},
  label style={font=\scriptsize},
  xtick=\empty, ytick=\empty,
  xmin=0, xmax=4.2, ymin=0,
  enlargelimits=false,
]
  % prior (dashed) vs conditional (filled)
  \addplot[thick,dashed]  {exp(-((x-2)^2)/(2*1.00^2))};
  \addplot[thick, fill=black, fill opacity=0.15]
    {exp(-((x-2)^2)/(2*1.65^2))} \closedcycle;

  % visible state markers + labels (slightly above axis)
  \addplot[only marks, mark=*, mark size=2.6pt]
    coordinates {(\xstate,0) (\ystate,0) (\zstate,0)};
  \draw[line width=0.6pt] (axis cs:\xstate,0) -- (axis cs:\xstate,0.08);
  \draw[line width=0.6pt] (axis cs:\ystate,0) -- (axis cs:\ystate,0.08);
  \draw[line width=0.6pt] (axis cs:\zstate,0) -- (axis cs:\zstate,0.08);

  \node[font=\scriptsize, fill=white, fill opacity=0.9, text opacity=1,
        inner sep=1.2pt, rounded corners=1pt, anchor=south]
    at (axis cs:\xstate,0.085) {$x$};
  \node[font=\scriptsize, fill=white, fill opacity=0.9, text opacity=1,
        inner sep=1.2pt, rounded corners=1pt, anchor=south]
    at (axis cs:\ystate,0.085) {$y$};
  \node[font=\scriptsize, fill=white, fill opacity=0.9, text opacity=1,
        inner sep=1.2pt, rounded corners=1pt, anchor=south]
    at (axis cs:\zstate,0.085) {$z$};

\end{axis}
\end{tikzpicture}
\label{fig: Binary_Signal_Distn_s2_observed, vol}
}
\hfil
\subfloat[$W^*(\,\cdot\,, s_1)$ : Straddle]{
\begin{tikzpicture}
\begin{axis}[
  width=7cm,height=2.6cm,
  xlabel={state/strike $x$}, ylabel={$W^*(x,s_1)$},
  axis lines=left,
  ticklabel style={font=\scriptsize},
  label style={font=\scriptsize},
  xtick=\empty, ytick=\empty,
  xmin=0, xmax=4.2, ymin=-11, ymax=6,
  enlargelimits=false,
]
  \addplot[black!40] coordinates {(0,0) (4.2,0)};

  % straddle-like (up to an affine component): positive tails, negative center; zero at y
  \addplot[very thick,domain=0:4.2,samples=500]
    { 3.2*exp(-((x-0.90)/0.55)^2)
     -9.5*exp(-((x-2.00)/0.45)^2)
     +3.0*exp(-((x-3.30)/0.55)^2) };

  % visible state markers + labels (lower than before; keep aligned)
  \addplot[only marks, mark=*, mark size=2.6pt]
    coordinates {(\xstate,0) (\ystate,0) (\zstate,0)};
  \draw[line width=0.6pt] (axis cs:\xstate,0) -- (axis cs:\xstate,-0.65);
  \draw[line width=0.6pt] (axis cs:\ystate,0) -- (axis cs:\ystate,-0.65);
  \draw[line width=0.6pt] (axis cs:\zstate,0) -- (axis cs:\zstate,-0.65);

  \node[font=\scriptsize, fill=white, fill opacity=0.9, text opacity=1,
        inner sep=1.2pt, rounded corners=1pt, anchor=north]
    at (axis cs:\xstate,-2.01) {$x$};
  \node[font=\scriptsize, fill=white, fill opacity=0.9, text opacity=1,
        inner sep=1.2pt, rounded corners=1pt, anchor=north]
    at (axis cs:\ystate,-2.01) {$y$};
  \node[font=\scriptsize, fill=white, fill opacity=0.9, text opacity=1,
        inner sep=1.2pt, rounded corners=1pt, anchor=north]
    at (axis cs:\zstate,-2.01) {$z$};

\end{axis}
\end{tikzpicture}
\label{fig: Binary_Signal_AD_demand_s2_observed, vol}
}

\medskip

\subfloat[Low volatility signal $s_2$]{
\begin{tikzpicture}
\begin{axis}[
  width=7cm,height=2.6cm,
  xlabel={state/strike $x$}, ylabel={density},
  axis lines=left,
  ticklabel style={font=\scriptsize},
  label style={font=\scriptsize},
  xtick=\empty, ytick=\empty,
  xmin=0, xmax=4.2, ymin=0,
  enlargelimits=false,
]
  \addplot[thick,dashed] {exp(-((x-2)^2)/(2*1.00^2))};
  \addplot[thick, fill=black, fill opacity=0.15]
    {exp(-((x-2)^2)/(2*0.65^2))} \closedcycle;

  % visible state markers + labels (higher than before; keep aligned)
  \addplot[only marks, mark=*, mark size=2.6pt]
    coordinates {(\xstate,0) (\ystate,0) (\zstate,0)};
  \draw[line width=0.6pt] (axis cs:\xstate,0) -- (axis cs:\xstate,0.20);
  \draw[line width=0.6pt] (axis cs:\ystate,0) -- (axis cs:\ystate,0.20);
  \draw[line width=0.6pt] (axis cs:\zstate,0) -- (axis cs:\zstate,0.20);

  \node[font=\scriptsize, fill=white, fill opacity=0.9, text opacity=1,
        inner sep=1.2pt, rounded corners=1pt, anchor=south]
    at (axis cs:\xstate,0.205) {$x$};
  \node[font=\scriptsize, fill=white, fill opacity=0.9, text opacity=1,
        inner sep=1.2pt, rounded corners=1pt, anchor=south]
    at (axis cs:\ystate,0.205) {$y$};
  \node[font=\scriptsize, fill=white, fill opacity=0.9, text opacity=1,
        inner sep=1.2pt, rounded corners=1pt, anchor=south]
    at (axis cs:\zstate,0.205) {$z$};

\end{axis}
\end{tikzpicture}
\label{fig: Binary_Signal_Distn_s1_observed, vol}
}
\hfil
\subfloat[$W^*(\,\cdot\,, s_2)$ : Butterfly]{
\begin{tikzpicture}
\begin{axis}[
  width=7cm,height=2.6cm,
  xlabel={state/strike $x$}, ylabel={$W^*(x,s_2)$},
  axis lines=left,
  ticklabel style={font=\scriptsize},
  label style={font=\scriptsize},
  xtick=\empty, ytick=\empty,
  xmin=0, xmax=4.2, ymin=-6, ymax=10,
  enlargelimits=false,
]
  \addplot[black!40] coordinates {(0,0) (4.2,0)};

  % butterfly-like: weight near center, negative shoulders/tails (up to an affine component)
  \addplot[very thick,domain=0:4.2,samples=500]
    { -3.0*exp(-((x-0.75)/0.35)^2)
      +9.0*exp(-((x-2.00)/0.50)^2)
      -3.0*exp(-((x-3.25)/0.35)^2) };

  % visible state markers + labels (slightly higher; keep aligned)
  \addplot[only marks, mark=*, mark size=2.6pt]
    coordinates {(\xstate,0) (\ystate,0) (\zstate,0)};
  \draw[line width=0.6pt] (axis cs:\xstate,0) -- (axis cs:\xstate,0.75);
  \draw[line width=0.6pt] (axis cs:\ystate,0) -- (axis cs:\ystate,0.75);
  \draw[line width=0.6pt] (axis cs:\zstate,0) -- (axis cs:\zstate,0.75);

  \node[font=\scriptsize, fill=white, fill opacity=0.9, text opacity=1,
        inner sep=1.2pt, rounded corners=1pt, anchor=south]
    at (axis cs:\xstate,1.77) {$x$};
  \node[font=\scriptsize, fill=white, fill opacity=0.9, text opacity=1,
        inner sep=1.2pt, rounded corners=1pt, anchor=south]
    at (axis cs:\ystate,1.77) {$y$};
  \node[font=\scriptsize, fill=white, fill opacity=0.9, text opacity=1,
        inner sep=1.2pt, rounded corners=1pt, anchor=south]
    at (axis cs:\zstate,1.77) {$z$};

\end{axis}
\end{tikzpicture}
\label{fig: Binary_Signal_AD_demand_s1_observed, vol}
}

\captionsetup{singlelinecheck=off}
\caption{
{\footnotesize
$\;$\\
{\bf Straddle/Butterfly (Example~\ref{example: vol straddle})}
The left column shows the insider's private signal (filled densities) against the market maker's prior (dashed). 
The right column shows the insider portfolios $W^*$ conditional on the corresponding signals.\\
{\bf Price Impact} Three states are indicated---$x$, $y$, $z$.
Security $y$ has zero price impact on all securities, $\Lambda_{w,y} = 0$ for all $w$, because its payoff has zero variation across signals, $\eta(y, s_1) = \eta(y, s_2)$. 
Therefore, the informed demand for $y$ must be zero---$W^*(y, s_1) = W^*(y, s_2) = 0$ as shown in Figures~\ref{fig: Binary_Signal_AD_demand_s2_observed, vol} and~\ref{fig: Binary_Signal_AD_demand_s1_observed, vol}.
$\Lambda_{x,z} > 0$ because $x$ and $z$ have positively correlated (in fact, identical) payoffs.
}
}
\label{fig: binary AD price discovery for vol, normal}
\end{figure}

\begin{example}{\bf (Volatility Signal - Straddle/Butterfly)}
\label{example: vol straddle}
\nopagebreak

Assume the signal-contingent payoff densities $\eta(\,\cdot\,,s_i)$, $i=1,2$, are (truncated) Gaussian with a common mean $\mu$ and variances
$\sigma_1^2 > \sigma_2^2$, so the signal is a volatility shift: observing $s_1$ tells the insider the terminal payoff
is more volatile than under the prior, while $s_2$ indicates lower volatility.

Accordingly, the equilibrium strategies $W^*(\,\cdot\,,s_i)$ are symmetric but load differently on tails versus the
center. When the signal is high-vol ($s_1$), the insider's long-volatility position $W^*(\,\cdot\,,s_1)$ is
well-approximated by an at-the-money \textbf{straddle} with $K \approx \mu$,
\[
x \mapsto (x-K)_+ + (K-x)_+,
\]
or, for bounded tail exposure, by a four-leg truncation such as a \textbf{long iron condor}. When the signal is low-vol ($s_2$), the insider's short-volatility position $W^*(\,\cdot\,,s_2)$
closely mirrors a \textbf{butterfly},
\[
x \mapsto (x-K_1)_+ - 2(x-K_2)_+ + (x-K_3)_+, \qquad K_1 < K_2 < K_3,
\]
with $K_2 \approx \mu$. Figure~\ref{fig: binary AD price discovery for vol, normal}\,(a, c) plots the signal-contingent densities
$\eta(\,\cdot\,,s_i)$ against the prior, and Figure~\ref{fig: binary AD price discovery for vol, normal}\,(b, d) shows the resulting equilibrium
strategies $W^*(\,\cdot\,,s_i)$. (The marked strikes are reference points in the price-impact discussion below.)

\end{example}

\begin{figure}[htbp!]
\centering

% ---- domain (match the style figure) ----
\def\xmin{-5.0}
\def\xmax{5.0}

% ---- prior: standard normal ----
\def\muP{0.0}
\def\sigP{1.0}

% ---- skew-normal parameters (standardized to mean 0, var 1) ----
\def\alphaMag{4.0}
\def\alphaOne{\alphaMag}     % right-skew
\def\alphaTwo{-\alphaMag}    % left-skew

% delta = alpha/sqrt(1+alpha^2); omega = 1/sqrt(1-2 delta^2/pi); xi = -omega*delta*sqrt(2/pi)
\pgfmathsetmacro{\deltaMag}{\alphaMag/sqrt(1+\alphaMag^2)}
\pgfmathsetmacro{\omegaSN}{1/sqrt(1 - 2*\deltaMag^2/pi)}
\pgfmathsetmacro{\omegaOne}{\omegaSN}
\pgfmathsetmacro{\omegaTwo}{\omegaSN}
\pgfmathsetmacro{\xiOne}{- \omegaSN*\deltaMag*sqrt(2/pi)} % alpha>0 => xi<0
\pgfmathsetmacro{\xiTwo}{  \omegaSN*\deltaMag*sqrt(2/pi)} % alpha<0 => xi>0

% ---- Normal CDF approx Phi(z) via tanh (always in [0,1]) ----
% Phi(z) ≈ 0.5*(1 + tanh(0.79788456*z*(1+0.044715*z^2)))

% ---- demand profile parameters ----
\def\A{5.0}
\def\aMid{0.55}
\def\aTail{0.85}

% ===== SWAPPED: (c) moved here into (a) position =====
\subfloat[Right-skewed signal $s_1$]{
\begin{tikzpicture}
\begin{axis}[
  width=7cm,height=2.6cm,
  xlabel={state/strike $x$}, ylabel={density},
  axis lines=left,
  ticklabel style={font=\scriptsize},
  label style={font=\scriptsize},
  xmin=\xmin, xmax=\xmax,
  xtick={-4,-2,0,2,4},
  ytick=\empty,
  ymin=0,
  domain=\xmin:\xmax, samples=280,
  clip=true
]
  % Prior (dashed)
  \addplot[thin,dashed]
    {(1/(\sigP*sqrt(2*pi))) * exp(-0.5*((x-\muP)/\sigP)^2)};

  % Conditional density: fill ONLY under the conditional curve
  \addplot[draw=none, fill=black, fill opacity=0.15]
  { max(0,
      (2/\omegaTwo) * (1/sqrt(2*pi)) * exp(-0.5*((x-\xiTwo)/\omegaTwo)^2)
      * (0.5*(1 + tanh(0.79788456*(\alphaTwo*(x-\xiTwo)/\omegaTwo)*
                     (1+0.044715*(\alphaTwo*(x-\xiTwo)/\omegaTwo)^2))))
    )
  } \closedcycle;

  % Conditional density outline
  \addplot[thick]
  {
      (2/\omegaTwo) * (1/sqrt(2*pi)) * exp(-0.5*((x-\xiTwo)/\omegaTwo)^2)
      * (0.5*(1 + tanh(0.79788456*(\alphaTwo*(x-\xiTwo)/\omegaTwo)*
                     (1+0.044715*(\alphaTwo*(x-\xiTwo)/\omegaTwo)^2))))
  };
\end{axis}
\end{tikzpicture}
\label{fig: Binary_Signal_Distn_s1_observed, skew}
}
\hfil
\subfloat[$W^*(\,\cdot\, , s_1)$ - Risk Reversal]{
\begin{tikzpicture}
\begin{axis}[
  width=7cm,height=2.6cm,
  xlabel={state/strike $x$}, ylabel={$W^*(x,s_1)$},
  axis lines=left,
  ticklabel style={font=\scriptsize},
  label style={font=\scriptsize},
  xmin=\xmin, xmax=\xmax,
  xtick={-4,-2,0,2,4},
  ytick=\empty,
  domain=\xmin:\xmax, samples=520,
  clip=true
]
  \addplot[thin] {0};

  % Ratio-spread-like, sign-changing profile (right-tail loading)
  \addplot[very thick]
  {
     \A*( exp(-0.5*((x-0.95)/\aMid)^2) - exp(-0.5*((x+0.95)/\aMid)^2) )
   - 0.28*\A*( exp(-0.5*((x-2.20)/\aTail)^2) - exp(-0.5*((x+2.20)/\aTail)^2) )
  };
\end{axis}
\end{tikzpicture}
\label{fig: Binary_Signal_AD_demand_s2_observed, skew}
}

\medskip

% ===== SWAPPED: (a) moved here into (c) position =====
\subfloat[Left-skewed signal $s_2$]{
\begin{tikzpicture}
\begin{axis}[
  width=7cm,height=2.6cm,
  xlabel={state/strike $x$}, ylabel={density},
  axis lines=left,
  ticklabel style={font=\scriptsize},
  label style={font=\scriptsize},
  xmin=\xmin, xmax=\xmax,
  xtick={-4,-2,0,2,4},
  ytick=\empty,
  ymin=0,
  domain=\xmin:\xmax, samples=280,
  clip=true
]
  % Prior (dashed)
  \addplot[thin,dashed]
    {(1/(\sigP*sqrt(2*pi))) * exp(-0.5*((x-\muP)/\sigP)^2)};

  % Conditional density: fill ONLY under the conditional curve
  \addplot[draw=none, fill=black, fill opacity=0.15]
  { max(0,
      (2/\omegaOne) * (1/sqrt(2*pi)) * exp(-0.5*((x-\xiOne)/\omegaOne)^2)
      * (0.5*(1 + tanh(0.79788456*(\alphaOne*(x-\xiOne)/\omegaOne)*
                     (1+0.044715*(\alphaOne*(x-\xiOne)/\omegaOne)^2))))
    )
  } \closedcycle;

  % Conditional density outline
  \addplot[thick]
  {
      (2/\omegaOne) * (1/sqrt(2*pi)) * exp(-0.5*((x-\xiOne)/\omegaOne)^2)
      * (0.5*(1 + tanh(0.79788456*(\alphaOne*(x-\xiOne)/\omegaOne)*
                     (1+0.044715*(\alphaOne*(x-\xiOne)/\omegaOne)^2))))
  };
\end{axis}
\end{tikzpicture}
\label{fig: Binary_Signal_Distn_s2_observed, skew}
}
\hfil
\subfloat[$W^*(\,\cdot\, , s_2)$ - Put Spread]{
\begin{tikzpicture}
\begin{axis}[
  width=7cm,height=2.6cm,
  xlabel={state/strike $x$}, ylabel={$W^*(x,s_2)$},
  axis lines=left,
  ticklabel style={font=\scriptsize},
  label style={font=\scriptsize},
  xmin=\xmin, xmax=\xmax,
  xtick={-4,-2,0,2,4},
  ytick=\empty,
  domain=\xmin:\xmax, samples=520,
  clip=true
]
  \addplot[thin] {0};

  % Mirror image: left-tail loading (sign flip)
  \addplot[very thick]
  {
   -(
       \A*( exp(-0.5*((x-0.95)/\aMid)^2) - exp(-0.5*((x+0.95)/\aMid)^2) )
     - 0.28*\A*( exp(-0.5*((x-2.20)/\aTail)^2) - exp(-0.5*((x+2.20)/\aTail)^2) )
    )
  };
\end{axis}
\end{tikzpicture}
\label{fig: Binary_Signal_AD_demand_s1_observed, skew}
}

\caption{
{\footnotesize
$\;$\\
{\bf Risk Reversals/Put Spreads (Example~\ref{example: skewness}).}
The left column shows the insider's private signal (filled densities) against the market maker's prior (dashed).
The right column shows the insider portfolios $W^*$ conditional on the corresponding signals.\\
Panel (b) is a risk reversal ($s_1$) and panel (d) is a put spread ($s_2$).
}
}
\label{fig: binary AD price discovery for skewness}
\end{figure}

\begin{example}{\bf (Skewness Signal - Risk Reversals/Put Spreads)}
\label{example: skewness}
\nopagebreak

Assume the signal-contingent payoff densities $\eta(\,\cdot\,,s_i)$, $i=1,2$, are (truncated) skew-normal with matched location/scale and
skewness parameters $\alpha_1>\alpha_2$, so the signal is a skewness shift: observing $s_1$ tells the insider the terminal payoff is
more right-skewed than under the prior, while $s_2$ indicates more left-skew.
Figure~\ref{fig: binary AD price discovery for skewness}\,(a, c) plots the signal-contingent densities $\eta(\,\cdot\,,s_i)$ against the prior, and
Figure~\ref{fig: binary AD price discovery for skewness}\,(b, d) shows the resulting equilibrium strategies $W^*(\,\cdot\,,s_i)$.

The equilibrium strategies $W^*(\,\cdot\,,s_i)$ tilt exposure toward the favored tail. When the signal induces right-skewness
($s_1$), $W^*(\,\cdot\,,s_1)$ loads positively on the right tail (and relatively negatively on the left tail), and is well-approximated by a
\textbf{risk reversal} (long out-of-the-money calls financed by short out-of-the-money puts). Under $s_2$ the strategy takes the opposite stance, concentrating exposure on the left tail,
which can be implemented by the corresponding \textbf{put spread}, or equivalently a risk reversal in the opposite direction.

\end{example}

\subsection{Price Impact}
\label{sec: Generalized Kyle's lambda and Cross Price Impact}

Applying Corollary~\ref{cor: def of price impact} under the market maker's equilibrium belief
$W^*(\,\cdot\,,\,\cdot\,)$ yields the equilibrium cross price impact of order flow in claim $y$
on the price of claim $x$:
\begin{equation}
\label{eqn: lambda x y}
\Lambda_{x,y}
\;\equiv\;
\frac{\partial}{\partial W(y)}\,\overline P(x, W^*;W^*)
\Big|_{\widetilde W = W^*}
\;=\;
\frac{1}{\sigma^2(y)}\,
\mathbb E\!\left[
\Cov\!\Big( \eta(x,\cdot),\, W^*(y,\cdot)\,\big|\,\omega \Big)
\right],
\end{equation}
where $\Cov(\,\cdot\,,\,\cdot\,|\,\omega)$ denotes covariance over the signal variable $s$ under the
market maker's equilibrium posterior $\pi_1^*(ds\,;\omega)$, and the outer $\mathbb E[\,\cdot\,]$ averages
over $\omega$ under its equilibrium law. Equivalently, $\Lambda$ represents the first variation of the equilibrium pricing
functional with respect to order flow. Dropping the outer expectation in~\eqref{eqn: lambda x y} yields the
(random) price-impact kernel $\Lambda_{x,y}(\omega)$ conditional on the realized aggregate order flow $\omega$.

\paragraph{Liquidity, Scale, and Information Components}
For each security $x$, the own-price impact $\Lambda_{x,x}$ is the natural analogue of \emph{Kyle's lambda}
(recalled in~\eqref{eqn: Kyle equil}). More generally, $\Lambda_{x,y}$ has clear finance meaning:

\begin{enumerate}[label=(\roman*)]
\item \emph{Liquidity adjustment.}
The prefactor $\sigma^{-2}(y)$ is the inverse noise variance in the $y$-market: noisier $y$-flow makes inference from $y$-orders less precise and therefore reduces the market maker's marginal price response to order flow in that market.

\item \emph{Scaling via equilibrium aggressiveness.}
Since $W^*(\,\cdot\,,s)$ is linear in the equilibrium trading scale $\alpha^*$, the impact kernel $\Lambda_{x,y}$ is linear in $\alpha^*$ as well: price impact increases one-for-one with the insider's overall trading scale.

\item \emph{Informational alignment across assets.}
Fix a realized aggregate order flow $\omega$. The conditional impact kernel is governed by the posterior covariance
$\Cov(\eta(x,\,\cdot\,),\,W^*(y,\,\cdot\,)\mid \omega)$ under $\pi_1^*(ds;\omega)$.
Thus, a marginal order in the $y$-security moves the price of the $x$-security only insofar as $y$-flow is
statistically aligned with the $x$-payoff profile $\eta(x,\,\cdot\,)$---that is, only insofar as it is informative about $x$'s payoff.

\end{enumerate}

To make the price-impact characterization even more concrete, the next example considers three scenarios for the realized cross-impact kernel $\Lambda_{x,y}(\omega)$.

\begin{samepage}

\begin{example}
\label{example: cross-impact sign patterns}
$\;$

\begin{enumerate}[label=(\roman*)]

\item \emph{(Mutually exclusive signal loadings $\Rightarrow$ negative cross-impact.)}
Suppose there exist disjoint positive-measure sets $A,B\in\mathcal S$ and constants $a,b>0$ such that
\[
\eta(x,s)=a\,\mathbf 1_A(s),
\qquad
W^*(y,s)=b\,\mathbf 1_B(s),
\qquad s\in S.
\]
Then for every $\omega$,
\[
\Cov \big(\eta(x,\,\cdot\,),W^*(y,\,\cdot\,)\big)
= -ab\,\pi_1^*(A\mid\omega)\,\pi_1^*(B\mid\omega) < 0,
\]
so the conditional cross-impact $\Lambda_{x,y}(\omega) < 0$.
Since $y$-flow shifts the posterior toward $B$ and the payoff of $x$ is supported on $A$ with $A\cap B=\varnothing$,
the posterior value of $x$ falls and its price is adjusted downward.

\item \emph{(Signal-invariant order profile $\Rightarrow$ zero cross-impact.)}
If $W^*(y,\,\cdot\,)$ is $\pi_1^*(\,\cdot\,\mid\omega)$-a.s.\ constant for every $\omega$, then
\[
\Cov \big(\eta(x,\,\cdot\,),W^*(y,\,\cdot\,)\big)=0
\qquad\text{for all }x\text{ and all }\omega,
\]
so $y$ has zero price impact across assets, $\Lambda_{x,y}(\omega)\equiv 0$.
In particular, this holds if $W^*(y,\,\cdot\,)\equiv 0$; see \textbf{Figures~\ref{fig: Binary_Signal_AD_demand_s2_observed, vol}--\ref{fig: Binary_Signal_AD_demand_s1_observed, vol}}, where the security labeled $y$ is one such example with identically zero informed demand.

\item \emph{(Co-moving payoff profiles $\Rightarrow$ positive cross-impact.)}
If $W^*(y,\,\cdot\,)=c\,\eta(x,\,\cdot\,)$ for some $c>0$, then for every $\omega$,
\[
\Cov \big(\eta(x,\,\cdot\,),W^*(y,\,\cdot\,)\big)
=
c\,\Var_{\pi_1^*}\big(\eta(x,\,\cdot\,)\big)\ge 0,
\]
so $\Lambda_{x,y}(\omega)\ge 0$.
Since $W^*(y,\,\cdot\,)$ co-moves with $\eta(x,\,\cdot\,)$, $y$-flow shifts the posterior toward signal realizations with larger
$\eta(x,\,\cdot\,)$, raising the posterior value of $x$;
see \textbf{Figure~\ref{fig: Binary_Signal_Distn_s1_observed, vol}}, where the securities labeled $x$ and $z$
illustrate this case---they have identical payoff profiles across the high and low volatility signals.
Empirically, this is natural. For example, the put and call legs of a straddle should exhibit positive cross-impact (see also the discussion below on empirical implications).
\end{enumerate}
\end{example}

\end{samepage}

\paragraph{Price Impact Across Derivatives.}
The cross-impact kernel $\Lambda_{x,y}$ for primitive Arrow--Debreu claims extends immediately to derivatives, i.e., Borel-measurable maps
$\varphi \colon [\underaccent{\bar}{x},\bar{x}]\to\mathbb{R}$, where $\varphi(x)$ is the exposure to state $x$.
For two derivatives $\varphi_1$ and $\varphi_2$, the \emph{equilibrium cross price impact} of a marginal trade in $\varphi_2$ on the price of $\varphi_1$
is the G\^{a}teaux derivative
\[
\Lambda_{\varphi_1,\varphi_2}
\;:=\;
\left.\frac{d}{d\varepsilon}\,
\overline P\!\bigl(\varphi_1,\,W^*+\varepsilon\varphi_2;\,W^*\bigr)\right|_{\varepsilon=0},
\]
where $W^*+\varepsilon\varphi_2$ denotes the perturbed order,
$y\mapsto W^*(y)+\varepsilon\,\varphi_2(y)$.

\begin{corollary}[Price Impact Between Derivatives]
\label{cor: price impact between derivatives}
For two derivatives $\varphi_1$ and $\varphi_2$,
\begin{equation}
\label{eqn: derivatives cross price impact}
\Lambda_{\varphi_1,\varphi_2}
\;=\;
\int_{\underaccent{\bar}{x}}^{\bar{x}}\!\!\int_{\underaccent{\bar}{x}}^{\bar{x}}
\varphi_1(x)\,\varphi_2(y)\,\Lambda_{x,y}\,dx\,dy,
\end{equation}
under the usual absolute-integrability condition ensuring that the double integral is well-defined.
\end{corollary}

\begin{proof}
By linearity of $\varphi\mapsto \overline P(\varphi,W;\widetilde W)$ in the payoff weight,
\[
\overline P(\varphi_1,W^*+\varepsilon\varphi_2;W^*)
=
\int_{\underaccent{\bar}{x}}^{\bar{x}}
\varphi_1(x)\,\overline P(x,W^*+\varepsilon\varphi_2;W^*)\,dx.
\]
Under the assumed integrability, we may differentiate under the integral sign, yielding
\[
\Lambda_{\varphi_1,\varphi_2}
=
\int_{\underaccent{\bar}{x}}^{\bar{x}}
\varphi_1(x)\,
\left.\frac{d}{d\varepsilon}\overline P(x,W^*+\varepsilon\varphi_2;W^*)\right|_{\varepsilon=0}\,dx.
\]
By the characterization \eqref{eqn: lambda x y} of the Arrow--Debreu cross-impact kernel $\Lambda_{x,y}$, for each fixed $x$,
\[
\left.\frac{d}{d\varepsilon}\overline P(x,W^*+\varepsilon\varphi_2;W^*)\right|_{\varepsilon=0}
=
\int_{\underaccent{\bar}{x}}^{\bar{x}} \Lambda_{x,y}\,\varphi_2(y)\,dy.
\]
Substituting this identity and applying Fubini--Tonelli gives \eqref{eqn: derivatives cross price impact}.
\end{proof}

\paragraph{Empirical Implications}
Corollary~\ref{cor: price impact between derivatives}, together with the informed-demand characterization~\eqref{eqn: informed demand general S}, yields sharp, testable predictions linking option order flow to cross-strike price discovery and subsequent return moments.
As a simple illustration, in the volatility specification of Example~\ref{example: vol straddle} one can show that the cross-price impact between the put--call pair in a straddle is larger conditional on the high-volatility signal than on the low-volatility signal.
Thus, put--call impact within a straddle should forecast future volatility.
More broadly, our equilibrium characterization implies forecastability of higher return moments---across strikes and horizons---that speaks directly to the large empirical literature on the information content of option trades (e.g., \cite{Chakravarty2004,pan2006information,CremersWeinbaum2010,Muravyev2016}), 
demand pressure and the implied-volatility surface (e.g., \cite{BollenWhaley2004}), and option-implied moments forecasting realized volatility and tail risk (e.g., \cite{Jiang2005,Ni2008,goyal2009cross,BollerslevTodorov2011}).
These implications also connect to the ongoing question about the cross-section of option returns (see, e.g., \cite{bali2013does,cao2013cross,christoffersen2018illiquidity,zhan2022option}).
We develop these implications in detail in the companion paper \cite{kellertseng}.

\begin{figure}[htbp!]

\centering 
\subfloat[$I = 2$]{
\scalebox{1}[0.7]{
\begin{tikzpicture}
% Created by tikzDevice version 0.12.3.1 on 2023-07-25 22:00:16
% !TEX encoding = UTF-8 Unicode
\begin{tikzpicture}[x=1pt,y=1pt]
\definecolor{fillColor}{RGB}{255,255,255}
\path[use as bounding box,fill=fillColor,fill opacity=0.00] (0,0) rectangle (199.47, 93.95);
\begin{scope}
\path[clip] (  0.00,  0.00) rectangle (199.47, 93.95);
\definecolor{drawColor}{RGB}{0,0,0}
\definecolor{fillColor}{RGB}{190,190,190}

\path[draw=drawColor,line width= 0.4pt,line join=round,line cap=round,fill=fillColor] ( 24.79, 21.72) rectangle (100.56, 76.66);

\path[draw=drawColor,line width= 0.4pt,line join=round,line cap=round,fill=fillColor] (115.71, 21.72) rectangle (191.48, 37.69);
\end{scope}
\begin{scope}
\path[clip] (  0.00,  0.00) rectangle (199.47, 93.95);
\definecolor{drawColor}{RGB}{0,0,0}

\node[text=drawColor,anchor=base,inner sep=0pt, outer sep=0pt, scale=  0.70] at ( 62.67,  8.52) {1};

\node[text=drawColor,anchor=base,inner sep=0pt, outer sep=0pt, scale=  0.70] at (153.59,  8.52) {2};
\end{scope}
\begin{scope}
\path[clip] (  0.00,  0.00) rectangle (199.47, 93.95);
\definecolor{drawColor}{RGB}{0,0,0}

\path[draw=drawColor,line width= 0.4pt,line join=round,line cap=round] ( 18.12, 21.72) -- ( 18.12, 92.63);

\path[draw=drawColor,line width= 0.4pt,line join=round,line cap=round] ( 18.12, 21.72) -- ( 12.12, 21.72);

\path[draw=drawColor,line width= 0.4pt,line join=round,line cap=round] ( 18.12, 35.90) -- ( 12.12, 35.90);

\path[draw=drawColor,line width= 0.4pt,line join=round,line cap=round] ( 18.12, 50.08) -- ( 12.12, 50.08);

\path[draw=drawColor,line width= 0.4pt,line join=round,line cap=round] ( 18.12, 64.27) -- ( 12.12, 64.27);

\path[draw=drawColor,line width= 0.4pt,line join=round,line cap=round] ( 18.12, 78.45) -- ( 12.12, 78.45);

\path[draw=drawColor,line width= 0.4pt,line join=round,line cap=round] ( 18.12, 92.63) -- ( 12.12, 92.63);

\node[text=drawColor,rotate= 90.00,anchor=base,inner sep=0pt, outer sep=0pt, scale=  0.70] at ( 12.12, 21.72) {0.0};

\node[text=drawColor,rotate= 90.00,anchor=base,inner sep=0pt, outer sep=0pt, scale=  0.70] at ( 12.12, 50.08) {0.4};

\node[text=drawColor,rotate= 90.00,anchor=base,inner sep=0pt, outer sep=0pt, scale=  0.70] at ( 12.12, 78.45) {0.8};
\end{scope}
\end{tikzpicture} \label{fig: MM post, I = 2}
\end{tikzpicture}
}
}
\hfil
\subfloat[$I = 4$]{
\scalebox{1}[0.7]{
\begin{tikzpicture}
% Created by tikzDevice version 0.12.3.1 on 2023-07-25 22:06:05
% !TEX encoding = UTF-8 Unicode
\begin{tikzpicture}[x=1pt,y=1pt]
\definecolor{fillColor}{RGB}{255,255,255}
\path[use as bounding box,fill=fillColor,fill opacity=0.00] (0,0) rectangle (199.47, 93.95);
\begin{scope}
\path[clip] (  0.00,  0.00) rectangle (199.47, 93.95);
\definecolor{drawColor}{RGB}{0,0,0}
\definecolor{fillColor}{RGB}{190,190,190}

\path[draw=drawColor,line width= 0.4pt,line join=round,line cap=round,fill=fillColor] ( 24.79, 21.72) rectangle ( 61.02, 65.83);

\path[draw=drawColor,line width= 0.4pt,line join=round,line cap=round,fill=fillColor] ( 68.27, 21.72) rectangle (104.51, 30.64);

\path[draw=drawColor,line width= 0.4pt,line join=round,line cap=round,fill=fillColor] (111.76, 21.72) rectangle (147.99, 30.66);

\path[draw=drawColor,line width= 0.4pt,line join=round,line cap=round,fill=fillColor] (155.24, 21.72) rectangle (191.48, 30.66);
\end{scope}
\begin{scope}
\path[clip] (  0.00,  0.00) rectangle (199.47, 93.95);
\definecolor{drawColor}{RGB}{0,0,0}

\node[text=drawColor,anchor=base,inner sep=0pt, outer sep=0pt, scale=  0.70] at ( 42.91,  8.52) {1};

\node[text=drawColor,anchor=base,inner sep=0pt, outer sep=0pt, scale=  0.70] at ( 86.39,  8.52) {2};

\node[text=drawColor,anchor=base,inner sep=0pt, outer sep=0pt, scale=  0.70] at (129.87,  8.52) {3};

\node[text=drawColor,anchor=base,inner sep=0pt, outer sep=0pt, scale=  0.70] at (173.36,  8.52) {4};
\end{scope}
\begin{scope}
\path[clip] (  0.00,  0.00) rectangle (199.47, 93.95);
\definecolor{drawColor}{RGB}{0,0,0}

\path[draw=drawColor,line width= 0.4pt,line join=round,line cap=round] ( 18.12, 21.72) -- ( 18.12, 92.63);

\path[draw=drawColor,line width= 0.4pt,line join=round,line cap=round] ( 18.12, 21.72) -- ( 12.12, 21.72);

\path[draw=drawColor,line width= 0.4pt,line join=round,line cap=round] ( 18.12, 35.90) -- ( 12.12, 35.90);

\path[draw=drawColor,line width= 0.4pt,line join=round,line cap=round] ( 18.12, 50.08) -- ( 12.12, 50.08);

\path[draw=drawColor,line width= 0.4pt,line join=round,line cap=round] ( 18.12, 64.27) -- ( 12.12, 64.27);

\path[draw=drawColor,line width= 0.4pt,line join=round,line cap=round] ( 18.12, 78.45) -- ( 12.12, 78.45);

\path[draw=drawColor,line width= 0.4pt,line join=round,line cap=round] ( 18.12, 92.63) -- ( 12.12, 92.63);

\node[text=drawColor,rotate= 90.00,anchor=base,inner sep=0pt, outer sep=0pt, scale=  0.70] at ( 12.12, 21.72) {0.0};

\node[text=drawColor,rotate= 90.00,anchor=base,inner sep=0pt, outer sep=0pt, scale=  0.70] at ( 12.12, 50.08) {0.4};

\node[text=drawColor,rotate= 90.00,anchor=base,inner sep=0pt, outer sep=0pt, scale=  0.70] at ( 12.12, 78.45) {0.8};
\end{scope}
\end{tikzpicture} \label{fig: MM post, I = 4}
\end{tikzpicture}
}
}

\subfloat[$I = 6$]{
\scalebox{1}[0.7]{
\begin{tikzpicture}
% Created by tikzDevice version 0.12.3.1 on 2023-07-25 23:38:57
% !TEX encoding = UTF-8 Unicode
\begin{tikzpicture}[x=1pt,y=1pt]
\definecolor{fillColor}{RGB}{255,255,255}
\path[use as bounding box,fill=fillColor,fill opacity=0.00] (0,0) rectangle (199.47, 93.95);
\begin{scope}
\path[clip] (  0.00,  0.00) rectangle (199.47, 93.95);
\definecolor{drawColor}{RGB}{0,0,0}
\definecolor{fillColor}{RGB}{190,190,190}

\path[draw=drawColor,line width= 0.4pt,line join=round,line cap=round,fill=fillColor] ( 24.79, 21.72) rectangle ( 48.60, 61.33);

\path[draw=drawColor,line width= 0.4pt,line join=round,line cap=round,fill=fillColor] ( 53.36, 21.72) rectangle ( 77.18, 28.01);

\path[draw=drawColor,line width= 0.4pt,line join=round,line cap=round,fill=fillColor] ( 81.94, 21.72) rectangle (105.75, 27.93);

\path[draw=drawColor,line width= 0.4pt,line join=round,line cap=round,fill=fillColor] (110.51, 21.72) rectangle (134.33, 28.03);

\path[draw=drawColor,line width= 0.4pt,line join=round,line cap=round,fill=fillColor] (139.09, 21.72) rectangle (162.90, 27.92);

\path[draw=drawColor,line width= 0.4pt,line join=round,line cap=round,fill=fillColor] (167.66, 21.72) rectangle (191.48, 28.01);
\end{scope}
\begin{scope}
\path[clip] (  0.00,  0.00) rectangle (199.47, 93.95);
\definecolor{drawColor}{RGB}{0,0,0}

\node[text=drawColor,anchor=base,inner sep=0pt, outer sep=0pt, scale=  0.70] at ( 36.69,  8.52) {1};

\node[text=drawColor,anchor=base,inner sep=0pt, outer sep=0pt, scale=  0.70] at ( 65.27,  8.52) {2};

\node[text=drawColor,anchor=base,inner sep=0pt, outer sep=0pt, scale=  0.70] at ( 93.84,  8.52) {3};

\node[text=drawColor,anchor=base,inner sep=0pt, outer sep=0pt, scale=  0.70] at (122.42,  8.52) {4};

\node[text=drawColor,anchor=base,inner sep=0pt, outer sep=0pt, scale=  0.70] at (151.00,  8.52) {5};

\node[text=drawColor,anchor=base,inner sep=0pt, outer sep=0pt, scale=  0.70] at (179.57,  8.52) {6};
\end{scope}
\begin{scope}
\path[clip] (  0.00,  0.00) rectangle (199.47, 93.95);
\definecolor{drawColor}{RGB}{0,0,0}

\path[draw=drawColor,line width= 0.4pt,line join=round,line cap=round] ( 18.12, 21.72) -- ( 18.12, 92.63);

\path[draw=drawColor,line width= 0.4pt,line join=round,line cap=round] ( 18.12, 21.72) -- ( 12.12, 21.72);

\path[draw=drawColor,line width= 0.4pt,line join=round,line cap=round] ( 18.12, 35.90) -- ( 12.12, 35.90);

\path[draw=drawColor,line width= 0.4pt,line join=round,line cap=round] ( 18.12, 50.08) -- ( 12.12, 50.08);

\path[draw=drawColor,line width= 0.4pt,line join=round,line cap=round] ( 18.12, 64.27) -- ( 12.12, 64.27);

\path[draw=drawColor,line width= 0.4pt,line join=round,line cap=round] ( 18.12, 78.45) -- ( 12.12, 78.45);

\path[draw=drawColor,line width= 0.4pt,line join=round,line cap=round] ( 18.12, 92.63) -- ( 12.12, 92.63);

\node[text=drawColor,rotate= 90.00,anchor=base,inner sep=0pt, outer sep=0pt, scale=  0.70] at ( 12.12, 21.72) {0.0};

\node[text=drawColor,rotate= 90.00,anchor=base,inner sep=0pt, outer sep=0pt, scale=  0.70] at ( 12.12, 50.08) {0.4};

\node[text=drawColor,rotate= 90.00,anchor=base,inner sep=0pt, outer sep=0pt, scale=  0.70] at ( 12.12, 78.45) {0.8};
\end{scope}
\end{tikzpicture} \label{fig: MM post, I = 6}
\end{tikzpicture}
}
}
\hfil
\subfloat[$I = 8$]{
\scalebox{1}[0.7]{
\begin{tikzpicture}
% Created by tikzDevice version 0.12.3.1 on 2023-07-25 23:45:41
% !TEX encoding = UTF-8 Unicode
\begin{tikzpicture}[x=1pt,y=1pt]
\definecolor{fillColor}{RGB}{255,255,255}
\path[use as bounding box,fill=fillColor,fill opacity=0.00] (0,0) rectangle (199.47, 93.95);
\begin{scope}
\path[clip] (  0.00,  0.00) rectangle (199.47, 93.95);
\definecolor{drawColor}{RGB}{0,0,0}
\definecolor{fillColor}{RGB}{190,190,190}

\path[draw=drawColor,line width= 0.4pt,line join=round,line cap=round,fill=fillColor] ( 24.79, 21.72) rectangle ( 42.52, 58.61);

\path[draw=drawColor,line width= 0.4pt,line join=round,line cap=round,fill=fillColor] ( 46.07, 21.72) rectangle ( 63.80, 26.57);

\path[draw=drawColor,line width= 0.4pt,line join=round,line cap=round,fill=fillColor] ( 67.35, 21.72) rectangle ( 85.08, 26.60);

\path[draw=drawColor,line width= 0.4pt,line join=round,line cap=round,fill=fillColor] ( 88.63, 21.72) rectangle (106.36, 26.60);

\path[draw=drawColor,line width= 0.4pt,line join=round,line cap=round,fill=fillColor] (109.91, 21.72) rectangle (127.64, 26.57);

\path[draw=drawColor,line width= 0.4pt,line join=round,line cap=round,fill=fillColor] (131.19, 21.72) rectangle (148.92, 26.57);

\path[draw=drawColor,line width= 0.4pt,line join=round,line cap=round,fill=fillColor] (152.47, 21.72) rectangle (170.20, 26.59);

\path[draw=drawColor,line width= 0.4pt,line join=round,line cap=round,fill=fillColor] (173.74, 21.72) rectangle (191.48, 26.56);
\end{scope}
\begin{scope}
\path[clip] (  0.00,  0.00) rectangle (199.47, 93.95);
\definecolor{drawColor}{RGB}{0,0,0}

\node[text=drawColor,anchor=base,inner sep=0pt, outer sep=0pt, scale=  0.70] at ( 33.65,  8.52) {1};

\node[text=drawColor,anchor=base,inner sep=0pt, outer sep=0pt, scale=  0.70] at ( 54.93,  8.52) {2};

\node[text=drawColor,anchor=base,inner sep=0pt, outer sep=0pt, scale=  0.70] at ( 76.21,  8.52) {3};

\node[text=drawColor,anchor=base,inner sep=0pt, outer sep=0pt, scale=  0.70] at ( 97.49,  8.52) {4};

\node[text=drawColor,anchor=base,inner sep=0pt, outer sep=0pt, scale=  0.70] at (118.77,  8.52) {5};

\node[text=drawColor,anchor=base,inner sep=0pt, outer sep=0pt, scale=  0.70] at (140.05,  8.52) {6};

\node[text=drawColor,anchor=base,inner sep=0pt, outer sep=0pt, scale=  0.70] at (161.33,  8.52) {7};

\node[text=drawColor,anchor=base,inner sep=0pt, outer sep=0pt, scale=  0.70] at (182.61,  8.52) {8};
\end{scope}
\begin{scope}
\path[clip] (  0.00,  0.00) rectangle (199.47, 93.95);
\definecolor{drawColor}{RGB}{0,0,0}

\path[draw=drawColor,line width= 0.4pt,line join=round,line cap=round] ( 18.12, 21.72) -- ( 18.12, 92.63);

\path[draw=drawColor,line width= 0.4pt,line join=round,line cap=round] ( 18.12, 21.72) -- ( 12.12, 21.72);

\path[draw=drawColor,line width= 0.4pt,line join=round,line cap=round] ( 18.12, 35.90) -- ( 12.12, 35.90);

\path[draw=drawColor,line width= 0.4pt,line join=round,line cap=round] ( 18.12, 50.08) -- ( 12.12, 50.08);

\path[draw=drawColor,line width= 0.4pt,line join=round,line cap=round] ( 18.12, 64.27) -- ( 12.12, 64.27);

\path[draw=drawColor,line width= 0.4pt,line join=round,line cap=round] ( 18.12, 78.45) -- ( 12.12, 78.45);

\path[draw=drawColor,line width= 0.4pt,line join=round,line cap=round] ( 18.12, 92.63) -- ( 12.12, 92.63);

\node[text=drawColor,rotate= 90.00,anchor=base,inner sep=0pt, outer sep=0pt, scale=  0.70] at ( 12.12, 21.72) {0.0};

\node[text=drawColor,rotate= 90.00,anchor=base,inner sep=0pt, outer sep=0pt, scale=  0.70] at ( 12.12, 50.08) {0.4};

\node[text=drawColor,rotate= 90.00,anchor=base,inner sep=0pt, outer sep=0pt, scale=  0.70] at ( 12.12, 78.45) {0.8};
\end{scope}
\end{tikzpicture} \label{fig: MM post, I = 8}
\end{tikzpicture}
}
}
\caption{
{\footnotesize
{\bf Sensitivity Analysis - Market Maker's Posterior (finite-$S$ illustration)} ($x$-axis: $\{ s_1, \cdots, s_I \}$, $y$-axis: probability).
For the finite-signal case $S=\{s_1,\dots,s_I\}$ with uniform prior, these graphs show the market maker's expected posterior probabilities over $S$, conditional on the insider observing signal $s_1$. 
} 
}
\label{fig: expected MM posterior}
\end{figure}

\iffalse

\begin{figure}[h!]

\centering 
%\tikzset{every picture/.style={scale=0.8}}%
\subfloat[{Information Efficiency $\mathbb{E}[ q_i^{(i)}]$}]{
%\includegraphics[width = 7cm, height = 2.8cm]{Expected_post_prob_large_range.pdf} \label{fig: expected post prob plot}
\scalebox{1}[0.85]{
\begin{tikzpicture}
\input{info_efficiency.tex} \label{fig: expected post prob plot}
\end{tikzpicture}
}
}
\iffalse
\hfil
\subfloat[Posterior Log-likelihood Ratio of Observed vs.~Other Signals]{
%\includegraphics[width = 7cm, height = 2.8cm]{Expected_post_log_LR_total_unobserved.pdf} \label{fig: expected post prob log LR total unobs}
\scalebox{1}[0.75]{
\begin{tikzpicture}
\input{post_likelihood_ratio.tex} \label{fig: expected post prob log LR total unobs}
\end{tikzpicture}
}
}
\fi
\caption{
{\footnotesize
{\bf Comparative Statics - Information Efficiency of Prices}\\
Figure~\ref{fig: expected post prob plot}: $x$-axis - number of signals $I$, $y$-axis - {$\mathbb{E}[ q_i^{(i)}]$} (information efficiency). %Figure~\ref{fig: expected post prob log LR total unobs}: $x$-axis - number of signals $I$, $y$-axis - $\log ( \frac{\mathbb{E}[q_i^{(i)}]}{(I-1) \mathbb{E}[q_{-i}^{(i)}]})$ (posterior log-likelihood ratio of observed signal $s_i$ vs.~all other signals $s_{-i}$).\\
} 
}
\label{fig: expected post prob}
\end{figure}
\fi

\begin{figure}[htbp!]

\centering 
\scalebox{1}[0.85]{
\begin{tikzpicture}
% Created by tikzDevice version 0.12.3.1 on 2023-07-25 23:53:34
% !TEX encoding = UTF-8 Unicode
\begin{tikzpicture}[x=1pt,y=1pt]
\definecolor{fillColor}{RGB}{255,255,255}
\path[use as bounding box,fill=fillColor,fill opacity=0.00] (0,0) rectangle (199.47, 93.95);
\begin{scope}
\path[clip] ( 18.12, 21.72) rectangle (198.15, 92.63);
\definecolor{drawColor}{RGB}{0,0,0}

\path[draw=drawColor,line width= 1.2pt,line join=round,line cap=round] ( 24.79, 90.00) --
	( 28.19, 59.18) --
	( 31.59, 50.29) --
	( 34.99, 45.40) --
	( 38.39, 42.32) --
	( 41.80, 40.06) --
	( 45.20, 38.31) --
	( 48.60, 37.11) --
	( 52.00, 36.02) --
	( 55.40, 34.97) --
	( 58.81, 34.00) --
	( 62.21, 33.28) --
	( 65.61, 32.70) --
	( 69.01, 32.05) --
	( 72.41, 31.63) --
	( 75.82, 31.08) --
	( 79.22, 30.61) --
	( 82.62, 30.17) --
	( 86.02, 29.78) --
	( 89.42, 29.52) --
	( 92.82, 29.40) --
	( 96.23, 28.95) --
	( 99.63, 28.64) --
	(103.03, 28.45) --
	(106.43, 28.07) --
	(109.83, 27.76) --
	(113.24, 27.88) --
	(116.64, 27.53) --
	(120.04, 27.30) --
	(123.44, 26.94) --
	(126.84, 26.86) --
	(130.24, 26.68) --
	(133.65, 26.63) --
	(137.05, 26.29) --
	(140.45, 26.10) --
	(143.85, 26.06) --
	(147.25, 26.05) --
	(150.66, 25.81) --
	(154.06, 25.73) --
	(157.46, 25.59) --
	(160.86, 25.40) --
	(164.26, 25.24) --
	(167.66, 25.28) --
	(171.07, 25.04) --
	(174.47, 24.94) --
	(177.87, 24.79) --
	(181.27, 24.73) --
	(184.67, 24.80) --
	(188.08, 24.63) --
	(191.48, 24.35);
\end{scope}
\begin{scope}
\path[clip] (  0.00,  0.00) rectangle (199.47, 93.95);
\definecolor{drawColor}{RGB}{0,0,0}

\path[draw=drawColor,line width= 0.4pt,line join=round,line cap=round] ( 18.12, 21.72) --
	(198.15, 21.72) --
	(198.15, 92.63) --
	( 18.12, 92.63) --
	( 18.12, 21.72);
\end{scope}
\begin{scope}
\path[clip] (  0.00,  0.00) rectangle (199.47, 93.95);
\definecolor{drawColor}{RGB}{0,0,0}

\path[draw=drawColor,line width= 0.4pt,line join=round,line cap=round] ( 24.79, 21.72) -- (193.18, 21.72);

\path[draw=drawColor,line width= 0.4pt,line join=round,line cap=round] ( 24.79, 21.72) -- ( 24.79, 18.17);

\path[draw=drawColor,line width= 0.4pt,line join=round,line cap=round] ( 57.11, 21.72) -- ( 57.11, 18.17);

\path[draw=drawColor,line width= 0.4pt,line join=round,line cap=round] ( 91.12, 21.72) -- ( 91.12, 18.17);

\path[draw=drawColor,line width= 0.4pt,line join=round,line cap=round] (125.14, 21.72) -- (125.14, 18.17);

\path[draw=drawColor,line width= 0.4pt,line join=round,line cap=round] (159.16, 21.72) -- (159.16, 18.17);

\path[draw=drawColor,line width= 0.4pt,line join=round,line cap=round] (193.18, 21.72) -- (193.18, 18.17);

\node[text=drawColor,anchor=base,inner sep=0pt, outer sep=0pt, scale=  0.70] at ( 24.79,  8.52) {1};

\node[text=drawColor,anchor=base,inner sep=0pt, outer sep=0pt, scale=  0.70] at ( 57.11,  8.52) {20};

\node[text=drawColor,anchor=base,inner sep=0pt, outer sep=0pt, scale=  0.70] at ( 91.12,  8.52) {40};

\node[text=drawColor,anchor=base,inner sep=0pt, outer sep=0pt, scale=  0.70] at (125.14,  8.52) {60};

\node[text=drawColor,anchor=base,inner sep=0pt, outer sep=0pt, scale=  0.70] at (159.16,  8.52) {80};

\node[text=drawColor,anchor=base,inner sep=0pt, outer sep=0pt, scale=  0.70] at (193.18,  8.52) {100};

\path[draw=drawColor,line width= 0.4pt,line join=round,line cap=round] ( 18.12, 21.72) -- ( 18.12, 90.00);

\path[draw=drawColor,line width= 0.4pt,line join=round,line cap=round] ( 18.12, 32.25) -- ( 14.57, 32.25);

\path[draw=drawColor,line width= 0.4pt,line join=round,line cap=round] ( 18.12, 51.50) -- ( 14.57, 51.50);

\path[draw=drawColor,line width= 0.4pt,line join=round,line cap=round] ( 18.12, 70.75) -- ( 14.57, 70.75);

\path[draw=drawColor,line width= 0.4pt,line join=round,line cap=round] ( 18.12, 90.00) -- ( 14.57, 90.00);

\node[text=drawColor,anchor=base east,inner sep=0pt, outer sep=0pt, scale=  0.70] at ( 14.52, 29.84) {0.4};

\node[text=drawColor,anchor=base east,inner sep=0pt, outer sep=0pt, scale=  0.70] at ( 14.52, 49.09) {0.6};

\node[text=drawColor,anchor=base east,inner sep=0pt, outer sep=0pt, scale=  0.70] at ( 14.52, 68.34) {0.8};

\node[text=drawColor,anchor=base east,inner sep=0pt, outer sep=0pt, scale=  0.70] at ( 14.52, 87.59) {1};
\end{scope}
\end{tikzpicture} \label{fig: expected post prob plot}
\end{tikzpicture}
}

\caption{
{\footnotesize
\textbf{Sensitivity Analysis: Information Efficiency (finite-$S$ illustration).}
In the finite-signal case $S=\{s_1,\dots,s_I\}$ with uniform prior, let
$q^{(i)}=(q^{(i)}_1,\dots,q^{(i)}_I)$ denote the market maker's equilibrium posterior weights conditional on
$S=s_i$, so that $\mathrm{IE}(s_i)=\mathbb{E}[q_i^{(i)}]$.
The figure plots $\mathrm{IE}(s_i)$ as a function of $I$ (by symmetry, independent of $i$); the expectation is taken
under the equilibrium law of the posterior induced by the symmetric equilibrium with intensity
$\alpha^*=\alpha^*(I)$.
}}
\label{fig: expected post prob}
\end{figure}

\subsection{Information Efficiency of Prices}
\label{sec: AD price discovery}

The \emph{information efficiency} of market prices is a foundational question in financial economics (\cite{fama1970efficient,grossman1980impossibility}).
We now address this question in our Arrow--Debreu setting---where the question concerns the full state-contingent price schedule (pricing kernel)---by introducing a formal definition and
deriving an equilibrium characterization.

As a first step, we record the equilibrium probability law of the market maker's posterior over order flow $\omega$ by specializing Lemma~\ref{lemma: claim for MM posterior general S} to $\alpha=\alpha^*$.

\begin{proposition}
\label{prop: MM posterior general S}
Conditional on the insider observing $s\in S$, the market maker's equilibrium posterior
$\hat\pi_1^*(\,\cdot\,)\equiv \hat\pi_1(\,\cdot\,\mid \hat\omega;\Theta^*)$ (evaluated at
$\hat\omega=\theta^*(s)+\widehat X$) is a logistic-normal random probability measure:
it is absolutely continuous with respect to $\mu$, with (random) Radon--Nikodym density given by
\begin{equation}
\label{eqn: equilibrium logistic-normal posterior}
p_s(u)
=
\frac{\exp\!\left(
Z(u)+(\alpha^*)^2\left\langle Qk(\cdot,u),Qk(\cdot,s)\right\rangle_{\mathcal H_k}
-\frac{(\alpha^*)^2}{2}\left\|Qk(\cdot,u)\right\|_{\mathcal H_k}^2
\right)}
{\int_S \exp\!\left(
Z(v)+(\alpha^*)^2\left\langle Qk(\cdot,v),Qk(\cdot,s)\right\rangle_{\mathcal H_k}
-\frac{(\alpha^*)^2}{2}\left\|Qk(\cdot,v)\right\|_{\mathcal H_k}^2
\right)\mu(dv)},
\end{equation}
where $Z(\cdot)$ is the centered Gaussian process $Z(u):=\alpha^*\,\widehat X\!\bigl(Qk(\cdot,u)\bigr)$, so that
\begin{equation}\label{eqn: equilibrium Z covariance}
\Cov\!\bigl(Z(u),Z(v)\bigr)
=
(\alpha^*)^2\left\langle Qk(\cdot,u),Qk(\cdot,v)\right\rangle_{\mathcal H_k}.
\end{equation}
\end{proposition}

%The equilibrium distribution of AD prices is an immediate corollary. We can similarly characterize the equilibrium $\omega$-by-$\omega$ pricing kernel.

\begin{corollary}[Equilibrium Arrow--Debreu Pricing Kernel]
\label{cor: equil AD prices}
Conditional on the insider observing $s \in S$, and with the market maker's equilibrium posterior
$\widehat\pi_1^*(du)=p_s(u)\,\mu(du)$ as in Proposition~\ref{prop: MM posterior general S}, the equilibrium
Arrow--Debreu pricing kernel satisfies, almost surely,
\begin{equation}
\label{eqn: equil AD price kernel general S}
P^*(\,\cdot\,,\omega\kern0.045em;s)
\;=\;
\int_S \eta(\,\cdot\,,u)\,p_s(u)\,\mu(du),
\end{equation}
so that $P^*(\,\cdot\,,\omega\kern0.045em;s)$ is a random $\mu$--mixture of the payoff densities
$\{\eta(\,\cdot\,,u)\}_{u\in S}$ with mixing measure $\widehat\pi_1^*$.  The expected equilibrium
Arrow--Debreu prices are
\begin{equation}
\label{eqn: equil expected AD price}
\overline P^*(\,\cdot\,;s)
\;:=\;
\mathbb{E}\!\left[P^*(\,\cdot\,,\omega\kern0.045em;s)\,\big|\,S=s\right]
\;=\;
\int_S \mathbb{E}\!\left[p_s(u)\,\big|\,S=s\right]\eta(\,\cdot\,,u)\,\mu(du).
\end{equation}
\end{corollary}

\begin{proof}
By the market maker's competitive
(zero-profit) pricing rule, for each Arrow--Debreu security $x$ the equilibrium price equals the posterior mean
of the corresponding payoff density:
\[
P^*(x,\omega\kern0.045em;s)
=\int_S \eta(x,u)\,\widehat\pi_1^*(du)
\qquad\text{for $\mathbb{P}(\,\cdot\,|S=s)$-a.e.\ $\omega$.}
\]
Using $\widehat\pi_1^*(du)=p_s(u)\,\mu(du)$ yields \eqref{eqn: equil AD price kernel general S}.  Taking
conditional expectations given $S=s$ and applying Fubini--Tonelli gives \eqref{eqn: equil expected AD price}.
\end{proof}

\paragraph{Information Efficiency}
By Corollary~\ref{cor: equil AD prices}, the conditional mean pricing kernel admits the mixture representation
\[
\overline P^*(\,\cdot\,;s)
\;=\;
\int_S \eta(\,\cdot\,,u)\,\overline\pi_s^*(du),
\]
where
\[
\overline\pi_s^*(du)
\;:=\;
\mathbb{E}\!\left[\widehat\pi_1^*(du\,;\,s)\,\big|\,S=s\right]
\;=\;
\mathbb{E}\!\left[p_s(u)\,\big|\,S=s\right]\mu(du).
\]
The \textbf{information efficiency} of Arrow--Debreu prices conditional on signal $s$ is then defined as the weight that
$\overline\pi_s^*$ assigns to the true payoff distribution $\eta(\,\cdot\,,s)$:
\[
\mathrm{IE}(s)
:=
\begin{cases}
\overline\pi_s^*(\{s\}), & \text{if $\overline\pi_s^*$ has an atom at $s$,}\\[6pt]
\dfrac{d\overline\pi_s^*}{d\mu}(s)
=\mathbb{E}\!\left[p_s(s)\,\big|\,S=s\right], & \text{otherwise.}
\end{cases}
\]

Economically, $\mathrm{IE}(s)$ is the market maker's \emph{expected posterior weight} on the correct
payoff distribution $\eta(\,\cdot\,,s)$ in $\overline P^*(\,\cdot\,;s)$.  It therefore measures how much of the
insider's information is incorporated into Arrow--Debreu prices.
We consider the finite, uniform case as an illustration.

\paragraph{Finite-Uniform Case}
Suppose $S=\{s_1,\dots,s_I\}$ and $\mu$ is uniform.  Conditional on $S=s_i$, the equilibrium posterior
\eqref{eqn: equilibrium logistic-normal posterior} reduces to a logistic-normal (softmax) probability vector
$q^{(i)}=(q^{(i)}_1,\dots,q^{(i)}_I)\in\Delta^{I-1}$ of the form
\begin{equation}
\label{eqn: canonical posterior}
q^{(i)}_j
\;\propto\;
\exp \bigl(Z_j+(\alpha^*)^2\mathbf{1}_{\{j=i\}}\bigr),
\qquad
(Z_1,\dots,Z_I)^T \stackrel{d}{=} \mathcal{N}\!\bigl(0,(\alpha^*)^2{\bf Q}\bigr),
\end{equation}
so that $\sum_{j=1}^I q^{(i)}_j=1$.
The expected equilibrium Arrow--Debreu prices are
\begin{equation}
\label{eqn: equil expected AD price finite S}
\overline P^*(\,\cdot\,;s_i)
=\sum_{j=1}^I \mathbb{E}\!\left[q^{(i)}_j\right]\eta(\,\cdot\,,s_j).
\end{equation}
Thus, the information efficiency index conditional on $s_i$ is exactly the expected posterior weight on the
true payoff density:
\[
\mathrm{IE}(s_i)=\mathbb{E}\!\left[q^{(i)}_i\right].
\]
Figure~\ref{fig: expected MM posterior} plots the conditional mean posterior weights
$\mathbb{E}[q^{(1)}_j]$, $1\le j\le I$, conditional on $s_1$, as the number of signals $I$ increases.  By symmetry,
$\mathrm{IE}(s_1)=\mathbb{E}[q^{(1)}_1]=\mathbb{E}[q^{(i)}_i]$, and the figure shows that information
efficiency declines as $I$ grows.  Figure~\ref{fig: expected post prob} further shows that
$\mathbb{E}[q^{(i)}_i]$ is decreasing and convex in $I$.  Economically, enlarging the signal set makes the
prior more dispersed, so order flow is less diagnostic about the true signal, reducing the
posterior mass assigned to the truth.  At the same time, equilibrium adjusts through the endogenous
intensity $\alpha^*$, which increases the information content of order flow and partially offsets the
loss of identification.

\paragraph{Invariance}
Because the canonical game determines the posterior law, the induced Arrow--Debreu prices, and hence $\mathrm{IE}(s)$,
without reference to $\eta(\,\cdot\,,\cdot)$ or $\sigma(\cdot)$, the following corollary is immediate.

\begin{corollary}
\label{cor: invariance of info efficiency}
For every $s\in S$, $\mathrm{IE}(s)$ is independent of the payoff family $\eta(\,\cdot\,,\cdot)$ and of the
noise-trading intensity $\sigma(\cdot)$.
\end{corollary}

Corollary~\ref{cor: invariance of info efficiency} is the price-discovery counterpart of the standard
complete-markets invariance result in asset pricing.  In a complete Arrow--Debreu economy, equilibrium
risk sharing is implemented by a state-price kernel: once state prices are pinned down, the value of
any contingent claim is obtained by the same linear functional, and thus depends on payoffs only
through that kernel---not on how particular payoff risks are distributed across agents or repackaged
across traded securities; see, e.g., \cite{bodie2021investments}. Likewise here, the \emph{aggregate}
informativeness of the entire state-price schedule is governed by the equilibrium posterior mixing
kernel generated by the canonical game, and hence depends only on the \emph{aggregate} dimension of private
uncertainty---not on which informed types possess which particular payoff-relevant details.

Empirically, in liquid options markets, informed traders can shift expression across contracts and
synthetics, relocating observed price impact across strikes and tenors without materially changing
the information content of the \emph{entire} option surface.  Accordingly, when surface-level
informativeness \emph{does} move as flow rotates across contracts, the frictionless benchmark of
Corollary~\ref{cor: invariance of info efficiency} suggests the presence of impediments to
cross-contract substitution (e.g., margin or position limits, segmentation, dealer balance-sheet
constraints).

\section{Conclusion}
\label{sec: conclusion}

Unifying the complete-markets viewpoint of \cite{arrow1954existence} with the strategic trading mechanism of
\cite{kyle1985continuous}, we have developed a theory of informed trading across a continuum of contingent-claim
markets. The model accommodates private information about arbitrary aspects of the payoff distribution.
Despite this generality, the model remains parsimonious: a single endogenous scale parameter delivers closed-form
characterizations of price discovery in the infinite-dimensional trading environment.

Our framework systematically characterizes the economic channel through which information is transmitted across
assets---a result that lies beyond single-asset, or more generally finite-dimensional, formulations.
Equilibrium informed demand takes the form of familiar option-trading strategies, thereby bridging market microstructure
theory and empirically observed trading practices.
Our characterization of cross-market price impact maps order flow in one claim into price responses across others and
produces tractable, testable restrictions for derivatives markets.

We hope this framework may serve as a foundation for future work in informed-trading theory.
Natural next steps include dynamic extensions, risk aversion and inventory frictions, imperfect
competition and multiple informed traders, incomplete menus of traded claims, and econometric identification and
estimation of cross-impact objects from high-frequency order flow. More generally, the present analysis suggests
that a general theory of price discovery for contingent claims should be organized around (i) what information is
economically relevant and (ii) how market design and the traded payoff span determine the way that information is
encoded and revealed in prices.

%\clearpage
%\newpage
%\interlinepenalty=10000
\bibliographystyle{chicago}
\bibliography{main}
%\clearpage

%\fi

\end{document}